\documentclass[dvipsnames]{lmcs}
\pdfoutput=1

\usepackage{lastpage}
\lmcsdoi{22}{1}{26}
\lmcsheading{}{\pageref{LastPage}}{}{}%
{Jul.~15,~2025}{Mar.~16,~2026}{}

\usepackage[utf8]{inputenc}

\usepackage{amsmath}
\usepackage{amssymb}

\usepackage{xspace}

\usepackage{colortbl}

\usepackage{booktabs}

\usepackage{mathtools}

\usepackage{hyperref}
\usepackage{tikz}
\usetikzlibrary{arrows,decorations.text,decorations.pathmorphing,decorations.pathreplacing,positioning,patterns,tikzmark,arrows.meta,backgrounds,shapes.geometric}
\tikzset{  
  >=stealth',
  plainnode/.style 	= {draw, thick,circle, minimum size=8mm, inner sep=0mm}
}

\newcommand{\nats}{\mathbb{N}}
\renewcommand{\epsilon}{\varepsilon}
\renewcommand{\phi}{\varphi}
\newcommand{\size}[1]{|#1|}
\newcommand{\pow}[1]{2^{#1}}
\newcommand{\cceq}{\mathop{::=}}
\newcommand{\set}[1]{\{#1\}}

\newcommand{\F}{\mathop{\mathbf{F}\vphantom{a}}\nolimits}
\newcommand{\G}{\mathop{\mathbf{G}\vphantom{a}}\nolimits}
\DeclareMathOperator{\U}{\mathbf{U}}
\newcommand{\X}{\mathop{\mathbf{X}\vphantom{a}}\nolimits}

\newcommand{\ltl}{\mathrm{LTL}\xspace}

\newcommand{\ctlstar}{\mathrm{CTL^*}\xspace}
\newcommand{\hyltl}{\mathrm{Hyper\-LTL}\xspace}
\newcommand{\hyqptl}{\mathrm{Hyper\-QPTL}\xspace}
\newcommand{\hyqptlplus}{\mathrm{Hyper\-QPTL^+}\xspace}
\newcommand{\sohyltl}{\mathrm{Hyper^2LTL}\xspace}
\newcommand{\sohyltlfp}{\mathrm{Hyper^2LTL_{\mathrm{mm}}}\xspace}
\newcommand{\sohyltlfpmax}{\mathrm{Hyper^2LTL_{\mathrm{mm}}^{\largest}}\xspace}
\newcommand{\sohyltlfpmin}{\mathrm{Hyper^2LTL_{\mathrm{mm}}^{\smallest}}\xspace}
\newcommand{\sohyltlfpsmalar}{\mathrm{Hyper^2LTL_{\mathrm{mm}}^{\smalar}}\xspace}
\newcommand{\lfpsohyltlfp}{\mathrm{lfp\text{-}{Hyper^2LTL_{\mathrm{mm}}}}\xspace}

\newcommand{\hyctlstar}{\mathrm{HyperCTL^*}\xspace}
\newcommand{\sohyltlfpold}{\mathrm{Hyper^2LTL_{\mathrm{fp}}}\xspace}
\newcommand{\teamltl}{\mathrm{TeamLTL}\xspace}

\newcommand{\suffix}[2]{#1[#2,\infty)}

\newcommand{\fovar}{\mathcal{V}_1}
\newcommand{\sovar}{\mathcal{V}_2}
\newcommand{\ap}[0]{\mathrm{AP}}

\newcommand{\univar}{X_a}
\newcommand{\unidisvar}{X_d}

\newcommand{\smallest}{\curlyvee}
\newcommand{\largest}{\curlywedge}
\newcommand{\smalar}{{\mathrlap{\curlyvee}\curlywedge}}
\newcommand{\solutions}{\mathrm{sol}}

\newcommand{\pspace}{\textsc{PSpace}\xspace}

\newcommand{\tower}{\textsc{Tower}\xspace}

\newcommand{\myquot}[1]{``#1''}

\newcommand{\tsys}{\mathfrak{T}}
\newcommand{\traces}{\mathrm{Tr}}
\newcommand{\initmark}{I}

\newcommand{\posprop}{\mathtt{+}}
\newcommand{\negprop}{\mathtt{-}}
\newcommand{\setprop}{\mathtt{s}}
\newcommand{\inprop}{\mathtt{x}}

\newcommand{\contcard}{\mathfrak{c}}

\newcommand{\proposition}{\texttt{p}}
\newcommand{\argone}{\texttt{arg1}}
\newcommand{\argtwo}{\texttt{arg2}}
\newcommand{\res}{\texttt{res}}
\newcommand{\add}{\texttt{add}}
\newcommand{\mult}{\texttt{mult}}

\newcommand{\plustimes}{{(+,\cdot)}}
\newcommand{\hyperize}{{\mathit{hyp}}}
\newcommand{\arithmetize}{{\mathit{ar}}}
\newcommand{\allsets}{{\mathit{allSets}}}
\newcommand{\pair}{\mathit{pair}}
\newcommand{\istrace}{\mathit{isTrace}}

\newcommand{\arith}{\mathit{arith}}
\newcommand{\onlytraces}{\mathit{onlyTraces}}
\newcommand{\alltraces}{\mathit{allTraces}}
\newcommand{\parti}{\mathit{part}}
\newcommand{\update}{\mathit{update}}

\newcommand{\calY}{\mathcal{Y}}
\newcommand{\ists}{\mathit{isTS}}
\newcommand{\ispath}{\mathit{isPath}}
\newcommand{\traceof}{\mathit{traceOf}}

\newcommand{\prefs}{\mathit{Prefs}}

\newcommand{\natsstruct}{(\nats, +, \cdot, <, \in)}
\newcommand{\class}{\mathcal{C}}

\newcommand{\con}{\mathrm{con}}
\newcommand{\step}{\mathrm{step}}
\newcommand{\tracein}{\mathrel{\triangleright}}
\newcommand{\quant}{Q}
\newcommand{\overdot}[1]{\dot{#1}}
\newcommand{\overdotdot}[1]{\ddot{#1}}
\newcommand{\cw}{\mathrm{cw}}
\newcommand{\expansion}{e}

\newcommand{\lfp}{\mathrm{lfp}}

\newcommand{\skolem}{\mathit{Sk}}

\newcommand{\equals}[3]{#1 =_{#3} #2}

\newcommand{\tree}{B}

\newcommand{\ck}{\mathit{ck}}
\newcommand{\gni}{\mathit{gni}}

\newcommand{\itereval}{e}

\newcommand{\marker}{\texttt{m}}
\newcommand{\complete}{\mathit{complete}}
\newcommand{\merge}{^\smallfrown}
\newcommand{\guard}{\mathit{guard}}
\newcommand{\enc}{\mathit{enc}}
\newcommand{\extend}{\mathit{ext}}

\newcommand{\hastree}{\mathit{hasTree}}
\newcommand{\encode}[1]{\langle #1\rangle}

\newcommand{\vertices}{\mathit{V}}
\newcommand{\edges}{\mathit{E}}
\newcommand{\edge}{\mathit{edge}}
\newcommand{\prefix}{\mathit{prefix}}
\newcommand{\prefixes}{\mathit{prefixes}}
\newcommand{\rel}{\mathit{rel}}
\newcommand{\init}{\mathit{init}}
\newcommand{\consistent}{\mathit{cons}}

\usepackage{wrapfig}

\begin{document}

\title{The Complexity of Second-order HyperLTL}
\titlecomment{A preliminary version of this article has been published in the proceedings of CSL 2025~\cite{sohypercomplexity}. Here, all problems left open in the conference version have been solved, as reported in the technical report~\cite{fragments}.}

\author[H.~Frenkel]{Hadar Frenkel\lmcsorcid{0000-0002-3566-0338}}[a]
\author[G.~Regaud]{Gaëtan Regaud\lmcsorcid{0009-0000-1409-5707}}[b]
\author[M.~Zimmermann]{Martin Zimmermann\lmcsorcid{0000-0002-8038-2453}}[c]

\address{Bar-Ilan University, Ramat Gan, Israel}	
\email{hadar.frenkel@biu.ac.il}  

\address{ENS Rennes, Rennes, France}	
\email{gaetan.regaud@ens-rennes.fr}  

\address{Aalborg University, Aalborg, Denmark}	
\email{mzi@cs.aau.dk}  





\begin{abstract}
  \noindent We determine the complexity of second-order HyperLTL satisfiability, finite-state satisfiability, and model-checking: All three are equivalent to truth in third-order arithmetic. 
    
    We also consider two fragments of second-order HyperLTL that have been introduced with the aim to facilitate effective model-checking by restricting the sets one can quantify over. 
The first one restricts second-order quantification to smallest/largest sets that satisfy a guard while the second one restricts second-order quantification further to least fixed points of (first-order) HyperLTL definable functions.
    All three problems for the first fragment are still equivalent to truth in third-order arithmetic while satisfiability for the second fragment is $\Sigma_1^2$-complete, and finite-state satisfiability and model-checking are equivalent to truth in second-order arithmetic. 

    Finally, we also introduce closed-world semantics for second-order HyperLTL, where set quantification ranges only over subsets of the model, while set quantification in standard semantics ranges over arbitrary sets of traces. 
    Here, satisfiability for the least fixed point fragment becomes $\Sigma_1^1$-complete, but all other results are unaffected.
\end{abstract}

\maketitle

\section{Introduction}
\label{sec:intro}

The introduction of hyperlogics~\cite{ClarksonFKMRS14} for the specification and verification of hyperproperties~\cite{ClarksonS10} -- properties that relate multiple system executions -- has been one of the major success stories of formal verification during the last decade.
Logics like $\hyltl$ and $\hyctlstar$~\cite{ClarksonFKMRS14}, the extensions of $\ltl$~\cite{Pnueli77} 
and $\ctlstar$~\cite{EmersonH86} (respectively) with trace quantification, are natural specification languages for information-flow and security properties, have a decidable model-checking problem~\cite{FinkbeinerRS15}, and hence found many applications in program verification.

However, while expressive enough to express common information-flow properties, they are unable to express other important hyperproperties, e.g., common knowledge in multi-agent systems and asynchronous hyperproperties (witnessed by a plethora of asynchronous extensions of $\hyltl$, e.g.,~\cite{BartocciHNC23,DBLP:conf/cav/BaumeisterCBFS21,BeutnerF23,DBLP:journals/corr/abs-2404-16778,DBLP:conf/lics/BozzelliPS21,DBLP:conf/concur/BozzelliPS22,DBLP:journals/pacmpl/GutsfeldMO21,DBLP:conf/mfcs/KontinenSV23,KontinenSV24,DBLP:conf/mfcs/KrebsMV018}).
These examples all have in common that they are \emph{second-order} properties, i.e., they naturally require quantification over \emph{sets} of traces, while $\hyltl$ (and $\hyctlstar$) only allows quantification over traces.

In light of this situation, Beutner et al.~\cite{DBLP:conf/cav/BeutnerFFM23} introduced the logic $\sohyltl$, which extends $\hyltl$ with second-order quantification, i.e., quantification over sets of traces. 
They show that the resulting logic, $\sohyltl$, is indeed able to capture common knowledge, asynchronous extensions of $\hyltl$, and many other applications. 

Consider, e.g., common knowledge in multi-agent systems where each agent~$i$ only observes some parts of the system. The agent \emph{knows} that a statement~$\varphi$ holds if it holds on all traces that are \emph{indistinguishable} in the agent's view. 
We write $\pi \sim_i \pi'$ if the traces~$\pi$ and $\pi'$ are indistinguishable for agent $i$. 
A property~$\varphi$ is common knowledge among all agents if all agents know $\varphi$, all agents know that all agents know $\varphi$, and so on, i.e., one takes the infinite closure of knowledge among all agents. 
This infinite closure cannot be expressed using first-order quantification over traces~\cite{DBLP:conf/fossacs/BozzelliMP15}, like the one used in $\hyltl$. 
The second-order quantification suggested by Beutner et al.\ allows us to express common knowledge, as demonstrated by the formula~$\varphi_\ck$, which 
 states that $\varphi$ is common knowledge on all traces of the system (we use a simplified syntax for readability): 
\begin{align*}
   \varphi_\ck =  \forall \pi.\, \exists X .\ \pi \in X \land \Big( \forall \pi' \in X.\, \forall \pi''.\ \big(\bigvee\nolimits_{i=1}^n \pi' \sim_i \pi'' \big) \rightarrow \pi'' \in X \Big)\, \land\, \forall \pi' \in X.\ \varphi({\pi'})
\end{align*}
The formula~$\varphi_\ck$ expresses that for every trace $t$ (instantiating $\pi$), there exists a set $T$ (an instantiation of the second-order variable $X$) such that $t$ is in $T$, $T$ is closed under the observations of all agents (if $t'$ is in $T$ and $t''$ is indistinguishable from $t'$ for some agent~$i$, then $t''$ is also in $T$), and all traces in $T$ satisfy $\varphi$.  

However, Beutner et al.\ also note that this expressiveness comes at a steep price:
model-checking $\sohyltl$ is highly undecidable, i.e., $\Sigma_1^1$-hard.
Thus, their main result is a partial model-checking algorithm for a fragment of $\sohyltl$ where second-order quantification degenerates to least fixed point computations of $\hyltl$ definable functions. 
A prototype implementation of the algorithm is able to model-check properties capturing common knowledge, asynchronous hyperproperties, and distributed computing. 

However, one question has been left open: Just how complex is $\sohyltl$ verification?

\subsection{Complexity Classes for Undecidable Problems}
The complexity of undecidable problems is typically captured in terms of the arithmetical and analytical hierarchy, where decision problems (encoded as subsets of $\nats$) are classified based on their definability by formulas of higher-order arithmetic, namely by the type of objects one can quantify over and by the number of alternations of such quantifiers.
We refer to Roger's textbook~\cite{Rogers87} for fully formal definitions and refer to Figure~\ref{fig_hierarchies} for a visualization.

\begin{figure}[h]
    \centering
    
  \scalebox{.95}{
  \begin{tikzpicture}[xscale=1.0,yscale=.8,thick]

    \fill[fill = gray!25, rounded corners] (-1.05,2) rectangle (0.5,-1.5);
    \fill[fill = gray!25, rounded corners] (.6,2) rectangle (14.5,-1.5);

    \node[align=center] (s00) at (0,0) {$\Sigma^0_0$ \\ $=$ \\ $\Pi^0_0$} ;
    \node (s01) at (1,1) {$\Sigma^0_1$} ;
    \node (p01) at (1,-1) {$\Pi^0_1$} ;
    \node (s02) at (2,1) {$\Sigma^0_2$} ;
    \node (p02) at (2,-1) {$\Pi^0_2$} ;
    \node (s03) at (3,1) {$\Sigma^0_3$} ;
    \node (p03) at (3,-1) {$\Pi^0_3$} ;
    \node (s04) at (4,1) {$\cdots$} ;
    \node (p04) at (4,-1) {$\cdots$} ;
    
    \node[align=center] (s10) at (5,0) {$\Sigma^1_0 $ \\ $=$ \\ $ \Pi^1_0$} ;
    \node (s11) at (6,1) {$\Sigma^1_1$} ;
    \node (p11) at (6,-1) {$\Pi^1_1$} ;
    \node (s12) at (7,1) {$\Sigma^1_2$} ;
    \node (p12) at (7,-1) {$\Pi^1_2$} ;
    \node (s13) at (8,1) {$\Sigma^1_3$} ;
    \node (p13) at (8,-1) {$\Pi^1_3$} ;
    \node (s14) at (9,1) {$\cdots$} ;
    \node (p14) at (9,-1) {$\cdots$} ;
    
    \node[align=center] (s20) at (10,0) {$\Sigma^2_0$ \\ $ =$ \\ $ \Pi^2_0$} ;
    \node (s21) at (11,1) {$\Sigma^2_1$} ;
    \node (p21) at (11,-1) {$\Pi^2_1$} ;
    \node (s22) at (12,1) {$\Sigma^2_2$} ;
    \node (p22) at (12,-1) {$\Pi^2_2$} ;
    \node (s23) at (13,1) {$\Sigma^2_3$} ;
    \node (p23) at (13,-1) {$\Pi^2_3$} ;
    \node (s24) at (14,1) {$\cdots$} ;
    \node (p24) at (14,-1) {$\cdots$} ;
    
    \foreach \i in {0,1,2} {
      \draw (s\i0) -- (s\i1) ;
      \draw (s\i0) -- (p\i1) ;
      \draw (s\i1) -- (s\i2) ;
      \draw (s\i1) -- (p\i2) ;
      \draw (p\i1) -- (s\i2) ;
      \draw (p\i1) -- (p\i2) ;
      \draw (s\i2) -- (s\i3) ;
      \draw (s\i2) -- (p\i3) ;
      \draw (p\i2) -- (s\i3) ;
      \draw (p\i2) -- (p\i3) ;
      \draw (s\i3) -- (s\i4) ;
      \draw (s\i3) -- (p\i4) ;
      \draw (p\i3) -- (s\i4) ;
      \draw (p\i3) -- (p\i4) ;
    }
    \foreach \i [evaluate=\i as \iplus using int(\i+1)] in {0,1} {
      \draw (s\i4) -- (s\iplus0) ;
      \draw (p\i4) -- (s\iplus0) ;
}
      \node[] at (-0.25,1.7) {\small Decidable} ;
      \node at (13.4,1.7) {\small Undecidable} ;


      \node[align=left,font=\small,
      rounded corners=2pt] at (3.2,1.7)
      (re) {\small Recursively enumerable} ;
      \draw[-stealth,rounded corners=2pt] (s01) |- (re) ;

    \path (4.25,-1.75) edge[decorate,decoration={brace,amplitude=3pt}]
    node[below] {\begin{minipage}{4cm}\centering
    	arithmetical hierarchy $ $\\$\equiv$\\ first-order arithmetic 
    \end{minipage}} (.75,-1.75) ;

    \path (9.25,-1.75) edge[decorate,decoration={brace,amplitude=3pt}]
    node[below] {\begin{minipage}{4cm}\centering
    	analytical hierarchy $ $\\$\equiv$\\ second-order arithmetic 
    \end{minipage}} (5,-1.75) ;

    \path (14.25,-1.75) edge[decorate,decoration={brace,amplitude=3pt}]
    node[below] {\begin{minipage}{4cm}\centering
    	\myquot{the third hierarchy} $ $\\ {$\equiv$}\\ third-order arithmetic 
    \end{minipage}} (10,-1.75) ;

  \end{tikzpicture}}
    
    \caption{The arithmetical hierarchy, the analytical hierarchy, and beyond.}
    \label{fig_hierarchies}
\end{figure}

The class~$\Sigma_1^0$ contains the sets of natural numbers of the form
\[
\set{x \in \nats \mid \exists x_0.\  \cdots \exists x_k.\ \psi(x, x_0, \ldots, x_k)}
\] 
where quantifiers range over natural numbers and $\psi$ is a quantifier-free arithmetic formula. 
Note that this is exactly the class of recursively enumerable sets.
The notation~$\Sigma_1^0$ signifies that there is a single block of existential quantifiers (the subscript~$1$) ranging over natural numbers (type~$0$ objects, explaining the superscript~$0$).
Analogously, $\Sigma_1^1$ is induced by arithmetic formulas with existential quantification of type~$1$ objects (sets of natural numbers) and arbitrary (universal and existential) quantification of type~$0$ objects.
So, $\Sigma_1^0$ is part of the first level of the arithmetical hierarchy while $\Sigma_1^1$ is part of the first level of the analytical hierarchy.
In general, level~$\Sigma_n^0$ (level~$\Pi_n^0$) of the arithmetical hierarchy is induced by formulas with at most $n-1$ alternations between existential and universal type~$0$ quantifiers, starting with an existential (universal) quantifier. 
Similar hierarchies can be defined for arithmetic of any fixed order by limiting the alternations of the highest-order quantifiers and allowing arbitrary lower-order quantification.
In this work, the highest order we are concerned with is three, i.e., quantification over sets of sets of natural numbers.

$\hyltl$ satisfiability is $\Sigma_1^1$-complete~\cite{hyperltlsatconf}, $\hyltl$ finite-state satisfiability is $\Sigma_1^0$-complete~\cite{FinkbeinerH16,hyperltlsat}, and, as mentioned above, $\sohyltl$ model-checking is $\Sigma_1^1$-hard~\cite{DBLP:conf/cav/BeutnerFFM23}, but, prior to this work, no upper bounds were known for $\sohyltl$.

Another yardstick is truth for order~$k$ arithmetic, i.e., the question whether a given sentence of order~$k$ arithmetic evaluates to true. In the following, we are in particular interested in the case~$k=3$, i.e., we consider formulas with arbitrary quantification over type~$0$ objects, type~$1$ objects, and type~$2$ objects (sets of sets of natural numbers). 
Note that these formulas span the whole third hierarchy, as we allow arbitrary nesting of existential and universal third-order quantification.

\subsection{Our Contributions}
In this work, we determine the exact complexity of $\sohyltl$ satisfiability, finite-state satisfiability, and model-checking, for the full logic and the two fragments introduced by Beutner et al.~\cite{DBLP:conf/cav/BeutnerFFM23}, as well as for two variants of the semantics.

An important stepping stone for us is the investigation of the cardinality of models of $\sohyltl$. 
It is known that every satisfiable $\hyltl$ sentence has a countable model, and that some have no finite models~\cite{FZ17}.
This restricts the order of arithmetic that can be simulated in $\hyltl$ and explains in particular the $\Sigma_1^1$-completeness of $\hyltl$ satisfiability~\cite{hyperltlsatconf}.
We show that (unsurprisingly) second-order quantification allows to write formulas that only have uncountable models by generalizing the lower bound construction of $\hyltl$ to $\sohyltl$. 
Note that the cardinality of the continuum is a trivial upper bound on the size of models, as they are sets of traces. 

With this tool at hand, we are able to show that $\sohyltl$ satisfiability is equivalent to truth in third-order arithmetic, i.e., much harder than $\hyltl$ satisfiability.
This increase in complexity is not surprising, as second-order quantification can be expected to increase the complexity considerably.
But what might be surprising at first glance is that the problem is not $\Sigma_1^2$-complete, i.e., at the same position of the third hierarchy that $\hyltl$ satisfiability occupies in the analytic hierarchy (see Figure~\ref{fig_hierarchies}). However, arbitrary second-order trace quantification corresponds to arbitrary quantification over type 2 objects, which allows to capture the full third hierarchy.
Furthermore, we also show that $\sohyltl$ finite-state satisfiability is equivalent to truth in third-order arithmetic, and therefore as hard as general satisfiability.
This should be contrasted with the situation for $\hyltl$ described above, where finite-state satisfiability is $\Sigma_1^0$-complete (i.e., recursively enumerable) and thus much simpler than general satisfiability, which is $\Sigma_1^1$-complete.

Finally, our techniques for $\sohyltl$ satisfiability also shed light on the exact complexity of $\sohyltl$ model-checking, which we show to be equivalent to  truth in third-order arithmetic as well, i.e., all three problems we consider have the same complexity. In particular, this increases the lower bound on $\sohyltl$ model-checking from $\Sigma_1^1$ to truth in third-order arithmetic.
Again, this has be contrasted with the situation for $\hyltl$, where model-checking is decidable, albeit \tower-complete~\cite{Rabe16diss,MZ20}.

So, quantification over arbitrary sets of traces makes verification very hard. 
However, Beutner et al.~\cite{DBLP:conf/cav/BeutnerFFM23} noticed that many of the applications of $\sohyltl$ described above do not require full second-order quantification, but can be expressed with restricted forms of second-order quantification.
To capture this, they first restrict second-order quantification to smallest/largest sets satisfying a guard  (obtaining the fragment~$\sohyltlfp$)\footnote{In~\cite{DBLP:conf/cav/BeutnerFFM23} this fragment is termed $\sohyltlfpold$. For clarity, since it is not fixed point based, but uses minimality/maximality constraints, we use the subscript \myquot{mm} instead of \myquot{fp}.} and then further restrict those to least fixed points induced by $\hyltl$ definable operators (obtaining the fragment~$\lfpsohyltlfp$). By construction, these least fixed points are unique, i.e., second-order quantification degenerates to least fixed point computation.

As an example, consider again $ \varphi_\ck $ above. The internal constraint 
\begin{align*}
    \forall \pi' \in X.\, \forall \pi''.\ \big(\bigvee\nolimits_{i=1}^n \pi' \sim_i \pi'' \big) \rightarrow \pi'' \in X 
\end{align*}
defines a condition on what traces have to be in the set $X$, and how they are added gradually to $X$, a behavior that can be captured by a fixed point computation for the (monotone) operator induced by the formula above. Since the last part~$\forall \pi' \in X.\ \varphi({\pi'})$ of $\phi_\ck$ universally quantifies over all traces in $X$, and since $X$ is existentially quantified, it is enough to consider the minimal set that satisfies the internal constraint: if \emph{some} set satisfies a universal condition, then so does the minimal set.
This minimal set is exactly the least fixed point of the operator induced by the formula above.
Similar behavior is exhibited by many other applications of the logic, which gives the motivation to explore the fragment~$\lfpsohyltlfp$.

Nevertheless, we show that $\sohyltlfp$ retains the same complexity as $\sohyltl$, i.e., all three problems are still equivalent to  truth in third-order arithmetic: Just restricting to guarded second-order quantification does not decrease the complexity. Furthermore, we show that this is even the case when we allow only minimality constraints ($\sohyltlfpmin$) and if we allow only maximality constraints ($\sohyltlfpmax$).

But if we consider $\lfpsohyltlfp$, the complexity finally decreases: we show that satisfiability is $\Sigma_1^2$-complete while finite-state satisfiability and model-checking are both equivalent to truth in second-order arithmetic.

Furthermore, we introduce and study a new semantics for $\sohyltl$ and its fragments: 
In standard semantics as introduced by Beutner et al.~\cite{DBLP:conf/cav/BeutnerFFM23}, second-order quantifiers range over arbitrary sets of traces that may contain traces that are not in the model.
In closed-world semantics introduced here, second-order quantifiers range only over subsets of the model. 

We show that for all but one case, the three problems have the same complexity under closed-world semantics as they have under standard semantics. 
The sole exception is satisfiability for $\lfpsohyltlfp$ under closed-world semantics, which we prove $\Sigma_1^1$-complete, i.e., as hard as $\hyltl$ satisfiability.
Stated differently, one can add least fixed points of $\hyltl$ definable operators to $\hyltl$ without increasing the complexity of the satisfiability problem.

Table~\ref{table_results} lists our results and compares them to the related logics~$\ltl$, $\hyltl$, and $\hyqptl$ (see Section~\ref{sec_relatedwork} for more details on these results). 
Recall that Beutner et al.\ showed that $\lfpsohyltlfp$ yields (partial) model checking and monitoring algorithms~\cite{DBLP:conf/cav/BeutnerFFM23,BeutnerFFM24}. 
Our results confirm the usability of the $\lfpsohyltlfp$ fragment also from a theoretical point of view, as all problems relevant for verification have significantly lower complexity (albeit, still  highly undecidable).

\begin{table}[h]
    \centering
    \caption{List of our results (in bold and red) and comparison to related logics~\cite{FinkbeinerH16,hyperltlsat,sohypercomplexity,regaud2024complexityhyperqptl}. \myquot{T2A-equiv.} stands for \myquot{equivalent to truth in second-order arithmetic}, and \myquot{T3A-equiv.} stands for \myquot{equivalent to truth in third-order arithmetic}. Unless explicitly specified, results hold for both semantics.}
    \renewcommand{\arraystretch}{1.25}
    \setlength{\tabcolsep}{1pt}
    \footnotesize
    \begin{tabular}{llll}
       Logic &  Satisfiability  &  Finite-state satisfiability  &  Model-checking \\
         \midrule
        $\ltl$ &  \pspace-compl.~\cite{SistlaC85} &  \pspace-compl.~\cite{SistlaC85} &  \pspace-compl.~\cite{SistlaC85} \\
        \rowcolor{lightgray!40}$\hyltl$ & $\Sigma_1^1$-compl.~\cite{hyperltlsatconf}  & $\Sigma_1^0$-compl.~\cite{FinkbeinerH16,hyperltlsat}  & \tower-compl.~\cite{Rabe16diss,MZ20} \\
            $\sohyltl$ &  \textcolor{Maroon}{\textbf{T3A-equiv.} (Thm.\ref{thm_satcomplexity})} &  \textcolor{Maroon}{\textbf{T3A-equiv.} (Thm.\ref{thm_finsatcomplexity})} &  \textcolor{Maroon}{\textbf{T3A-equiv.} (Thm.\ref{thm_mccomplexity})}\\
        \rowcolor{lightgray!40}$\sohyltlfp$ &  \textcolor{Maroon}{\textbf{T3A-equiv.} (Thm.\ref{thm_hyltlmm_complexity})} &  \textcolor{Maroon}{\textbf{T3A-equiv.} (Thm.\ref{thm_hyltlmm_complexity})} &  \textcolor{Maroon}{\textbf{T3A-equiv.} (Thm.\ref{thm_hyltlmm_complexity})}\\
        $\sohyltlfpmax$ &  \textcolor{Maroon}{\textbf{T3A-equiv.} (Thm.\ref{thm_hyltlmm_complexity})} &    \textcolor{Maroon}{\textbf{T3A-equiv.} (Thm.\ref{thm_hyltlmm_complexity})} &    \textcolor{Maroon}{\textbf{T3A-equiv.} (Thm.\ref{thm_hyltlmm_complexity})}\\
        \rowcolor{lightgray!40}$\sohyltlfpmin$ &    \textcolor{Maroon}{\textbf{T3A-equiv.} (Thm.\ref{thm_hyltlmm_complexity})} &    \textcolor{Maroon}{\textbf{T3A-equiv.} (Thm.\ref{thm_hyltlmm_complexity})} &    \textcolor{Maroon}{\textbf{T3A-equiv.} (Thm.\ref{thm_hyltlmm_complexity})}\\
        $\lfpsohyltlfp$ &  \textcolor{Maroon}{\textbf{\boldmath$\Sigma_1^2$-compl.} (std) (Thm.\ref{thm_satcomplexity_lfp_ss})} &    \textcolor{Maroon}{\textbf{T2A-equiv.} (Thm.\ref{thm_fssatmccomplexity_lfp})} &    \textcolor{Maroon}{\textbf{T2A-equiv.} (Thm.\ref{thm_fssatmccomplexity_lfp})}\\
                        & \textcolor{Maroon}{\textbf{\boldmath$\Sigma_1^1$-compl.} (cw) (Thm.\ref{thm_lfpsatcomplexity})} &&\\
        \rowcolor{lightgray!40}$\hyqptl$ &  $\Sigma^2_1$-compl.~\cite{regaud2024complexityhyperqptl} & $\Sigma_1^0$-compl.~\cite{Rabe16diss} & $\tower$-compl.~\cite{Rabe16diss}\\
        $\hyqptlplus$ & T3A-equiv.~\cite{regaud2024complexityhyperqptl} &  T3A-equiv.~\cite{regaud2024complexityhyperqptl} &  T3A-equiv.~\cite{regaud2024complexityhyperqptl}\\
    \end{tabular}
    \label{table_results}
\end{table}

\subsection{Outline}
In Section~\ref{sec:prels}, we introduce $\sohyltl$, the problems we are interested in, and how to use arithmetic to classify the complexity of undecidable problems.
In Section~\ref{sec_models}, we show that there are $\sohyltl$ sentences that only have uncountable models.
The proof of this result is then used in Section~\ref{sec_sat} to show that $\sohyltl$ satisfiability and finite-state satisfiability are equivalent to truth in third-order arithmetic. 
Then, in Section~\ref{sec_mc}, we show that similar techniques allow us to show that $\sohyltl$ model-checking is equivalent to truth in third-order arithmetic.

Then, in Section~\ref{sec_mm}, we show that all three problems for $\sohyltl$ can be reduced to the analogous problems for $\sohyltlfp$ where all second-order quantifiers are guarded. Hence, all three problems for $\sohyltlfp$ are equivalent to truth in third-order arithmetic as well.

Finally, in Section~\ref{sec_lfp}, we consider the fragment~$\lfpsohyltlfp$ where second-order quantification is restricted to computation of least fixed points. Here, all three problems are simpler, but still highly undecidable.

We conclude with discussing related work in Section~\ref{sec_relatedwork} and with mentioning some questions for future work in Section~\ref{sec_conc}.

\section{Preliminaries}
\label{sec:prels}

We denote the nonnegative integers by $\nats$. 
An alphabet is a nonempty finite set. 
The set of infinite words over an alphabet~$\Sigma$ is denoted by $\Sigma^\omega$.
Let $\ap$ be a nonempty finite set of atomic propositions. 
A trace over $\ap$ is an infinite word over the alphabet~$\pow{\ap}$.
Given a subset~$\ap' \subseteq \ap$, the $\ap'$-projection of a trace~$t(0)t(1)t(2) \cdots$ over $\ap$ is the trace~$(t(0) \cap \ap')(t(1) \cap \ap')(t(2) \cap \ap') \cdots$ over $\ap'$.
The $\ap'$-projection of $T \subseteq (\pow{\ap})^\omega$ is defined as the set of $\ap$-projections of traces in $T$.
Now, let $\ap$ and $\ap'$ be two disjoint sets, let $t$ be a trace over~$\ap$, and let $t'$ be a trace over $\ap'$. 
Then, we define $t \merge t'$ as the pointwise union of $t$ and $t'$, i.e., $t \merge t'$ is the trace over $\ap \cup \ap'$ defined as $(t(0) \cup t'(0))(t(1) \cup t'(1))(t(2) \cup t'(2))\cdots$.

A finite transition system~$\tsys = (V,E,I, \lambda)$ consists of a finite nonempty set~$V$ of vertices, a set~$E \subseteq V \times V$ of (directed) edges, a set~$I \subseteq V$ of initial vertices, and a labeling~$\lambda\colon V \rightarrow \pow{\ap}$ of the vertices by sets of atomic propositions.
We assume that every vertex has at least one outgoing edge.
A path~$\rho$ through~$\tsys$ is an infinite sequence~$\rho(0)\rho(1)\rho(2)\cdots$ of vertices with $\rho(0) \in I$ and $(\rho(n),\rho(n+1))\in E$ for every $n \ge 0$.
The trace of $\rho$ is defined as $ \lambda(\rho ) = \lambda(\rho(0))\lambda(\rho(1))\lambda(\rho(2))\cdots$.
The set of traces of $\tsys$ is $\traces(\tsys) = \set{\lambda(\rho) \mid \rho \text{ is a path through $\tsys$}}$.

\subsection{\texorpdfstring{\boldmath$\sohyltl$}{Second-order HyperLTL}}
\label{subsec_hyperltl}
Let $\fovar$ be a set of first-order trace variables (i.e., ranging over traces) and $\sovar$ be a set of second-order trace variables (i.e., ranging over sets of traces) such that $\fovar \cap \sovar = \emptyset$.
We typically use $\pi$ (possibly with decorations) to denote first-order variables and $X, Y, Z$ (possibly with decorations) to denote second-order variables.
Also, we assume the existence of two distinguished second-order variables~$\univar, \unidisvar  \in \sovar$ such that $\univar$ refers to the set~$(\pow{\ap})^\omega$ of all traces, and $\unidisvar$ refers to the universe of discourse (the set of traces the formula is evaluated over).

The formulas of $\sohyltl$ are given by the grammar
\begin{align*}
\phi  {} \cceq {}&{} \exists X.\ \phi \mid \forall X.\ \phi \mid \exists \pi \in X.\ \phi \mid \forall \pi \in X.\ \phi \mid \psi \\
\psi {}  \cceq {}&{} \proposition_\pi \mid \neg \psi \mid \psi \vee \psi \mid \X \psi \mid \psi \U \psi    
\end{align*}
where $\proposition$ ranges over $\ap$, $\pi$ ranges over $\fovar$, $X$ ranges over $\sovar$, and $\X$ (next) and $\U$ (until) are temporal operators. Conjunction~($\wedge$), exclusive disjunction~$(\oplus)$, implication~($\rightarrow$), and equivalence~$(\leftrightarrow)$ are defined as usual, and the temporal operators eventually~($\F$) and always~($\G$) are derived as $\F\psi = \neg \psi\U \psi$ and $\G \psi = \neg \F \neg \psi$. We measure the size of a formula by its number of distinct subformulas.

The semantics of $\sohyltl$ is defined with respect to a variable assignment, i.e., a partial mapping~$\Pi \colon \fovar \cup \sovar \rightarrow (\pow{\ap})^\omega \cup \pow{(\pow{\ap})^\omega}$
such that
\begin{itemize}
    \item if $\Pi(\pi)$ for $\pi\in\fovar$ is defined, then $\Pi(\pi) \in (\pow{\ap})^\omega$ and
    \item if $\Pi(X)$ for $X \in\sovar$ is defined, then $\Pi(X) \in \pow{(\pow{\ap})^\omega}$.
\end{itemize}
Given a variable assignment~$\Pi$, a variable~$\pi \in \fovar$, and a trace~$t$, we denote by $\Pi[\pi \mapsto t]$ the assignment that coincides with $\Pi$ on all variables but $\pi$, which is mapped to $t$. 
Similarly, for a variable~$X \in \sovar$, and a set~$T$ of traces, $\Pi[X \mapsto T]$ is the assignment that coincides with $\Pi$ everywhere but $X$, which is mapped to $T$.
Furthermore, $\suffix{\Pi}{j}$ denotes the variable assignment mapping every $\pi \in \fovar$ in $\Pi$'s domain to $\Pi(\pi)(j)\Pi(\pi)(j+1)\Pi(\pi)(j+2) \cdots $, the suffix of $\Pi(\pi)$ starting at position $j$.\footnote{Note that the assignment of variables~$X \in \sovar$ is not updated, as this is not necessary for our application: $\suffix{\Pi}{j}$ is used to define the semantics of the temporal operators and at that point the assignments to second-order variables is irrelevant, as we only consider formulas in prenex normal form.}

For a variable assignment~$\Pi$ we define 
\begin{itemize}
	\item $\Pi \models \proposition_\pi$ if $\proposition \in \Pi(\pi)(0)$,
	\item $\Pi \models \neg \psi$ if $\Pi \not\models \psi$,
	\item $\Pi \models \psi_1 \vee \psi_2 $ if $\Pi \models \psi_1$ or $\Pi  \models \psi_2$,
	\item $\Pi \models \X \psi$ if $\suffix{\Pi}{1} \models \psi$,
	\item $\Pi \models \psi_1 \U \psi_2$ if there is a $j \ge 0$ such that $\suffix{\Pi}{j} \models \psi_2$ and for all $0 \le j' < j$ we have $ \suffix{\Pi}{j'} \models \psi_1$ ,
	\item $\Pi \models \exists \pi \in X.\ \phi$ if there exists a trace~$t \in \Pi(X)$ such that $\Pi[\pi \mapsto t] \models \phi$ ,
	\item $\Pi \models \forall \pi \in X.\ \phi$ if for all traces~$t \in \Pi(X)$ we have $\Pi[\pi \mapsto t] \models \phi$,
    \item $\Pi \models \exists X.\ \phi$ if there exists a set~$T \subseteq (\pow{\ap})^\omega$ such that $\Pi[X\mapsto T] \models \phi$, and
    \item $\Pi \models \forall X.\ \phi$ if for all sets~$T \subseteq (\pow{\ap})^\omega$ we have $\Pi[X\mapsto T] \models \phi$.
\end{itemize}

A sentence is a formula in which only the variables $\univar,\unidisvar$ can be free. 
The variable assignment with empty domain is denoted by $\Pi_\emptyset$. 
We say that a set~$T$ of traces satisfies a $\sohyltl$ sentence~$\phi$, written $T \models \phi$, if $\Pi_\emptyset[\univar \mapsto (\pow{\ap})^\omega, \unidisvar \mapsto T]\models \phi$, i.e., if we assign the set of all traces to~$\univar$ and the set~$T$ to the universe of discourse~$\unidisvar$.
In this case, we say that $T$ is a model of $\varphi$.
A transition system~$\tsys$ satisfies $\phi$, written $\tsys \models \phi$, if $\traces(\tsys)\models \phi$.

Although $\sohyltl$ sentences are required to be in prenex normal form, $\sohyltl$ sentences are closed under Boolean combinations, which can easily be seen by transforming such a sentence into an equivalent one in prenex normal form (which might require renaming of variables).
Thus, in examples and proofs we will often use Boolean combinations of $\sohyltl$ sentences.

Throughout the paper, we use the following shorthands to simplify our formulas:
\begin{itemize}
    
    \item We write $\equals{\pi_1}{\pi_2}{\ap'}$ for a set~$\ap' \subseteq \ap$ for the formula~$ \G\bigwedge_{\proposition\in\ap'}(\proposition_{\pi_1} \leftrightarrow \proposition_{\pi_2})$ expressing that the $\ap'$-projection of $\pi_1$ and the $\ap'$-projection of $\pi_2$ are equal.
    
    \item We write $\pi\tracein X $  for the formula~$\exists \pi' \in X.\ \equals{\pi}{\pi'}{\ap}$ expressing that the trace~$\pi$ is in $X$. Note that this shorthand cannot be used under the scope of temporal operators, as we require formulas to be in prenex normal form.

\end{itemize}

\begin{rem}
\label{rem_hyltlisfragment}
$\hyltl$ is the fragment of $\sohyltl$ obtained by disallowing second-order quantification and only allowing first-order quantification of the form~$\exists \pi \in \unidisvar$ and $\forall \pi \in \unidisvar$, i.e., one can only quantify over traces from the universe of discourse.
Hence, we typically simplify our notation to $\exists\pi$ and $\forall\pi$ in $\hyltl$ formulas.
\end{rem}

\subsection{Closed-World Semantics}
Second-order quantification in $\sohyltl$ as defined by Beutner et al.~\cite{DBLP:conf/cav/BeutnerFFM23} (and introduced above) ranges over arbitrary sets of traces (not necessarily from the universe of discourse) and first-order quantification ranges over elements in such sets, i.e., (possibly) again over arbitrary traces.
To disallow this, we introduce \emph{closed-world} semantics for $\sohyltl$, only considering formulas that do not use the variable~$\univar$. We change the semantics of set quantifiers as follows, 
where the closed-world semantics of atomic propositions, Boolean connectives, temporal operators, and trace quantifiers is defined as before:
\begin{itemize}
    \item $\Pi \models_\cw \exists X.\ \phi$ if there exists a set~$T \subseteq \Pi(\unidisvar)$ such that $\Pi[X\mapsto T] \models_\cw \phi$, and 
    \item $\Pi \models_\cw \forall X.\ \phi$ if for all sets~$T \subseteq \Pi(\unidisvar)$ we have $\Pi[X\mapsto T] \models_\cw \phi$.
\end{itemize}
We say that $T \subseteq (\pow{\ap})^\omega$ satisfies $\phi$ under closed-world semantics, if $\Pi_\emptyset[\unidisvar \mapsto T] \models_\cw \phi$.
Hence, under closed-world semantics, second-order quantifiers only range over subsets of the universe of discourse. 
Consequently, first-order quantifiers also range over traces from the universe of discourse. 

\begin{lem}
\label{lem_cwtoclassical}
Every $\univar$-free $\sohyltl$ sentence~$\phi$ can be translated in polynomial time (in $\size{\varphi}$) into a $\sohyltl$ sentence~$\phi'$ such that for all sets~$T$ of traces we have that $T \models_\cw \phi$  if and only if  $T \models \phi'$ (under standard semantics).
\end{lem}

\begin{proof}
Second-order quantification over subsets of the universe of discourse can easily be mimicked by guarding classical quantifiers ranging over arbitrary sets.
Here, we rely on the formula~$\psi_{\subseteq \unidisvar}(X) = \forall \pi \in X.\ \pi \tracein \unidisvar$, which expresses that every trace in $X$ is also in $\unidisvar$.

Now, given an $\univar$-free $\sohyltl$ sentence~$\phi$, let $\phi'$ be the $\sohyltl$ sentence obtained by recursively replacing
\begin{itemize}
    \item each existential second-order quantifier~$\exists X.\ \psi$ in $\phi$ by 
    $\exists X.\ \psi_{\subseteq \unidisvar}(X) \wedge \psi $, 
    
    \item each universal second-order quantifier~$\forall X.\ \psi$ in $\phi$ by 
    $\forall X.\ \psi_{\subseteq \unidisvar}(X)\rightarrow \psi$,
\end{itemize}
and then bringing the resulting sentence into prenex normal form, which can be done as no quantifier is under the scope of a temporal operator.
Then, we have $T \models_\cw \phi$  if and only if  $T \models \phi'$ (under standard semantics).
\end{proof}

Thus, all complexity upper bounds we derive for standard semantics also hold for closed-world semantics and all lower bounds for closed-world semantics hold for standard semantics.

\begin{rem}
\label{rem_cwvsclassicalforfullsetoftraces}
Let $\phi$ be an $\univar$-free $\sohyltl$ sentence over $\ap$.
We have $(\pow{\ap})^\omega \models \varphi$ (under standard semantics) if and only if $(\pow{\ap})^\omega \models_\cw \varphi$, as the second-order quantifiers range in both cases over subsets of $(\pow{\ap})^\omega$, which implies that the trace quantifiers in both cases range over traces from $(\pow{\ap})^\omega$.
\end{rem}

\subsection{Problem Statement}

We are interested in the complexity of the following three problems for fragments of $\sohyltl$ for both semantics:
\begin{itemize}
    \item Satisfiability: Given a sentence~$\phi$, does it have a model, i.e., is there a set~$T$ of traces such that $T \models \phi$?
    \item Finite-state satisfiability: Given a sentence~$\phi$, is it satisfied by a finite transition system, i.e., is there a finite transition system~$\tsys$ such that $\tsys \models \phi$?
    \item Model-checking: Given a sentence~$\phi$ and a finite transition system~$\tsys$, do we have $\tsys \models \phi$?
\end{itemize}

\subsection{Arithmetic}
\label{subsec_arithmetic}

To capture the complexity of undecidable problems, we consider formulas of arithmetic, i.e., predicate logic with signature~$(+, \cdot, <, \in)$, evaluated over the structure~$\natsstruct$. 
A type~$0$ object is a natural number in $\nats$, a type~$1$ object is a subset of $\nats$, and a type~$2$ object is a set of subsets of $\nats$.
In the following, we use lower-case roman letters~$x,y$ (possibly with decorations) for first-order variables, upper-case roman letters~$X,Y$ (possibly with decorations) for second-order variables, and upper-case calligraphic roman letters~$\mathcal{X},\mathcal{Y}$ (possibly with decorations) for third-order variables. 

First-order arithmetic allows to quantify over type~$0$ objects, second-order arithmetic allows to quantify over type~$0$ and type~$1$ objects, and third-order arithmetic allows to quantify over type~$0$, type~$1$, and type~$2$ objects.
Note that every fixed natural number is definable in first-order arithmetic, so we freely use them as syntactic sugar. Similarly, equality can be eliminated if necessary, as it can be expressed using~$<$.

Truth in second-order arithmetic is the following problem: given a sentence~$\phi$ of second-order arithmetic, does $\natsstruct $ satisfy $\phi$?
Truth in third-order arithmetic is defined analogously.
Furthermore, arithmetic formulas with a single free first-order variable define sets of natural numbers. We are interested in the classes
\begin{itemize}
    \item $\Sigma_1^1$ containing sets of the form
    \[\set{x \in\nats \mid \exists X_1 \subseteq \nats.\ \cdots \exists X_k\subseteq \nats.\ \psi(x, X_1, \ldots,X_k )},\] where $\psi$ is a formula of arithmetic with arbitrary quantification over type~$0$ objects (but no other quantifiers), and
    \item $\Sigma_1^2$ containing sets of the form
    \[
    \set{x \in\nats \mid \exists \mathcal{X}_1\subseteq \pow{\nats}.\ \cdots \exists \mathcal{X}_k\subseteq \pow{\nats}. \ \psi(x, \mathcal{X}_1, \ldots,\mathcal{X}_k)},
    \]
    where $\psi$ is a formula of arithmetic with arbitrary quantification over type~$0$ and type~$1$ objects (but no other quantifiers).
\end{itemize}

Let $A$ be an alphabet.
We say that a language~$L \subseteq A^*$ is $\Sigma_1^1$-complete if 
\begin{itemize}
    \item $\set{e(w) \mid w \in L}$ is in $\Sigma_1^1$ (where $e\colon A^* \rightarrow \nats$ is a computable encoding of words over $A$ by natural numbers), and

    \item for every $S \in \Sigma_1^1$ there is a computable function function~$f\colon \nats \rightarrow A^*$ such that $n \in S$ if and only if $f(n) \in L$.

\end{itemize}
We define $\Sigma_1^2$-completeness analogously.
Note that all encodings~$e$ and functions~$f$ we use are polynomial-time computable.

Thus, while, e.g., $\Sigma_1^1$ captures the complexity of formulas with a specific pattern of quantifier alternation, truth in, e.g., second-order arithmetic allows for arbitrary second-order formulas.

\section{\texorpdfstring{The Cardinality of $\sohyltl$ Models}{The Cardinality of Second-order HyperLTL Models}} 
\label{sec_models}

In this section, we investigate the cardinality of models of satisfiable $\sohyltl$ sentences, i.e., the number of traces in the model.
This is an important step in classifying the complexity of $\sohyltl$.

We begin by stating a (trivial) upper bound, which follows from the fact that models are sets of traces.
Here, $\contcard$ denotes the cardinality of the continuum (equivalently, the cardinality of $\pow{\nats}$ and of $(\pow{\ap})^\omega$ for any finite nonempty $\ap$). 

\begin{prop}
    Every satisfiable $\sohyltl$ sentence has a model of cardinality at most~$\contcard$.
\end{prop}

In this section, we show that this trivial upper bound is tight. 

\begin{rem}
\label{remark_unsatisfactory}
There is a very simple, albeit equally unsatisfactory way to obtain the desired lower bound: 
Consider $\forall \pi \in \univar.\ \pi\tracein\unidisvar$ expressing that every trace in the set of all traces is also in the universe of discourse, i.e., $(\pow{\ap})^\omega$ is its only model over $\ap$. 
However, this crucially relies on the fact that $\univar$ is, by definition, interpreted as the set of all traces.
In fact, the formula does not even use second-order quantification.
\end{rem}

We show how to construct a sentence that has only uncountable models, and which retains that property under closed-world semantics (which in particular means it cannot use $\univar$). 
This should be compared with $\hyltl$, where every satisfiable sentence has a countable model~\cite{FZ17}: 
Unsurprisingly, the addition of (even closed-world) second-order quantification increases the cardinality of minimal models, even without cheating.

\begin{exa}
\label{example_hyltlcardlb}
We begin by recalling a construction of Finkbeiner and Zimmermann giving a satisfiable $\hyltl$ sentence~$\psi$ that has no finite models~\cite{FZ17}.
The sentence intuitively posits the existence of a unique trace for every natural number~$n$. Our lower bound for $\sohyltl$ builds upon that construction.

Fix $\ap=\set{\inprop}$ and consider the conjunction~$\psi = \psi_1 \wedge \psi_2 \wedge \psi_3$ of the following three formulas:
\begin{enumerate}
    \item $\psi_1 = \forall \pi.\ \neg \inprop_\pi \U ( \inprop_\pi \wedge \X \G \neg \inprop_\pi ) $: every trace in a model is of the form~$\emptyset^n \set{\inprop} \emptyset^\omega$ for some $n \in\nats$, i.e., every model is a subset of $\set{\emptyset^n \set{\inprop} \emptyset^\omega \mid n\in\nats}$. 
    
    \item $\psi_2 = \exists\pi.\ \inprop_\pi$: the trace $\emptyset^0 \set{\inprop} \emptyset^\omega$ is in every model.
    
    \item $\psi_3 = \forall \pi.\ \exists\pi'.\ \F (\inprop_\pi \wedge \X \inprop_{\pi'})$: if $\emptyset^n \set{\inprop} \emptyset^\omega$ is in a model for some $n\in\nats$, then also $\emptyset^{n+1} \set{\inprop} \emptyset^\omega$.
    
\end{enumerate}
Then, $\psi$ has exactly one model (over $\ap$), namely $\set{\emptyset^n \set{\inprop} \emptyset^\omega \mid n\in\nats}$.
\end{exa}

A trace of the form~$\emptyset^n \set{\inprop} \emptyset^\omega $ encodes the natural number~$n$ and $\psi$ expresses that every model contains the encodings of all natural numbers and nothing else.
But we can of course also encode sets of natural numbers with traces as follows: a trace~$t$ over a set of atomic propositions containing $\inprop$ encodes the set~$\set{n \in \nats \mid \inprop \in t(n)}$. 
In the following, we show that second-order quantification allows us to express the existence of the encodings of all subsets of natural numbers by requiring that for every subset~$S \subseteq \nats$ (quantified as the set $S' = \set{\emptyset^n \set{\inprop} \emptyset^\omega \mid n\in S}$ of traces) there is a trace~$t$ encoding~$S$, which means $\inprop$ is in $t(n)$ if and only if $S'$ contains a trace in which $\inprop$ holds at position $n$.
This equivalence can be expressed in $\sohyltl$.
For technical reasons, we do not capture the equivalence directly but instead use encodings of both the natural numbers that are in $S$ and the natural numbers that are not in $S$.

\begin{thm}
\label{thm_modelsizelowerbound}
There is a satisfiable $\univar$-free $\sohyltl$ sentence that only has models of cardinality~$\contcard$ (both under standard and closed-world semantics).
\end{thm}

\begin{proof}
We first prove that there is a satisfiable $\univar$-free $\sohyltl$ sentence~$\phi_\allsets$  whose unique model (under standard semantics) has cardinality~$\contcard$.
To this end, we fix $\ap = \set{\posprop, \negprop,\setprop,\inprop}$ and consider the conjunction~$\phi_\allsets = \phi_0 \wedge \cdots \wedge \phi_4$ of the following formulas:
\begin{itemize}

\item $\phi_0 = \forall \pi \in \unidisvar.\ \bigvee_{\proposition \in \set{\posprop, \negprop,\setprop}} \G (\proposition_\pi \wedge \bigwedge_{\proposition' \in \set{\posprop, \negprop,\setprop} \setminus \set{\proposition}} \neg \proposition'_\pi )$: In each trace of a model, one of the propositions in $\set{\posprop, \negprop,\setprop}$ holds at every position and the other two propositions in $\set{\posprop, \negprop,\setprop}$ hold at none of the positions. Consequently, we speak in the following about type~$\proposition$ traces for $\proposition \in \set{\posprop, \negprop,\setprop}$.

\item $\phi_1 = \forall \pi \in \unidisvar.\ (\posprop_\pi \vee \negprop_\pi) \rightarrow \neg \inprop_\pi \U ( \inprop_\pi \wedge \X \G \neg \inprop_\pi )$: Type~$\proposition$ traces for $\proposition \in \set{\posprop,\negprop}$ in the model have the form~$\set{\proposition}^n \set{\inprop, \proposition} \set{\proposition}^\omega$ for some $n\in\nats$.

\item $\phi_2 = \bigwedge_{\proposition \in \set{\posprop,\negprop}} \exists\pi \in \unidisvar.\ \proposition_\pi \wedge \inprop_\pi$: for both $\proposition \in \set{\posprop,\negprop}$, the type~$\proposition$ trace $\set{\proposition}^0 \set{\inprop,\proposition} \set{\proposition}^\omega$ is in every model.

\item $\phi_3 = \bigwedge_{\proposition \in \set{\posprop,\negprop}}\forall \pi \in \unidisvar.\ \exists\pi'\in\unidisvar.\ \proposition_\pi \rightarrow ( \proposition_{\pi'} \wedge  \F (\inprop_\pi \wedge \X \inprop_{\pi'}))$: for both $\proposition \in \set{\posprop,\negprop}$, if the type~$\proposition$ trace~$\set{\proposition}^n \set{\inprop,\proposition} \set{\proposition}^\omega$ is in a model for some $n\in\nats$, then also $\set{\proposition}^{n+1} \set{\inprop,\proposition} \set{\proposition}^\omega$.
\end{itemize}

The formulas $\phi_1, \phi_2, \phi_3$ are similar to the formulas~$\psi_1, \psi_2, \psi_3$ from Example~\ref{example_hyltlcardlb}. 
So, every model of $\phi_0 \wedge \cdots \wedge \phi_3$ contains $\set{\set{\posprop}^n \set{\inprop,\posprop} \set{\posprop}^\omega \mid n\in\nats}$ and $\set{\set{\negprop}^n \set{\inprop,\negprop} \set{\negprop}^\omega \mid n\in\nats}$ as subsets, and no other type~$\posprop$ or type~$\negprop$ traces.

Now, consider a set~$T$ of traces over $\ap$ (recall that second-order quantification ranges over arbitrary sets, not only over subsets of the universe of discourse).
We say that $T$ is contradiction-free if it only contains traces of the form~$\set{\posprop}^n \set{\inprop,\posprop} \set{\posprop}^\omega$ or $\set{\negprop}^n \set{\inprop,\negprop} \set{\negprop}^\omega$ and if there is no $n \in \nats$ such that $\set{\posprop}^n \set{\inprop,\posprop} \set{\posprop}^\omega \in T$ and $\set{\negprop}^n \set{\inprop,\negprop} \set{\negprop}^\omega \in T$. 
Furthermore, a trace~$t$ over $\ap$ is consistent with a contradiction-free $T$ if 
\begin{description}
    \item[(C1)] $\set{\posprop}^n \set{\inprop,\posprop} \set{\posprop}^\omega \in T$ implies $\inprop\in t(n)$ and
    \item[(C2)] $\set{\negprop}^n \set{\inprop,\negprop} \set{\negprop}^\omega \in T$ implies $\inprop\notin t(n)$.
\end{description}
Note that $T$ does not necessarily specify the truth value of $\inprop$ in every position of $t$, i.e., in those positions~$n\in\nats$ where neither $\set{\posprop}^n \set{\inprop,\posprop} \set{\posprop}^\omega$ nor $\set{\negprop}^n \set{\inprop,\negprop} \set{\negprop}^\omega$ are in $T$.
Nevertheless, for every trace~$t$ over $\set{\inprop}$ there is a contradiction-free $T$ such that the $\set{\inprop}$-projection of every trace~$t'$ over $\ap$ that is consistent with $T$ is equal to $t$. 
Thus, each of the uncountably many traces over $\set{\inprop}$ is induced by some subset of the model.
\begin{itemize}
\item Hence, we define $\phi_4$ as the formula
\begin{align*}
\forall X.\ {}&{} \overbrace{[ \forall \pi \in X.\ (\pi\tracein \unidisvar \wedge (\posprop_\pi \vee \negprop_\pi)) \wedge \forall \pi \in X.\ \forall \pi' \in X.\ (\posprop_\pi \wedge \negprop_{\pi'}) \rightarrow \neg \F(\inprop_{\pi} \wedge \inprop_{\pi'}) ]}^{\text{$X$ is contradiction-free}} \rightarrow \\
&{}\exists \pi'' \in \unidisvar.\ \forall \pi''' \in X.\ \setprop_{\pi''} \wedge  \underbrace{(\posprop_{\pi'''} \rightarrow \F(\inprop_{\pi'''} \wedge \inprop_{\pi''}))}_{\text{(C1)}}
\wedge 
\underbrace{(\negprop_{\pi'''} \rightarrow \F(\inprop_{\pi'''} \wedge \neg\inprop_{\pi''}))}_{\text{(C2)}},
\end{align*} 
expressing that for every contradiction-free set of traces~$T$, there is a type~$\setprop$ trace~$t''$ in the model (note that $\pi''$ is required to be in $\unidisvar$) that is consistent with $T$. 
\end{itemize}
While $\phi_\allsets$ is not in prenex normal form, it can easily be turned into an equivalent formula in prenex normal form (at the cost of readability).

Now, the set 
\begin{align*}
T_\allsets ={}&{} \set{\set{\posprop}^n \set{\inprop,\posprop} \set{\posprop}^\omega \mid n\in\nats} \cup \set{\set{\negprop}^n \set{\inprop,\negprop} \set{\negprop}^\omega \mid n\in\nats} \cup\\
&{}\set{ (t(0) \cup \set{\setprop})(t(1) \cup \set{\setprop})(t(2) \cup \set{\setprop}) \cdots \mid t \in (\pow{\set{\inprop}})^\omega} 
\end{align*}
of traces satisfies $\phi_\allsets$.
On the other hand, every model of $\phi_\allsets$ must indeed contain $T_\allsets$ as a subset, as $\phi_\allsets$ requires the existence of all of its traces in the model.
Finally, due to $\phi_0$ and $\phi_1$, a model (over $\ap$) cannot contain any traces that are not in $T_\allsets$,
i.e., $T_\allsets$ is the unique model of $\phi_\allsets$.

To conclude, we just remark that 
\[\set{ (t(0) \cup \set{\setprop})(t(1) \cup \set{\setprop})(t(2) \cup \set{\setprop}) \cdots \mid t \in (\pow{\set{\inprop}})^\omega} \subseteq T_\allsets\] has indeed cardinality~$\contcard$, as
$(\pow{\set{\inprop}})^\omega$ has cardinality~$\contcard$.

Finally, let us consider closed-world semantics. The second-order quantifier in $\phi_4$ (the only one in $\phi_\allsets$) is already restricted to subsets of the universe of discourse.
Thus, $\phi_\allsets$ has the unique model $T_\allsets$ even under closed-world semantics.
\end{proof}

\section{\texorpdfstring{The Complexity of $\sohyltl$ Satisfiability}{The Complexity of Second-order HyperLTL Satisfiability}}
\label{sec_sat}

A $\sohyltl$ sentence is satisfiable if it has a model. The $\sohyltl$ satisfiability problem asks, given a $\sohyltl$ sentence~$\phi$, whether $\phi$ is satisfiable. In this section, we determine tight bounds on the complexity of $\sohyltl$ satisfiability and some of its variants.

Recall that in Section~\ref{sec_models}, we encoded sets of natural numbers as traces over a set of propositions containing $\inprop$ and encoded natural numbers as singleton sets.
The proof of Theorem~\ref{thm_modelsizelowerbound} relies on constructing a sentence that requires each of its models to encode every subset of $\nats$ by a trace in the model. 
Hence, sets of traces can encode sets of sets of natural numbers, i.e., type~$2$ objects. 

Another important ingredient in the following proof is the implementation of addition and multiplication in $\hyltl$~\cite{hyperltlsat}.
Let $\ap_\arith = \set{\argone, \argtwo, \res, \add, \mult}$ and let $T_\plustimes$ be the set of all traces $t \in (\pow{\ap_\arith})^\omega$ such that
\begin{itemize}

    \item there are unique $n_1, n_2, n_3 \in \nats$  with $\argone \in t(n_1)$, $\argtwo \in t(n_2)$, and $\res \in t(n_3)$, and

    \item either $\add \in t(n)$ and $\mult \notin t(n)$ for all $n$, and $n_1+n_2 = n_3$, or $\mult \in t(n)$ and $\add \notin t(n)$ for all $n$, and $n_1 \cdot n_2 = n_3$.

\end{itemize}

\begin{prop}[Lemma 5.5 of \cite{hyperltlsat}]
\label{prop_plustimesinhyperltl}
There is a satisfiable $\hyltl$ sentence $\phi_\plustimes$ such that the $\ap_\arith$-projection of every model of $\phi_\plustimes$ is $T_\plustimes$.
\end{prop}

Combining the capability of quantifying over type~$0$, type~$1$, and type~$2$ objects and the encoding of addition and multiplication, we show that $\sohyltl$ is at least as hard as truth in third-order arithmetic.
A matching upper bound will be obtained by showing that one can encode (using arithmetic) traces as sets of natural numbers and thus sets of traces as sets of sets of natural numbers. 

\begin{thm}
\label{thm_satcomplexity}
The $\sohyltl$ satisfiability problem is polynomial-time equivalent to truth in third-order arithmetic. The lower bound holds even for $\univar$-free sentences.
\end{thm}

\begin{proof}
We begin with the lower bound by reducing truth in third-order arithmetic to $\sohyltl$ satisfiability: we present a polynomial-time translation from sentences~$\phi$ of third-order arithmetic to $\sohyltl$ sentences~$\phi'$ such that $\natsstruct \models \phi$ if and only if $\phi'$ is satisfiable.

Given a third-order sentence~$\phi$, we define \[\phi' = \exists X_\allsets.\ \exists X_\arith.\ (\phi_\allsets[\unidisvar/X_\allsets] \wedge \phi'_\plustimes \wedge \hyperize(\phi))\] where 
\begin{itemize}
    \item $\phi_\allsets[\unidisvar/X_\allsets]$ is the $\sohyltl$ sentence from the proof of Theorem~\ref{thm_modelsizelowerbound} where every occurrence of $\unidisvar$ is replaced by $X_\allsets$ and thus enforces every subset of $\nats$ to be encoded in the interpretation of $X_\allsets$ (as introduced in the proof of Theorem~\ref{thm_modelsizelowerbound}),
    \item $\phi'_\plustimes$ is the $\sohyltl$ formula obtained from the $\hyltl$ formula~$\phi_\plustimes$ by replacing each quantifier~$\exists \pi$ ($\forall \pi$, respectively) by $\exists \pi \in X_\arith$ ($\forall \pi \in X_\arith$, respectively) and thus enforces that $X_\arith$ is interpreted by a set whose $\ap_\arith$-projection is $T_\plustimes$, and
\end{itemize}
  where $\hyperize(\phi)$ is defined inductively as follows:
  
\begin{itemize}
    
    \item For third-order variables~$\calY$, \[\hyperize(\exists \calY.\ \psi) = \exists X_\calY.\ (\forall \pi \in X_\calY.\ \exists \pi' \in X_\allsets.\ (\equals{\pi}{\pi'}{\set{\posprop, \negprop,\setprop,\inprop}}) \wedge \setprop_{\pi}) \wedge \hyperize(\psi).\]
    
    \item For third-order variables~$\calY$, \[\hyperize(\forall \calY.\ \psi) = \forall X_\calY.\ (\forall \pi \in X_\calY.\  \exists \pi' \in X_\allsets.\ (\equals{\pi}{\pi'}{\set{\posprop, \negprop,\setprop,\inprop}}) \wedge \setprop_{\pi}) \rightarrow \hyperize(\psi).\]
    
    \item For second-order variables~$Y$, $\hyperize(\exists Y.\ \psi) = \exists \pi_Y \in X_\allsets.\  \hyperize(\psi)$.
    
    \item For second-order variables~$Y$, $\hyperize(\forall Y.\ \psi) = \forall \pi_Y \in X_\allsets.\ \hyperize(\psi)$.
    
    \item For first-order variables~$y$, \[\hyperize(\exists y.\ \psi) = \exists \pi_y \in X_\allsets.\  [(\neg \inprop_{\pi_y}) \U (\inprop_{\pi_y}\wedge \X\G\neg\inprop_{\pi_y})] \wedge \hyperize(\psi).\]
    
    \item For first-order variables~$y$, \[\hyperize(\forall y.\ \psi) = \forall \pi_y \in X_\allsets.\ [(\neg \inprop_{\pi_y}) \U (\inprop_{\pi_y}\wedge \X\G\neg\inprop_{\pi_y})] \rightarrow \hyperize(\psi).\]
    
    \item $\hyperize(\psi_1 \vee \psi_2) = \hyperize(\psi_1) \vee \hyperize(\psi_2)$.
    
    \item $\hyperize(\neg \psi) = \neg \hyperize(\psi)$.

    \item For second-order variables~$Y$ and third-order variables~$\calY$, \[\hyperize(Y \in \calY) = \exists \pi \in X_{\calY}.\ \equals{\pi_Y}{\pi}{\set{\inprop}}.\]

    \item For first-order variables~$y$ and second-order variables~$Y$, $\hyperize(y \in Y) = \F(\inprop_{\pi_y} \wedge \inprop_{\pi_{Y}})$.
    
    \item For first-order variables~$y,y'$, $\hyperize(y < y') = \F(\inprop_{\pi_y} \wedge \X\F\inprop_{\pi_{y'}})$.
    
    \item For first-order variables~$y_1,y_2, y$, \[\hyperize(y_1 + y_2 = y) = \exists \pi \in X_\arith.\ \add_{\pi} \wedge 
    \F(\argone_{\pi} \wedge \inprop_{\pi_{y_1}}) \wedge 
    \F(\argtwo_{\pi} \wedge \inprop_{\pi_{y_2}}) \wedge
    \F(\res_{\pi} \wedge \inprop_{\pi_y}).\]

    \item For first-order variables~$y_1,y_2, y$, \[\hyperize(y_1 \cdot y_2 = y) = \exists \pi \in X_\arith.\ \mult_{\pi} \wedge 
    \F(\argone_{\pi} \wedge \inprop_{\pi_{y_1}}) \wedge 
    \F(\argtwo_{\pi} \wedge \inprop_{\pi_{y_2}}) \wedge
    \F(\res_{\pi} \wedge \inprop_{\pi_y}).\]

\end{itemize}
While $\phi'$ is not in prenex normal form, it can easily be brought into prenex normal form, as there are no quantifiers under the scope of a temporal operator.

As we are evaluating $\phi'$ w.r.t.\ standard semantics and the variable~$\unidisvar$ (interpreted with the model) does not occur in $\phi'$, satisfaction of $\varphi'$ is independent of the model, i.e., for all sets~$T,T'$ of traces, $T \models \varphi'$ if and only if $T' \models \varphi'$.
So, let us fix some set~$T$ of traces.
An induction shows that $\natsstruct$ satisfies $\phi$ if and only if $T$ satisfies $\phi'$.
Altogether we obtain the desired equivalence between $\natsstruct \models \phi$ and $\phi'$ being satisfiable.

For the upper bound, we conversely reduce $\sohyltl$ satisfiability to truth in third-order arithmetic: we present a polynomial-time translation from $\sohyltl$ sentences~$\phi$ 
 to sentences~$\phi'$ of third-order arithmetic such that  $\phi$ is satisfiable if and only if  $\natsstruct \models \phi'$. Here, we assume $\ap$ to be fixed, so that we can use $\size{\ap}$ as a constant in our formulas (which is definable in first-order arithmetic). 

Let $\pair \colon \nats\times\nats \rightarrow\nats$ denote Cantor's pairing function defined as $\pair(i,j) = \frac{1}{2}(i+j)(i+j+1) +j$, which is a bijection.
Furthermore, fix some bijection~$e \colon \ap \rightarrow\set{0,1,\ldots,\size{\ap}-1}$.
Then, we encode a trace~$t \in (\pow{\ap})^\omega$ by the set~$S_t =\set{\pair(j,e(\proposition)) \mid j \in \nats \text{ and } \proposition \in t(j)} \subseteq \nats$.
As $\pair$ is a bijection, we have that $t \neq t'$ implies $S_t \neq S_{t'}$.
While not every subset of $\nats$ encodes some trace~$t$, the first-order formula 
\[\phi_\istrace(Y) = \forall x.\ \forall y.\ y \ge \size{\ap} \rightarrow \pair(x,y) \notin Y \] checks if a set does encode a trace.
Here, we use $\pair$ as syntactic sugar, which is possible as the definition of $\pair$ only uses addition and multiplication.

As (certain) sets of natural numbers encode traces, sets of (certain) sets of natural numbers encode sets of traces. 
This is sufficient to reduce $\sohyltl$ to third-order arithmetic, which allows the quantification over sets of sets of natural numbers. 
Before we present the translation, we need to introduce some more auxiliary formulas:
\begin{itemize}
    \item Let $\calY$ be a third-order variable (i.e., $\calY$ ranges over sets of sets of natural numbers). Then, the formula
    \[\phi_\onlytraces(\calY) = \forall Y.\  Y  \in \calY \rightarrow \phi_\istrace(Y) \]
    checks if a set of sets of natural numbers only contains sets encoding a trace.

    \item Further, the formula
    \[\phi_\alltraces(\calY) = \phi_\onlytraces(\calY) \wedge \forall Y.\ \phi_\istrace(Y) \rightarrow Y \in\calY \]
    checks if a set of sets of natural numbers contains exactly the sets encoding a trace.
\end{itemize}

Now, we are ready to define our encoding of $\sohyltl$ in third-order arithmetic.
Given a $\sohyltl$ sentence~$\phi$, let 
\[\phi' = \exists \calY_a.\ \exists\calY_d.\ \phi_\alltraces(\calY_a) \wedge \phi_\onlytraces(\calY_d)\wedge (\arithmetize(\phi))(0)\]
where $\arithmetize(\phi)$ is defined inductively as presented below.
Note that $\phi'$ requires $\calY_a$ to contain exactly the encodings of all traces (i.e., it corresponds to the distinguished $\sohyltl$ variable~$\univar$ in the following translation) and $\calY_d$ is an existentially quantified set of trace encodings (i.e., it corresponds to the distinguished $\sohyltl$ variable~$\unidisvar$ in the following translation).

In the inductive definition of $\arithmetize(\phi)$, we will employ a free first-order variable~$i$ to denote the position at which the formula is to be evaluated to capture the semantics of the temporal operators.
As seen above, in $\phi'$, this free variable is set to zero in correspondence with the $\sohyltl$ semantics.
\begin{itemize}
    
    \item $\arithmetize(\exists X.\ \psi) = \exists \calY_X.\ \phi_\onlytraces(\calY_X) \wedge \arithmetize(\psi)$. Here, the free variable of $\arithmetize(\exists X.\ \psi)$ is the free variable of $\arithmetize(\psi)$. 
    
    \item $\arithmetize(\forall X.\ \psi) = \forall \calY_X.\ \phi_\onlytraces(\calY_X) \rightarrow \arithmetize(\psi)$. Here, the free variable of $\arithmetize(\forall X.\ \psi)$ is the free variable of $\arithmetize(\psi)$.
    
    \item $\arithmetize(\exists\pi\in X.\ \psi) = \exists Y_\pi.\ Y_\pi\in \calY_X \wedge \arithmetize(\psi)$. Here, the free variable of $\arithmetize(\exists \pi \in X.\ \psi)$ is the free variable of $\arithmetize(\psi)$.
    
    \item $\arithmetize(\forall\pi\in X.\ \psi) = \forall Y_\pi.\ Y_\pi\in \calY_X \rightarrow \arithmetize(\psi)$. Here, the free variable of $\arithmetize(\forall \pi \in X.\ \psi)$ is the free variable of $\arithmetize(\psi)$.
    
    \item $\arithmetize(\psi_1 \vee \psi_2) = \arithmetize(\psi_1) \vee \arithmetize(\psi_2)$. Here, we require that the free variables of $\arithmetize(\psi_1)$ and $\arithmetize(\psi_2)$ are the same (which can always be achieved by variable renaming), which is then also the free variable of $\arithmetize(\psi_1 \vee \psi_2)$.
    
    \item $\arithmetize(\neg\psi) = \neg\arithmetize(\psi)$. Here, the free variable of $\arithmetize(\neg\psi) $ is the free variable of $ \arithmetize(\psi)$.
    
     \item $\arithmetize(\X\psi) =  \exists i' (i' = i+1) \wedge \arithmetize(\psi)$, where $i'$ is the free variable of $\arithmetize(\psi)$ and $i$ is the free variable of $\arithmetize(\X\psi)$.
    
    \item $\arithmetize(\psi_1\U\psi_2) =  \exists i_2.\ i_2 \ge i \wedge \arithmetize(\psi_2) \wedge \forall i_1.\ (i \le i_1 \wedge i_1 < i_2) \rightarrow \arithmetize(\psi_1)$, where $i_j$ is the free variable of $\arithmetize(\psi_j)$, 
    and $i$ is the free variable of $\arithmetize(\psi_1\U\psi_2)$.

    \item $\arithmetize(\proposition_\pi) = \pair(i,e(\proposition)) \in Y_\pi$, i.e., $i$ is the free variable of $\arithmetize(\proposition_\pi)$.
    
\end{itemize}
Now, an induction shows that $\Pi_\emptyset[\univar\rightarrow (\pow{\ap})^\omega, \unidisvar \mapsto T] \models \phi$ if and only if $\natsstruct$ satisfies $(\arithmetize(\phi))(0)$ when the variable~$\calY_a$ is interpreted by the encoding of $(\pow{\ap})^\omega$ and $\calY_d$ is interpreted by the encoding of $T$.
Hence, $\phi$ is indeed satisfiable if and only if $\natsstruct$ satisfies $\phi'$.
\end{proof}

In the lower bound proof above, we have turned a sentence~$\phi$ of third-order arithmetic into a $\sohyltl$ sentence~$\phi'$ such that $\natsstruct\models\phi$  if and only if  $\phi'$ is satisfiable. 
In fact, we have constructed $\phi'$ such that if it is satisfiable, then every set of traces satisfies it, in particular~$(\pow{\ap})^\omega$.
Recall that Remark~\ref{rem_cwvsclassicalforfullsetoftraces} states that $(\pow{\ap})^\omega$ satisfies $\phi'$ under standard semantics if and only if $(\pow{\ap})^\omega$ satisfies $\phi'$ under closed-world semantics.
Thus, altogether we obtain that $\natsstruct\models\phi$  if and only if  $\phi'$ is satisfiable under closed-world semantics, i.e, the lower bound holds even under closed-world semantics.
Together with Lemma~\ref{lem_cwtoclassical}, this settles the complexity of $\sohyltl$ satisfiability under closed-world semantics.

\begin{cor}
\label{cor_satcomplexity_cw}
The $\sohyltl$ satisfiability problem under closed-world semantics is polynomial-time equivalent to truth in third-order arithmetic.
\end{cor}

The $\sohyltl$ finite-state satisfiability problem asks, given a $\sohyltl$ sentence~$\phi$, whether there is a finite transition system satisfying~$\phi$.
Note that we do not ask for a finite set~$T$ of traces satisfying $\phi$. In fact, the set of traces of the finite transition system may still be infinite or even uncountable. 
Nevertheless, the problem is potentially simpler, as there are only countably many finite transition systems (and their sets of traces are much simpler).
However, we show that the finite-state satisfiability problem is as hard as the general satisfiability problem, as $\sohyltl$ allows the quantification over arbitrary (sets of) traces, i.e., restricting the universe of discourse to the traces of a finite transition system does not restrict second-order quantification at all (as the set of all traces is represented by a finite transition system).
This has to be contrasted with the finite-state satisfiability problem for $\hyltl$ (defined analogously), which is $\Sigma_1^0$-complete (a.k.a.\ recursively enumerable), as $\hyltl$ model-checking of finite transition systems is decidable~\cite{ClarksonFKMRS14}.

\begin{thm}
\label{thm_finsatcomplexity}
The $\sohyltl$ finite-state satisfiability problem is polynomial-time equivalent to truth in third-order arithmetic. The lower bound holds even for $\univar$-free sentences.
\end{thm}

\begin{proof}
For the lower bound under standard semantics, we reduce truth in third-order arithmetic to $\sohyltl$ finite-state satisfiability: we present a polynomial-time translation from sentences~$\phi$ of third-order arithmetic to $\sohyltl$ sentences~$\phi'$ such that $\natsstruct \models \phi$ if and only if $\phi'$ is satisfied by a finite transition system.

So, let $\phi$ be a sentence of third-order arithmetic. Recall that in the proof of Theorem~\ref{thm_satcomplexity}, we have shown how to construct from $\phi$ the $\sohyltl$ sentence~$\phi'$ such that the following three statements are equivalent:
\begin{itemize}
    \item $\natsstruct \models \phi$.
    \item $\phi'$ is satisfiable.
    \item $\phi'$ is satisfied by all sets~$T$ of traces (and in particular by some finite-state transition system).
\end{itemize}
Thus, the lower bound follows from Theorem~\ref{thm_satcomplexity}.

For the upper bound, we conversely reduce $\sohyltl$ finite-state satisfiability to truth in third-order arithmetic: we present a polynomial-time translation from $\sohyltl$ sentences~$\phi$ to sentences~$\phi''$ of third-order arithmetic such that  $\phi$ is satisfied by a finite transition system if and only if  $\natsstruct \models \phi''$.

Recall that in the proof of Theorem~\ref{thm_satcomplexity}, we have constructed a sentence
\[\phi' = \exists \calY_a.\ \exists\calY_d.\ \phi_\alltraces(\calY_a) \wedge \phi_\onlytraces(\calY_d)\wedge (\arithmetize(\phi))(0)\]
of third-order arithmetic where $\calY_a$ represents the distinguished $\sohyltl$ variable~$\univar$, $\calY_d$ represents the distinguished $\sohyltl$ variable~$\unidisvar$, and where $\arithmetize(\phi)$ is the encoding of $\phi$ in $\sohyltl$.

To encode the general satisfiability problem it was sufficient to express that $\calY_d$ only contains traces.
Here, we now require that $\calY_d$ contains exactly the traces of some finite transition system, which can easily be expressed in second-order arithmetic\footnote{With a little more effort, and a little less readability, first-order suffices for this task, as finite transition systems can be encoded by natural numbers.} as follows.

We begin with a formula~$\phi_\ists(n, E, I, \ell)$ expressing that the second-order variables~$E$, $I$, and $\ell$ encode a transition system with set~$\set{0,1, \ldots, n-1}$ of vertices.
Our encoding will make extensive use of the pairing function introduced in the proof of Theorem~\ref{thm_satcomplexity}.
Formally, we define $\phi_\ists(n, E, I, \ell)$ as the conjunction of the following formulas (where all quantifiers are first-order and we use $\pair$ as syntactic sugar):
\begin{itemize}
    \item $n > 0$: the transition system is nonempty.
    
    \item $\forall y.\ y \in E \rightarrow \exists v.\ \exists v'.\ (v < n \wedge v'<n \wedge y = \pair(v,v'))$: edges are pairs of vertices.

    \item $\forall v.\ v < n \rightarrow \exists v'.\ (v' < n \wedge \pair(v,v') \in E)$: every vertex has a successor.

    \item $\forall v.\ v \in I \rightarrow v < n$: the set of initial vertices is a subset of the set of all vertices.

    \item $\forall y.\ y \in \ell \rightarrow \exists v.\ \exists p.\ (v < n \wedge p < \size{\ap} \wedge y = \pair(v,p))$: the labeling of $v$ by the proposition encoded by $p$ is encoded by the pair $(v,p)$. Here, as in the proof of Theorem~\ref{thm_satcomplexity}, we assume $\ap$ to be fixed and therefore can use $\size{\ap}$ as a constant, and we identify propositions by numbers in $\set{0,1,\ldots, \size{\ap}-1}$.
\end{itemize}

Next, we define $\phi_\ispath(P, n, E, I)$, expressing that the second-order variable~$P$ encodes a path through the transition system encoded by $n$, $E$, and $I$, as the conjunction of the following formulas:
\begin{itemize}
    \item $\forall j.\ \exists v.\ (v < n \wedge \pair(j,v)\in P \wedge \neg \exists v'.\ (v' \neq v \wedge \pair(j,v') \in P))$: the fact that at position~$j$ the path visits vertex~$v$ is encoded by the pair $(j,v)$. Exactly one vertex is visited at each position.

    \item $\exists v.\ v\in I \wedge \pair(0,v) \in P$: the path starts in an initial vertex.

    \item $\forall j.\ \exists v.\ \exists v'.\ \pair(j,v) \in P \wedge \pair(j+1, v') \in P \wedge \pair(v,v') \in E $: successive vertices in the path are indeed connected by an edge.
\end{itemize}

Finally, we define $\phi_\traceof(T, P, \ell)$, expressing that the second-order variable $T$ encodes the trace (using the encoding from the proof of Theorem~\ref{thm_satcomplexity}) of the path encoded by the second-order variable~$P$, as the  following formula:
\begin{itemize}
    \item $\forall j.\ \forall p.\ \pair(j,p) \in T \leftrightarrow (\exists v.\ \pair(j,v) \in P \wedge \pair(v,p) \in \ell)$: a proposition holds in the trace at position~$j$ if and only if it is in the labeling of the $j$-th vertex of the path.
\end{itemize}

Now, we define the sentence~$\phi''$ as 
\begin{align*}
\exists \calY_a.\ \exists\calY_d.\ {}&{} \phi_\alltraces(\calY_a) \wedge \phi_\onlytraces(\calY_d) \wedge\\
&{}\Big[\underbrace{\exists n.\ \exists E.\ \exists I.\ \exists \ell.\ \phi_\ists(n, E, I, \ell)}_{\text{there exists a transition system~$\tsys$}} \wedge \\
&{}\underbrace{(\forall T.\ T \in \calY_d \rightarrow \exists P.\ (\phi_\ispath(P, n, E, I) \wedge \phi_\traceof(T,P, \ell)))}_{\text{$\calY_d$ contains only traces of paths through $\tsys$}}\wedge\\
&{}\underbrace{( \forall P.\ (\phi_\ispath(P, n,E,I) \rightarrow \exists T.\ T \in \calY_d \wedge \phi_\traceof(T,P, \ell)) )}_{\text{$\calY_d$ contains all traces of paths through $\tsys$.}}\Big]\wedge  (\arithmetize(\phi))(0),
\end{align*}
which holds in $\natsstruct$ if and only if $\phi$ is satisfied by a finite transition system.
\end{proof}

Again, the lower bound proof can easily be extended to the case of closed-world semantics, using the same arguments as in the case of general satisfiability.

\begin{cor}
\label{cor_finsatcomplexity_cw}
The $\sohyltl$ finite-state satisfiability problem under closed-world semantics is polynomial-time equivalent to truth in third-order arithmetic.
\end{cor}

Let us also just remark that the proof of Theorem~\ref{thm_finsatcomplexity} can easily be adapted to show that other natural variations of the satisfiability problem are also polynomial-time equivalent to  truth in third-order arithmetic, e.g., satisfiability by countable transition systems, satisfiability by finitely branching transition systems, etc.
In fact, as long as a class~$\class$ of transition systems is axiomatizable in third-order arithmetic, the $\sohyltl$ satisfiability problem restricted to transition systems in $\class$ is reducible to truth in third-order arithmetic.
Similarly, truth in third-order arithmetic is reducible to $\sohyltl$ satisfiability for every nonempty class of models (w.r.t.\ standard semantics) and to every class of models containing at least one model whose language contains $T_\allsets \cup T_\arith$ (w.r.t.\ closed-world semantics).

\section{\texorpdfstring{The Complexity of $\sohyltl$ Model-Checking}{The Complexity of Second-order HyperLTL Model-Checking}}
\label{sec_mc}

The $\sohyltl$ model-checking problem asks, given a finite transition system~$\tsys$ and a $\sohyltl$ sentence~$\phi$, whether $\tsys \models \phi$.
Beutner et al.~\cite{DBLP:conf/cav/BeutnerFFM23} have shown that $\sohyltl$ model-checking is $\Sigma_1^1$-hard. 
We improve the lower bound considerably, i.e., also to truth in third-order arithmetic, and show that this bound is tight.
This is the first upper bound on the problem's complexity. 

\begin{thm}
\label{thm_mccomplexity}
The $\sohyltl$ model-checking problem is polynomial-time equivalent to truth in third-order arithmetic. The lower bound already holds for $\univar$-free sentences.
\end{thm}

\begin{proof}
For the lower bound, we reduce truth in third-order arithmetic to the $\sohyltl$ model-checking problem: we present a polynomial-time translation from sentences~$\phi$ of third-order arithmetic to pairs $(\tsys, \phi')$ of a finite transition system~$\tsys$ and a $\sohyltl$ sentence~$\phi'$ such that $\natsstruct \models \phi$ if and only if $\tsys \models \phi'$.

In the proof of Theorem~\ref{thm_satcomplexity} we have, given a sentence~$\phi$ of third-order arithmetic, constructed a $\sohyltl$ sentence~$\phi'$ such that $\natsstruct\models \phi$ if and only if every set~$T$ of traces satisfies $\phi'$ (i.e., satisfaction is independent of the model). 
Thus, we obtain the lower bound by mapping $\phi$ to $\phi'$ and $\tsys^*$, where $\tsys^*$ is some fixed transition system.

For the upper bound, we reduce the $\sohyltl$ model-checking problem to truth in third-order arithmetic: we present a polynomial-time translation from pairs~$(\tsys, \phi)$ of a finite transition system and a $\sohyltl$ sentence~$\phi$ to sentences~$\phi'$ of third-order arithmetic such that $\tsys \models \phi$ if and only if $\natsstruct \models \phi'$.

In the proof of Theorem~\ref{thm_finsatcomplexity}, we have constructed, from a $\sohyltl$ sentence~$\phi$, a sentence~$\phi'$ of third-order arithmetic that expresses the existence of a finite transition system that satisfies $\phi$. 
We obtain the desired upper bound by modifying $\phi'$ to replace the existential quantification of the transition system by hardcoding $\tsys$ instead.
\end{proof}

Again, the lower bound proof can easily be extended to closed-world semantics, using the same arguments as in the case of satisfiability.

\begin{cor}
\label{cor_mccomplexity}
The $\sohyltl$ model-checking problem under closed-world semantics is polynomial-time equivalent to truth in third-order arithmetic.
\end{cor}

\section{\texorpdfstring{$\sohyltlfp$}{Second-order HyperLTL with Minimality/Maximality Constraints}} 
\label{sec_mm}

As we have seen, unrestricted second-order quantification makes $\sohyltl$ very expressive and therefore highly undecidable. 
But restricted forms of second-order quantification are sufficient for many application areas.
Beutner et al.~\cite{DBLP:conf/cav/BeutnerFFM23} introduced $\sohyltlfp$,
a fragment\footnote{In~\cite{DBLP:conf/cav/BeutnerFFM23} this fragment is termed $\sohyltlfpold$.}  of $\sohyltl$ in which second-order quantification ranges over smallest/largest sets that satisfy a given guard. For example, the formula $\exists (X,\smallest, \phi_1).\ \phi_2$ expresses that there is a set~$T$ of traces that satisfies both $\phi_1$ and $\phi_2$, and $T$ is \emph{a} smallest set that satisfies $\phi_1$ (i.e., $\phi_1$ is the guard). 
Note that the guards themselves may again contain (guarded) second-order quantifiers.
This fragment is expressive enough to express common knowledge, asynchronous hyperproperties, and causality in reactive systems~\cite{DBLP:conf/cav/BeutnerFFM23}, but it can also reason directly about maximal (or minimal) sets that satisfy some property. 
As a concrete example, consider the HyperLTL formula 
\[\varphi_{\gni} = \forall \pi_1. \forall \pi_2. \exists \pi. (\pi_1 =_{\text{high-in}} \pi) \wedge (\pi_2 =_{\text{low-in}} \pi) \] expressing generalized non-interference~\cite{ClarksonFKMRS14}. In case a system does not satisfy it, we might be interested in finding a maximal subset of the system that does satisfy $\varphi_{\gni}$, while preserving some correctness property $\varphi$. We can do so using the formula $\exists (X,\largest, \varphi_{\gni}).\ \phi$  with closed-world semantics (see below the formal definitions).

The formulas of $\sohyltlfp$ are given by the grammar
\begin{align*}
\phi  {}& \cceq {} \exists (X,\smalar,\phi).\ \phi \mid \forall (X,\smalar,\phi).\ \phi \mid \exists \pi \in X.\ \phi \mid \forall \pi \in X.\ \phi \mid \psi\\
\psi {}&  \cceq {} \proposition_\pi \mid \neg \psi \mid \psi \vee \psi \mid \X \psi \mid \psi \U \psi    
\end{align*}
where $\proposition$ ranges over $\ap$, $\pi$ ranges over~$\fovar$, $X$ ranges over $\sovar$, and $\smalar \in \set{\smallest, \largest}$, i.e., the only modification concerns the syntax of second-order quantification. 

We consider two fragments of $\sohyltlfp$ obtained by only allowing quantification over maximal sets and only allowing quantification over minimal sets, respectively:
\begin{itemize}
    \item $\sohyltlfpmax$ is the fragment of $\sohyltlfp$ obtained by disallowing second-order quantifiers of the form~$\exists (X,\smallest,\phi) $ and $ \forall (X,\smallest,\phi)$.
    \item $\sohyltlfpmin$ is the fragment of $\sohyltlfp$ obtained by disallowing second-order quantifiers of the form~$\exists (X,\largest,\phi) $ and $ \forall (X,\largest,\phi)$.
\end{itemize}

The semantics of $\sohyltlfp$ is similar to that of $\sohyltl$ but for the second-order quantifiers, for which we define (for $\smalar \in \set{\smallest,\largest}$):
\begin{itemize}
    \item $\Pi \models \exists (X,\smalar,\phi_1).\ \phi_2$ if there exists a set~$T \in \solutions(\Pi, (X,\smalar,\phi_1))$ such that $\Pi[X\mapsto T] \models \phi_2$.
    \item $\Pi \models \forall (X,\smalar,\phi_1).\ \phi_2$ if for all sets~$T \in \solutions(\Pi, (X,\smalar,\phi_1))$ we have $\Pi[X\mapsto T] \models \phi_2$.
\end{itemize}
Here, $\solutions(\Pi, (X,\smalar,\phi_1))$ is the set of all minimal/maximal models of the formula~$\phi_1$, which is defined as 
\begin{align*}
    {}&{}\!\!\!\solutions(\Pi, (X,\smallest,\phi_1))  = \set{T \subseteq (\pow{\ap})^\omega \mid \Pi[X\mapsto T] \models \phi_1\text{ and } 
    \Pi[X\mapsto T'] \not\models \phi_1 \text{ for all } T' \subsetneq T }
\end{align*}
and
\begin{align*}
    {}&{}\!\!\!\solutions(\Pi, (X,\largest,\phi_1))  = \set{T \subseteq (\pow{\ap})^\omega \mid \Pi[X\mapsto T] \models \phi_1 \text{ and }
     \Pi[X\mapsto T'] \not\models \phi_1 \text{ for all } T' \supsetneq T}.
\end{align*}
Note that $\solutions(\Pi, (X,\smalar,\phi_1))$ may be empty, may be a singleton, or may contain multiple sets, which then are pairwise incomparable.

Let us also define closed-world semantics for $\sohyltlfp$. Here, we again disallow the use of the variable~$\univar$ and change the semantics of set quantification to 
\begin{itemize}
    \item $\Pi \models_\cw \exists (X,\smalar,\phi_1).\ \phi_2$ if there exists a set~$T \in \solutions_\cw(\Pi, (X,\smalar,\phi_1))$ such that $\Pi[X\mapsto T] \models_\cw \phi_2$, and 
    \item $\Pi \models_\cw \forall (X,\smalar,\phi_1).\ \phi_2$ if for all sets~$T \in \solutions_\cw(\Pi, (X,\smalar,\phi_1))$ we have $\Pi[X\mapsto T] \models_\cw \phi_2$,
\end{itemize}
where $\solutions_\cw(\Pi, (X,\smallest,\phi_1))$ and $\solutions_\cw(\Pi, (X,\largest,\phi_1))$ are  defined as follows:
\begin{align*}
    \solutions_\cw(\Pi, (X,\smallest,\phi_1))  = \set{T \subseteq \Pi(\unidisvar) \mid{}&{} \Pi[X\mapsto T] \models_\cw \phi_1\\ 
    {}&{}\text{ and } \Pi[X\mapsto T'] \not\models_\cw \phi_1\text{ for all }T' \subsetneq T} \\
     \solutions_\cw(\Pi, (X,\largest,\phi_1))= \set{T \subseteq \Pi(\unidisvar) \mid{}&{}  \Pi[X\mapsto T] \models_\cw \phi_1\\ 
    {}&{}\text{ and }  \Pi[X\mapsto T'] \not\models_\cw \phi_1 \text{ for all } T \subsetneq   T' \subseteq\Pi(\unidisvar)  }.
\end{align*}
Note that $\solutions_\cw(\Pi, (X,\smalar,\phi_1))$ may still be empty, may be a singleton, or may contain multiple sets, but all sets in it are now incomparable subsets of $\Pi(\unidisvar)$.

A $\sohyltlfp$ formula is a sentence if it does not have any free variables except for $\univar$ and $ \unidisvar$ (also in the guards). 
Models are defined as for $\sohyltl$.

\begin{prop}[Proposition~1 of \cite{DBLP:conf/cav/BeutnerFFM23}]
\label{prop_fp2classical}
Every $\sohyltlfp$ sentence~$\phi$ can be translated in polynomial time (in $\size{\varphi}$) into a $\sohyltl$ sentence~$\phi'$ such that for all sets~$T$ of traces we have that $T \models \phi$  if and only if  $T \models \phi'$.\footnote{The polynomial-time claim is not made in \cite{DBLP:conf/cav/BeutnerFFM23}, but follows from the construction when using appropriate data structures for formulas.} 
\end{prop}

The same claim is also true for closed-world semantics, using the same proof.

\begin{rem}
\label{remark_fp2classical_cw}
Every $\sohyltlfp$ sentence~$\phi$ can be translated in polynomial time (in $\size{\varphi}$) into a $\sohyltl$ sentence~$\phi'$ such that for all sets~$T$ of traces we have that $T \models_\cw \phi$  if and only if  $T \models_\cw \phi'$.
\end{rem}

Thus, every complexity upper bound for $\sohyltl$ also holds for $\sohyltlfp$ and every lower bound for $\sohyltlfp$ also holds for $\sohyltl$.

\subsection{Satisfiability, Finite-State Satisfiability, and Model-Checking} 

In this subsection, we settle the complexity of satisfiability, finite-state satisfiability, and model-checking for the fragments~$\sohyltlfpmax$ and $\sohyltlfpmin$ using only second-order quantification over maximal respectively minimal sets satisfying a given guard.
We will show, for both semantics, that all three problems have the same complexity as the corresponding problems for full $\sohyltl$, i.e., they are equivalent to truth in third-order arithmetic.
Thus, contrary to the design goal of $\sohyltlfp$, it is in general not more feasible than full $\sohyltl$.

Due to Proposition~\ref{prop_fp2classical} and Remark~\ref{remark_fp2classical_cw}, the upper bounds already hold for full $\sohyltl$, hence the remainder of this subsection is concerned with lower bounds. 
We show that one can translate each $\sohyltl$ sentence~$\phi$ (over $\ap$) into a $\sohyltlfpmax$ sentence~$\phi^\largest$ and into a $\sohyltlfpmin$ sentence~$\phi^\smallest$ (both over some $\ap' \supsetneq \ap$) and each $T \subseteq (\pow{\ap})^\omega$ into a $T' \subseteq (\pow{\ap'})^\omega$ such that $T\models_\cw \phi$ if and only if $T'\models_\cw \phi^\largest$ and $T\models_\cw \phi$ if and only if $T'\models_\cw \phi^\smallest$.
This translation allows us to reduce $\sohyltl$ satisfiability, finite-state satisfiability, and model-checking to their $\sohyltlfpmax$ and $\sohyltlfpmin$ counterparts.
As closed-world semantics can be reduced to standard semantics, this suffices to prove the result for both semantics.

Intuitively, $\phi^\largest$ and $\phi^\smallest$ mimic the quantification over arbitrary sets in $\phi$ by quantification over maximal and minimal sets that satisfy a guard~$\phi_1$ that is only satisfied by uncountable sets.
In the proof of Theorem~\ref{thm_modelsizelowerbound}, we have constructed 
a $\sohyltl$ formula~$\phi_1$ with this property. 
Here, we show that similar formulas can also be written in $\sohyltlfpmax$ and $\sohyltlfpmin$.
With these formulas as guards (which use fresh propositions in $\ap' \setminus \ap$), we mimic arbitrary set quantification via quantification of sets of traces over $\ap'$ that are uncountable (enforced by the guards) and then consider their $\ap$-projections.
This approach works for all sets but the empty set, as projecting an uncountable set cannot result in the empty set. 
For this reason, we additionally mark some traces in the uncountable set and only project the marked ones, but discard the unmarked ones.
Thus, by marking no trace, the projection is the empty set.

\begin{thm}
\label{thm_hyltlmm_complexity}
\hfill
\begin{enumerate}
    \item $\sohyltlfpmax$ satisfiability, finite-state satisfiability, and model-checking (both under standard semantics and under closed-world semantics) are polynomial-time equivalent to truth in third-order arithmetic. The lower bounds for standard semantics already hold for $\univar$-free sentences.

    \item $\sohyltlfpmin$ satisfiability, finite-state satisfiability, and model-checking (both under standard semantics and under closed-world semantics) are polynomial-time equivalent to truth in third-order arithmetic. The lower bounds for standard semantics already hold for $\univar$-free sentences.

\end{enumerate}
\end{thm}

Recall that we only need to prove the lower bounds, as the upper bounds already hold for full $\sohyltl$.
Furthermore, as closed-world semantics can be reduced to standard semantics, we only need to prove the lower bounds for closed-world semantics.

We begin by constructing the desired guards that have only uncountable models.
As a first step, we modify the formulas constructed in the proof of Theorem~\ref{thm_modelsizelowerbound}: we show that $\sohyltlfpmax$ and $\sohyltlfpmin$ have formulas that require the interpretation of a free second-order variable to be uncountable.
To this end, fix $\ap_\allsets = \set{\posprop, \negprop, \setprop,\inprop}$ and consider the language
\begin{align*}
T_\allsets ={}&{} \set{ \set{\posprop}^\omega,\set{\negprop}^\omega} \cup\\
&{}\set{\set{\posprop}^n \set{\inprop,\posprop} \set{\posprop}^\omega \mid n\in\nats} \cup\\
&{}\set{\set{\negprop}^n \set{\inprop,\negprop} \set{\negprop}^\omega \mid n\in\nats}{} \cup\\
&{}\set{ (t(0) \cup \set{\setprop})(t(1) \cup \set{\setprop})(t(2) \cup \set{\setprop}) \cdots \mid t \in (\pow{\set{\inprop}})^\omega},
\end{align*}
which is an uncountable subset of $(\pow{\ap_\allsets})^\omega$.
Note that this definition differs from the one in the proof of Theorem~\ref{thm_modelsizelowerbound}, as we have added the two traces~$\set{\posprop}^\omega,\set{\negprop}^\omega$.
This allows us to construct a finite transition system that has exactly these traces (which is not possible without these two traces): Figure~\ref{fig_tallsets} depicts a transition system~$\tsys_\allsets$ satisfying $\traces(\tsys_\allsets) = T_\allsets$.

\begin{figure}[h]
    \centering

        \begin{tikzpicture}[thick]
        \def\d{1.25}
            \node[plainnode] (1p) at (0,0) {\small$\set{\posprop}$};
            \node[plainnode] (2p) at (3,0) {\small$\set{\inprop,\posprop}$};
            \node[plainnode] (3p) at (6,0) {\small$\set{\posprop}$};

            \path[->] 
            (-.85,0) edge (1p)
            (1p) edge[loop above] ()
            (1p) edge (2p)
            (3,.85) edge (2p)
            (2p) edge (3p) 
            (3p) edge[loop above] ()
            ;

            \node[plainnode] (1n) at (0,-\d) {\small$\set{\negprop}$};
            \node[plainnode] (2n) at (3,-\d) {\small$\set{\inprop,\negprop}$};
            \node[plainnode] (3n) at (6,-\d) {\small$\set{\negprop}$};

            \path[->] 
            (-.85,-\d) edge (1n)
            (1n) edge[loop below] ()
            (1n) edge (2n)
            (3,-\d-.85) edge (2n)
            (2n) edge (3n) 
            (3n) edge[loop below] ()
            ;

            \node[plainnode] (1s) at (9,0) {\small$\set{\inprop,\setprop}$};
            \node[plainnode] (2s) at (9,-\d) {\small$\set{\setprop}$};

            \path[->] 
            (8.15,-0) edge (1s)
            (8.15,-\d) edge (2s)
            (1s) edge[loop above] ()
            (1s) edge[bend left] (2s)
            (2s) edge[bend left] (1s)
            (2s) edge[loop below] ()
            ;

        \end{tikzpicture}
    
    \caption{The transition system $\tsys_\allsets$ with $\traces(\tsys_\allsets) = T_\allsets$.}
    \label{fig_tallsets}
\end{figure}

\begin{lem}\hfill
\label{lemma_uncountmodels}
\begin{enumerate}
    \item\label{lemma_uncountmodels_largest}
    There exists a $\sohyltlfpmax$ formula~$\phi_\allsets^\largest$ over $\ap_\allsets$ with a single free (second-order) variable~$Z$ such that $\Pi \models_\cw \varphi_\allsets^\largest$ if and only if the $\ap_\allsets$-projection of $\Pi(Z)$ is $T_\allsets$.

    \item\label{lemma_uncountmodels_smallest}
    There exists a $\sohyltlfpmin$ formula~$\phi_\allsets^\smallest$ over $\ap_\allsets$ with a single free (second-order) variable~$Z$ such that $\Pi \models_\cw \varphi_\allsets^\smallest$ if and only if the $\ap_\allsets$-projection of $\Pi(Z)$ is $T_\allsets$.
\end{enumerate}
\end{lem}

\begin{proof}
\ref{lemma_uncountmodels_largest}.)
Consider $\phi_\allsets^\largest = \phi_0 \wedge \cdots \wedge \phi_4$ where
\begin{itemize}
    \item $\phi_0 = \forall \pi \in Z.\ \bigvee_{\proposition \in \set{\posprop, \negprop,\setprop}} \G (\proposition_\pi \wedge \bigwedge_{\proposition' \in \set{\posprop, \negprop,\setprop} \setminus \set{\proposition}} \neg \proposition'_\pi )$ expresses that on each trace in the $\ap_\allsets$-projection of $\Pi(Z)$ exactly one of the propositions in $\set{\posprop, \negprop,\setprop}$ holds at each position and the other two at none. 
    In the following, we speak therefore about type~$\proposition$ traces for $\proposition \in \set{\posprop, \negprop,\setprop}$,

    \item $\phi_1 = \forall \pi \in Z.\  (\posprop_\pi \vee \negprop_\pi) \rightarrow ( (\G \neg \inprop_\pi) \vee (\neg \inprop_\pi \U ( \inprop_\pi \wedge \X \G \neg \inprop_\pi )))$
    expresses that $\inprop$ appears at most once on each type~$\proposition$ trace in the $\ap_\allsets$-projection of $\Pi(Z)$, for both $\proposition \in \set{\posprop,\negprop}$,

    \item $\phi_2 = \bigwedge_{\proposition \in \set{\posprop,\negprop}} (\exists \pi \in Z.\ \proposition_\pi\wedge \G\neg \inprop_\pi) \wedge (\exists\pi \in Z.\ \proposition_\pi \wedge \inprop_\pi)$ expresses that the type~$\proposition$ traces~$\set{\proposition}^\omega$ and $\set{\proposition}^0 \set{\inprop,\proposition} \set{\proposition}^\omega$ are in the $\ap_\allsets$-projection of $\Pi(Z)$, for both $\proposition \in \set{\posprop,\negprop}$, and

    \item $\phi_3 = \bigwedge_{\proposition \in \set{\posprop,\negprop}}\forall \pi \in Z.\ \exists\pi'\in Z.\ (\proposition_\pi \wedge \F \inprop_\pi) \rightarrow ( \proposition_{\pi'} \wedge  \F (\inprop_\pi \wedge \X \inprop_{\pi'}))$ expresses that for each type~$\proposition$ trace of the form~$\set{\proposition}^n \set{\inprop,\proposition} \set{\proposition}^\omega$ in the $\ap_\allsets$-projection of $\Pi(Z)$, the trace~$\set{\proposition}^{n+1} \set{\inprop,\proposition} \set{\proposition}^\omega$ is also in the $\ap_\allsets$-projection of $\Pi(Z)$, for both $\proposition \in \set{\posprop,\negprop}$.
    
\end{itemize}
Intuitively, the formulas~$\phi_0$ to $\phi_3$ are obtained from the analogous formulas in the proof of Theorem~\ref{thm_modelsizelowerbound} by replacing $\unidisvar$ by the free variable~$Z$ and by allowing for the two additional traces.

The conjunction of these first four formulas requires that the $\ap_\allsets$-projection of $\Pi(Z) $ (and thus, under closed-world semantics, also the model) contains at least the traces in
\begin{equation}
    \label{eq_nats}
\set{ \set{\posprop}^\omega,\set{\negprop}^\omega} \cup \set{\set{\posprop}^n \set{\inprop,\posprop} \set{\posprop}^\omega \mid n\in\nats} \cup \set{\set{\negprop}^n \set{\inprop,\negprop} \set{\negprop}^\omega \mid n\in\nats},
\end{equation}
and no other type~$\posprop$ or type~$\negprop$ traces.

We say that a set~$T$ of traces is contradiction-free if its $\ap_\allsets$-projection only contains traces of the form~$\set{\posprop}^n \set{\inprop,\posprop} \set{\posprop}^\omega$ or $\set{\negprop}^n \set{\inprop,\negprop} \set{\negprop}^\omega$ and if there is no $n \in \nats$ such that $\set{\posprop}^n \set{\inprop,\posprop} \set{\posprop}^\omega$ and $\set{\negprop}^n \set{\inprop,\negprop} \set{\negprop}^\omega$ are in the $\ap_\allsets$-projection of $T$.
A trace~$t$ is consistent with a contradiction-free~$T$ if the following two conditions are satisfied:
\begin{description}
    \item[(C1)] If $\set{\posprop}^n \set{\inprop,\posprop} \set{\posprop}^\omega $ is in the $\ap_\allsets$-projection of $T$ then $\inprop\in t(n)$.
    \item[(C2)] If $\set{\negprop}^n \set{\inprop,\negprop} \set{\negprop}^\omega $ is in the $\ap_\allsets$-projection of $T$ then $\inprop\notin t(n)$.
\end{description}
Note that $T$ does not necessarily specify the truth value of $\inprop$ in every position of $t$, i.e., in those positions~$n\in\nats$ where neither $\set{\posprop}^n \set{\inprop,\posprop} \set{\posprop}^\omega$ nor $\set{\negprop}^n \set{\inprop,\negprop} \set{\negprop}^\omega$ are in the $\ap_\allsets$-projection of $T$.
Nevertheless, due to (\ref{eq_nats}), for every trace~$t$ over $\set{\inprop}$ there exists a contradiction-free subset~$T$ of the $\ap_\allsets$-projection of $\Pi(Z)$ such that the $\set{\inprop}$-projection of every trace~$t'$ that is consistent with $T$ is equal to $t$. 
Here, we can, without loss of generality, restrict ourselves to maximal contradiction-free sets, i.e., sets that stop being contradiction-free if more traces are added (note that such sets do specify the truth value of $\inprop$ for every position of a consistent trace).

Thus, each of the uncountably many traces over $\set{\inprop}$ is induced by some maximal contradiction-free subset of the $\ap_\allsets$-projection of $\Pi(Z)$.
\begin{itemize}
\item Hence, we define $\phi_4$ as the formula
\begin{align*}
\forall {}&{}(X, \largest, \overbrace{\forall \pi \in X.\ (\pi\tracein Z \wedge (\posprop_\pi \vee \negprop_\pi)) \wedge \forall \pi \in X.\ \forall \pi' \in X.\ (\posprop_\pi \wedge \negprop_{\pi'}) \rightarrow \neg \F(\inprop_{\pi} \wedge \inprop_{\pi'}) }^{\text{$X$ is contradiction-free}}).\\
&{}\exists \pi'' \in Z.\ \forall \pi''' \in X.\ \setprop_{\pi''} \wedge  \underbrace{((\posprop_{\pi'''}\wedge \F\inprop_{\pi'''}) \rightarrow \F(\inprop_{\pi'''} \wedge \inprop_{\pi''}))}_{\text{(C1)}}
\wedge \\
&{}\phantom{\exists \pi'' \in Z.\ \forall \pi''' \in X.\ \setprop_{\pi''} \wedge  }\underbrace{((\negprop_{\pi'''} \wedge \F\inprop_{\pi'''}) \rightarrow \F(\inprop_{\pi'''} \wedge \neg\inprop_{\pi''}))}_{\text{(C2)}},
\end{align*} 
expressing that for every maximal contradiction-free set of traces~$T$, there exists a type~$\setprop$ trace~$t''$ in the $\ap_\allsets$-projection of $\Pi(Z)$ that is consistent with $T$.
\end{itemize}

Thus, if the $\ap_\allsets$-projection of $\Pi(Z)$ is $T_\allsets$, then $\Pi\models_\cw\phi_\allsets^\largest$.
Dually, we can conclude that $T_\allsets$ must be a subset of the $\ap_\allsets$-projection of $\Pi(Z)$ whenever $\Pi\models_\cw \phi_\allsets^\largest$, and that the $\ap_\allsets$-projection of $\Pi(Z)$ cannot contain other traces (due to $\phi_0$ and $\phi_1$).
Hence, $\Pi \models_\cw \varphi^\largest$ if and only if the $\ap_\allsets$-projection of $\Pi(Z)$ is $T_\allsets$.

\ref{lemma_uncountmodels_smallest}.)
Here, we will follow a similar approach, but have to overcome one obstacle: there exists a unique minimal contradiction-free set, i.e., the empty set.
Hence, we cannot naively replace quantification over maximal contradiction-free sets in $\phi_4$ above by quantification over minimal contradiction-free sets. 
Instead, we will quantify over minimal contradiction-free sets that have, for each $n$, either the trace~$\set{\posprop}^n \set{\inprop,\posprop} \set{\posprop}^\omega$ or the trace~$\set{\negprop}^n \set{\inprop,\negprop} \set{\negprop}^\omega$ in their $\ap_\allsets$-projection. 
The minimal sets satisfying this constraint are still \emph{rich} enough to enforce every possible trace over $\set{\inprop}$.

Formally, we replace $\phi_4$ by $\phi_4'$, which is defined as 
\begin{align*}
\forall {}&{}(X, \smallest, \phi_\guard) .\\
&{}\exists \pi'' \in Z.\ \forall \pi''' \in X.\ \setprop_{\pi''} \wedge  {(\posprop_{\pi'''} \rightarrow \F(\inprop_{\pi'''} \wedge \inprop_{\pi''}))}
\wedge 
{(\negprop_{\pi'''} \rightarrow \F(\inprop_{\pi'''} \wedge \neg\inprop_{\pi''}))},
\end{align*} 
where $\phi_\guard$ is the formula
\[
 \phi_\complete  \wedge \forall \pi \in X.\ (\pi\tracein Z \wedge (\posprop_\pi \vee \negprop_\pi)) \wedge \forall \pi \in X.\ \forall \pi' \in X.\ (\posprop_\pi \wedge \negprop_{\pi'}) \rightarrow \neg \F(\inprop_{\pi} \wedge \inprop_{\pi'})
\]
and $\phi_\complete$ is the formula
\[
\exists \pi \in X.\ (\inprop_\pi \wedge (\posprop_\pi \vee \negprop_\pi)) \wedge \forall \pi \in X.\ \exists \pi' \in X.\ (\posprop_\pi \vee \negprop_\pi) \rightarrow ((\posprop_{\pi'} \vee \negprop_{\pi'}) \wedge \F( \inprop_\pi \wedge \X\inprop_{\pi'})),
\]
i.e., the only changes are the change of the polarity from $\largest$ to $\smallest$ and the addition of $\phi_\complete$ in the guard.
This ensures that $\phi_\allsets^\smallest = \phi_0 \wedge \cdots \wedge \phi_3 \wedge \phi_4'$ has the desired properties.
\end{proof}

Now, we can begin with the translation of full $\sohyltl$ into $\sohyltlfpmax$ and $\sohyltlfpmin$.
Let us fix a $\sohyltl$ sentence~$\varphi$ over a set~$\ap$ of propositions that is, without loss of generality, disjoint from $\ap_\allsets$.
Hence, satisfaction of $\phi$ only depends on the projection of traces to $\ap$, i.e., if $T_0$ and $T_1$ have the same $\ap$-projection, then $T_0 \models_\cw \phi$ if and only if $T_1 \models_\cw\phi$. 
We assume without loss of generality that each variable is quantified at most once in $\phi$ and that $\univar$ and $\unidisvar$ are not bound by a quantifier in $\phi$, which can always be achieved by renaming variables. 
Let $X_0, \ldots, X_{k-1}$ be the second-order variables quantified in $\phi$. 
We can assume $k>0$, as there is nothing to show otherwise: Without second-order quantifiers, $\phi$ is already a $\sohyltlfpmax$ sentence and a $\sohyltlfpmin$ sentence.
To simplify our notation, we define $[k] = \set{0,1,\ldots, k-1}$.
For each $i \in [k]$, we introduce a fresh proposition~$\marker_{i}$ so that we can define $\ap'$ as the pairwise disjoint union of $\ap$, $\set{\marker_{i} \mid i \in [k]}$, and $\ap_\allsets$.

Let $i \in [k]$ and consider the formula
\begin{align*}
\phi_\parti^{i} ={}&{} \forall \pi \in X_i.\ ((\G ({\marker_{i}})_\pi) \vee (\G \neg ({\marker_{i}})_\pi) \wedge \bigwedge_{i' \in [k]\setminus\set{i}}\G\neg(\marker_{i'})_\pi ) \wedge \\
{}&{}\forall \pi \in X_i.\ \forall \pi' \in X_i.\ (\equals{\pi}{\pi'}{\ap_\allsets }) \rightarrow (\equals{\pi}{\pi'}{\ap'}).   
\end{align*}
Note that $\phi_\parti^{i}$ has a single free variable, i.e., $X_i$, and that it is both a formula of $\sohyltlfpmin$ and $\sohyltlfpmax$.
Intuitively, $\phi_\parti^{i}$ expresses that each trace in $\Pi(X_i)$ is either marked by $\marker_{i}$ (if ${\marker_{i}}$ holds at every position) or it is not marked (if ${\marker_{i}}$ holds at no position), it is not marked by any other $\marker_{i'}$, and that there may \emph{not} be two distinct traces~$t \neq t'$ in $\Pi(X_i)$ that have the same $\ap_\allsets$-projection.
The former condition means that $\Pi(X_i)$ is partitioned into two (possibly empty) parts, the subset of marked traces and the subset of unmarked traces; the latter condition implies that each trace in $\Pi(X_i)$ is uniquely identified by its $\ap_\allsets$-projection. Said differently, if two traces in $\Pi(X_i)$ have the same $\ap_\allsets$-projection, then they also have the same marking. 

Fix some $i \in [k]$ and some $\smalar \in \set{\largest,\smallest}$, and define 
\[\phi_\guard^{i,\smalar} = \phi_\allsets^{\smalar}[Z/X_i] \wedge \phi_\parti^{i}, \]
where $\phi_\allsets^{\smalar}[Z/X_i]$ is the formula obtained from $\phi_\allsets^{\smalar}$ by replacing each occurrence of $Z$ by $X_i$.
The only free variable of $\phi_\guard^{i,\smalar}$ is $X_i$.

The following lemma shows that no strict superset of a set satisfying $\phi_\guard^{i,\largest}$ satisfies $\phi_\guard^{i,\largest}$ and no strict subset of a set satisfying $\phi_\guard^{i,\smallest}$ satisfies $\phi_\guard^{i,\smallest}$, i.e., sets satisfying these guards are in the corresponding solution sets.

\begin{lem}
\label{lemma_guardcorrectness}
Let $\Pi_0'$ and $\Pi_1'$ be two variable assignments with $\Pi_0'(X_i) \subseteq (\pow{\ap'})^\omega$ and  $\Pi_1'(X_i) \subseteq (\pow{\ap'})^\omega$.
\begin{enumerate}
    
    \item\label{lemma_guardcorrectness_largest}
    If $\Pi_0' \models_\cw \phi_\guard^{i,\largest}$ and $\Pi_1'(X_i) \supsetneq \Pi_0'(X_i)$, then $\Pi_1' \not\models_\cw \phi_\guard^{i,\largest}$.

    \item\label{lemma_guardcorrectness_smallest}
    If $\Pi_0' \models_\cw \phi_\guard^{i,\smallest}$ and $\Pi_1'(X_i) \subsetneq \Pi_0'(X_i)$, then $\Pi_1' \not\models_\cw \phi_\guard^{i,\smallest}$.

\end{enumerate}
\end{lem}

\begin{proof}
\ref{lemma_guardcorrectness_largest}.) 
Towards a contradiction, assume we have $\Pi_0' \models_\cw \phi_\guard^{i,\largest}$ and $\Pi_1'(X_i) \supsetneq \Pi_0'(X_i)$, but also $\Pi_1' \models_\cw \phi_\guard^{i,\largest}$.
Then, there exists a trace~$t_1 \in \Pi_1'(X_i) \setminus \Pi_0'(X_i)$. 
Due to $\Pi_0'\models_\cw \phi_\allsets^\largest$ and Lemma~\ref{lemma_uncountmodels}.\ref{lemma_uncountmodels_largest}, there exists a trace~$t_0 \in \Pi_0'(X_i)$ that has the same $\ap_\allsets$-projection as $t_1$. 
Hence, $\varphi_\parti^{i}$ implies that $t_0$ and $t_1$ have the same $\ap'$-projection, i.e., they are the same trace (here we use $\Pi_0'(X_i) \subseteq (\pow{\ap'})^\omega$ and  $\Pi_1'(X_i) \subseteq (\pow{\ap'})^\omega$).
Thus, $t_1 = t_0$ is in $\Pi_0'(X_i)$, i.e., we have derived a contradiction.

\ref{lemma_guardcorrectness_smallest}.)
The argument here is similar, we just have to swap the roles of the two sets~$\Pi_0'(X_i)$ and 
$\Pi_1'(X_i)$.

Towards a contradiction, assume we have $\Pi_0' \models_\cw \phi_\guard^{i,\smallest}$ and $\Pi_1'(X_i) \subsetneq \Pi_0'(X_i)$, but also $\Pi_1' \models_\cw \phi_\guard^{i,\smallest}$.
Then, there exists a trace~$t_0 \in \Pi_0'(X_i) \setminus \Pi_1'(X_i)$. 
Due to $\Pi_1'\models_\cw \phi_\allsets^\smallest$ and Lemma~\ref{lemma_uncountmodels}.\ref{lemma_uncountmodels_smallest}, there exists a trace~$t_1 \in \Pi_1'(X_i)$ that has the same $\ap_\allsets$-projection as $t_0$. 
Hence, $\varphi_\parti^{i}$ implies that $t_1$ and $t_0$ have the same $\ap'$-projection, i.e., they are the same trace.
Thus, $t_0 = t_1$ is in $\Pi_1'(X_i)$, i.e., we have derived a contradiction.
\end{proof}

Recall that our goal is to show that quantification over minimal/maximal subsets of $(\pow{\ap'})^\omega$ satisfying $\phi_\guard^{i,\smalar}$ mimics quantification over subsets of $(\pow{\ap})^\omega$.
Now, we can make this statement more formal.
Let $T' \subseteq (\pow{\ap'})^\omega$. We define $\enc_i(T')$ as the set
\[
\set{
t \in (\pow{\ap})^\omega \mid t \text{ is the $\ap$-projection of some $t' \in T'$ whose $\set{\marker_{i}}$-projection is $\set{\marker_{i}}^\omega$}
}.
\]
Now, every $\enc_i(T')$ is, by definition, a subset of $(\pow{\ap})^\omega$. 
Our next result shows that, conversely, every subset of $(\pow{\ap})^\omega$ 
can be obtained as an encoding of some $T'$ that additionally satisfies the guard formulas.

\begin{lem}
\label{lemma_encodingrichness}
Let $T \subseteq (\pow{\ap})^\omega$, $i \in [k]$, and $\smalar \in \set{\largest,\smallest}$.
There exists a $T' \subseteq (\pow{\ap'})^\omega$ such that
\begin{itemize}
    \item $T = \enc_i(T')$, and 
    \item for all $\Pi$ with $\Pi(X_i) = T'$, we have $\Pi \models_\cw \phi_\guard^{i,\smalar}$.
\end{itemize}
\end{lem}

\begin{proof}
Fix a bijection~$f \colon (\pow{\ap})^\omega \rightarrow T_\allsets$, which can, e.g., be obtained by applying the Schröder-Bernstein theorem.
Now, we define 
\[T' = \set{t\merge f(t) \merge\set{\marker_{i}}^\omega \mid t \in T} \cup \set{t\merge f(t) \mid t \in (\pow{\ap})^\omega \setminus T}. \]
By definition, we have $\enc_i(T') = T$, satisfying the first requirement on $T'$.
Furthermore, the $\ap_\allsets$-projection of $T'$ is $T_\allsets$, each trace in $T'$ is either marked by $\marker_{i}$ at every position or at none, the other markers do not appear in traces in $T'$, and, due to $f$ being a bijection, there are no two traces in $T'$ with the same $\ap_\allsets$-projection.
Hence, due to Lemma~\ref{lemma_uncountmodels}, $T'$ does indeed satisfy $\phi_\guard^{i,\smalar}$.
\end{proof}

Now, we explain how to mimic quantification over arbitrary sets via quantification over maximal or minimal sets satisfying the guards we have constructed above.
Let $\smalar \in \set{\smallest, \largest}$ and let $\varphi^{\smalar}$ be the $\sohyltlfpsmalar$-sentence obtained from $\phi$ by inductively replacing 
\begin{itemize}
    \item each $\exists X_i.\ \psi$ by $\exists (X_i, \smalar, \varphi_\guard^{i,\smalar}).\ \psi$, 
    \item each $\forall X_i.\ \psi$ by $\forall (X_i, \smalar, \varphi_\guard^{i,\smalar}).\ \psi$, 
    \item each $\exists \pi \in X_i.\ \psi$ by $\exists \pi \in X_i.\ (\marker_{i})_{\pi} \wedge \psi$, and 
    \item each $\forall \pi \in X_i.\ \psi$ by $\forall \pi \in X_i.\ (\marker_{i})_{\pi} \rightarrow \psi$.
\end{itemize}

We show that $\phi$ and $\phi^{\smalar}$ are in some sense equivalent. 
Obviously, they are not equivalent in the sense that they have the same models, as $\phi^{\smalar}$ uses in general the additional propositions in $\ap' \setminus \ap$, which are not used by $\phi$.
But, by \emph{extending} a model of $\phi$ we obtain a model of $\phi^{\smalar}$.
Similarly, by \emph{ignoring} the additional propositions (i.e., the inverse of the extension) in a model of $\phi^{\smalar}$, we obtain a model of $\phi$.

As the model is captured by the assignment to the variable~$\unidisvar$, this means the interpretation of $\unidisvar$ when evaluating $\phi$ and the  interpretation of $\unidisvar$ when evaluating $\phi^{\smalar}$ have to satisfy the extension property to be defined.
However, we show correctness of our translation by induction.
Hence, we need to strengthen the induction hypothesis and require that the variable assignments for $\varphi$ and $\phi^{\smalar}$ satisfy the extension property for all free variables.

Formally, let $\Pi$ and $\Pi'$ be two variable assignments such that 
\begin{itemize}
    \item $\Pi(\pi) \in (\pow{\ap})^\omega$ and $\Pi(X) \subseteq (\pow{\ap})^\omega$ for all $\pi$ and $X$ in the domain of $\Pi$, and 
    \item $\Pi'(\pi) \in (\pow{\ap'})^\omega$ and $\Pi'(X) \subseteq (\pow{\ap'})^\omega$ for all $\pi$ and $X$ in the domain of $\Pi'$.
\end{itemize}
We say that $\Pi'$ extends $\Pi$ if they have the same domain (which must contain $\unidisvar$, but not $\univar$ as this variable may not be used under closed-world semantics) and we have that
\begin{itemize}
    \item for all $\pi$ in the domain, the $\ap$-projection of $\Pi'(\pi)$ is $\Pi(\pi)$, 
    \item for all $X_i$ for $i \in [k]$, in the domain, $\Pi(X_i) = \enc_i(\Pi'(X_i))$, and
    \item $\Pi'(\unidisvar) = \extend(\Pi(\unidisvar))$ where
    \[\extend(T) = \set{t \merge t'{} \merge \set{\marker_{i}} \mid t \in T, t' \in T_\allsets, \text{ and } i \in [k]} \cup \set{t \merge t' \mid t \in T \text{ and } t' \in T_\allsets},\]
    which implies that the $\ap$-projection of $\Pi'(\unidisvar)$ is $\Pi(\unidisvar)$.
\end{itemize}

Note that if $T$ is a subset of $\Pi(\unidisvar)$, then there exists a subset~$T'$ of $\Pi'(\unidisvar)$ such that $\enc_i(T') = T$, i.e., $\Pi'(\unidisvar)$ contains enough traces to mimic the quantification of subsets of $\Pi(\unidisvar)$ using our encoding.
Furthermore, if $T' \subseteq (\pow{\ap'})^\omega$  has the form~$T' = \extend(T)$ for some $T \subseteq (\pow{\ap})^\omega$, then this $T$ is unique.

The following lemma states that our translation of $\sohyltl$ into $\sohyltlfpmax$ and $\sohyltlfpmin$ is correct. 

\begin{lem}
\label{lemma_mmcorrectness}
Let $\Pi'$ extend $\Pi$. Then, $\Pi \models_\cw \phi$ if and only if $\Pi' \models_\cw \varphi^{\smalar}$.
\end{lem}

\begin{proof}
By induction over the subformulas~$\psi$ of $\phi$.

First, let us consider the case of atomic propositions, i.e., $\psi = \proposition_\pi$ for some $\proposition \in \ap$ and some $\pi$ in the domains of $\Pi$ and $\Pi'$. We have 
\[
\Pi\models_\cw\proposition_\pi \Leftrightarrow 
\proposition \in \Pi(\pi)(0) \Leftrightarrow 
\proposition \in \Pi'(\pi)(0) \Leftrightarrow 
\Pi'\models_\cw\proposition_\pi,
\]
as $\Pi(\pi)$ and $\Pi'(\pi)$ have the same $\ap$-projection due to $\Pi'$ extending $\Pi$.

The cases of Boolean and temporal operators are straightforward applications of the induction hypothesis. So, it only remains to consider the quantifiers. 

Let $\psi = \exists X_i.\ \psi_0$ for some $i \in [k]$. Then, we have $\varphi^{\smalar} = \exists (X_i, \smalar, \varphi_\guard^{i,\smalar}).\ \psi_0^{\smalar}$.
By induction hypothesis, we have for all $T \subseteq (\pow{\ap})^\omega$ and $T' \subseteq (\pow{\ap'})^\omega$ with $\enc_i(T') = T$, the equivalence
\[\Pi[X_i \mapsto T] \models_\cw \psi_0 \Leftrightarrow \Pi'[X_i \mapsto T' ] \models_\cw \psi_0^{\smalar}.\]
Now, we have 
\begin{align*}
&{}\Pi\models_\cw \psi\\
 \Leftrightarrow{} &{}\text{there exists a $T \subseteq \Pi(\unidisvar)$ such that } \Pi[X_i \mapsto T] \models_\cw \psi_0 \\
\xLeftrightarrow{\ast}{}&{}\text{there exists a $T' \subseteq \Pi'(\unidisvar)$ s.t. } \Pi'[X_i \mapsto T'] \models_\cw \psi_0^{\smalar} \text{ and  } T' \in \solutions_\cw(\Pi,( X_i, \smalar, \varphi_\guard^{i,\smalar} ))
\\
\Leftrightarrow{}&{}\Pi'\models_\cw \psi^{\smalar}.
\end{align*}
Here, the equivalence marked with $\ast$ is obtained by applying the induction hypothesis:
\begin{itemize}
    \item For the left-to-right direction, given $T$, we pick $T'$ such that $\enc_i(T') = T$, which is always possible due to Lemma~\ref{lemma_encodingrichness}. Furthermore, $T'$ is in $\solutions_\cw(\Pi,( X_i, \smalar, \varphi_\guard^{i,\smalar} )$, due to Lemma~\ref{lemma_encodingrichness} and Lemma~\ref{lemma_guardcorrectness}.
    \item For the right-to-left direction, given $T'$, we pick $T = \enc_i(T')$.
\end{itemize}
Thus, in both directions, $\Pi'[X_i \mapsto T']$ extends $\Pi[X_i \mapsto T] $, i.e., the induction hypothesis is indeed applicable.
The argument for universal set quantification is dual.

So, it remains to consider trace quantification. 
First, let $\psi= \exists \pi \in X_i.\ \psi_0$ for some $i \in [k]$, which implies
$\psi^{\smalar} = \exists \pi \in X_i.\ (\marker_{i})_{\pi} \wedge \psi'$.
Now,
\begin{align*}
{}&{}\Pi \models_\cw \psi \\
\Leftrightarrow{}&{}
\text{there exists a $t \in \Pi(X_i)$ such that } \Pi[\pi \mapsto t] \models_\cw \psi_0 \\
\xLeftrightarrow{\ast}{}&{}
\text{there exists a $t' \in \Pi'(X_i)$ such that } \Pi'[\pi \mapsto t'] \models_\cw (\marker_{i})_\pi \wedge\psi_0^{\smalar} \\
\Leftrightarrow{}&{}
\Pi' \models_\cw  \psi^{\smalar}.
\end{align*}
Here, the equivalence marked by $\ast$ follows from the induction hypothesis: 
\begin{itemize}
    \item For the left-to-right direction, given $t$, we can pick a $t'$ with the desired properties as $\Pi'$ extends $\Pi$, which implies that $\Pi(X_i) = \enc_i(\Pi'(X_i))$, which in turn implies that $\Pi'(X_i)$ contains a trace~$t'$ marked by $\marker_{i}$ whose $\ap$-projection is $t$.
    \item For the right-to-left direction, given $t'$, we pick $t$ as the $\ap$-projection of $t'$, which is in $\Pi(X_i)$, as $\Pi(X_i) = \enc_i(\Pi'(X_i))$ and $t'$ is marked by $\marker_i$ due to $\Pi'[\pi \mapsto t']  \models_\cw (\marker_{i})_\pi$.
\end{itemize}
Thus, in both directions, $\Pi'[\pi \mapsto t']$ extends $\Pi[\pi \mapsto t] $, i.e., the induction hypothesis is indeed applicable.
The argument for universal trace quantification is again dual. 

So, it only remains to consider trace quantification over $\unidisvar$.
Consider $\psi = \exists \pi \in \unidisvar.\ \psi_0$, which implies $\psi^{\smalar} = \exists \pi \in \unidisvar.\ \psi_0^{\smalar}$.
Then, we have
\begin{align*}
{}&{}\Pi\models_\cw \psi \\
\Leftrightarrow{}&{} \text{there exists a $t \in \Pi(\unidisvar)$ such that $\Pi[\pi \mapsto t] \models_\cw \psi_0$} \\
\xLeftrightarrow{\ast}{}&{} \text{there exists a $t' \in \Pi'(\unidisvar)$ such that $\Pi'[\pi \mapsto t'] \models_\cw \psi_0^{\smalar}$}\\
\Leftrightarrow{}&{} \Pi'\models_\cw \psi^{\smalar}
,    
\end{align*}
where the equivalence marked with $\ast$ follows from the induction hypothesis, which is applicable as the $\ap$-projection of $\Pi'(\unidisvar)$ is equal to $\Pi(\unidisvar)$, which implies that we can choose $t'$ from 
$t$ for the left-to-right direction (and choose $t$ from 
$t'$ for the right-to-left direction) such that $\Pi'[\pi \mapsto t']$ extends $\Pi[\pi \mapsto t]$.
Again, the argument for universal quantification is dual.
\end{proof}

Recall that $\phi$ is satisfied by some set~$T \subseteq (\pow{\ap})^\omega$ of traces (under closed-world semantics) if $\Pi_\emptyset[\unidisvar\mapsto T] \models_\cw \varphi$.
Similarly, $\phi^\smalar$ is satisfied by some set~$T' \subseteq (\pow{\ap'})^\omega$ of traces (under closed-world semantics) if $\Pi_\emptyset[\unidisvar\mapsto T'] \models_\cw \varphi^\smalar$.
The following corollary of Lemma~\ref{lemma_mmcorrectness} holds as $\Pi_\emptyset[\unidisvar\mapsto \extend(T)]$ extends $\Pi_\emptyset[\unidisvar\mapsto T]$.

\begin{cor}
\label{coro_mmcorrectness}
Let $T\subseteq (\pow{\ap})^\omega$, and $\smalar \in \set{\largest,\smallest}$. Then, $T \models_\cw \phi$ if and only if $\extend(T) \models_\cw \phi^\smalar$.
\end{cor}

So, to conclude the construction, we need to ensure that models of $\phi^\smalar$ have the form~$\extend(T)$ for some $T \subseteq (\pow{\ap})^\omega$.
For the two satisfiability problems, we do so using a sentence that only has such models, while for the model-checking problem, we can directly transform the transition system we are checking so that it satisfies this property.

First, we construct $\phi_\extend^{\smalar}$ for $\smalar \in \set{\largest,\smallest}$ such that $\Pi'\models_\cw \phi_\extend^{\smalar}$ if and only if $\Pi'(\unidisvar) = \extend(T)$ for some $T \subseteq (\pow{\ap})^\omega$:
\begin{align*}
\phi_\extend^{\smalar} ={}&{} (\forall \pi \in \unidisvar.\ (\proposition_\pi \wedge \neg \proposition_\pi)) \vee \\
{}&{}\phi_\allsets^{\smalar}[Z/\unidisvar] \wedge \forall \pi \in \unidisvar.\ \forall \pi' \in \unidisvar.\ \\
{}&{}\quad
\left(
\bigwedge_{i \in [k]} \exists \pi'' \in \unidisvar.\ 
\equals{\pi}{\pi''}{\ap_\allsets} \wedge 
\equals{\pi'}{\pi''}{\ap} \wedge 
\G (\marker_i)_{\pi''} \wedge \bigwedge_{i' \in [k]\setminus\set{i}} \G\neg(\marker_{i'})_{\pi''}
\right) \wedge\\
{}&{}\quad
\left( 
 \exists \pi'' \in \unidisvar.\ 
\equals{\pi}{\pi''}{\ap_\allsets} \wedge 
\equals{\pi'}{\pi''}{\ap} \wedge 
\bigwedge_{i \in [k]} \G\neg(\marker_{i})_{\pi''}
\right)
\end{align*}
where $\phi_\allsets^{\smalar}[Z/\unidisvar]$ is the formula obtained from $\phi_\allsets^{\smalar}$ by replacing each occurrence of $Z$ by $\unidisvar$.

The first disjunct of $\phi_\extend^{\smalar}$ is for the special case of the empty $T$ (where $\proposition$ is an arbitrary proposition) while the second one expresses intuitively that for each $t \in T_\allsets$ (bound to $\pi$) and each $t'$ in the $\ap$-projection of the interpretation of $\unidisvar$ (bound to $\pi'$), the traces~$t\merge t'$ and the traces~$t\merge t'{}\merge \set{\marker_i}^\omega$, for each $i \in [k]$, are in the interpretation of $\unidisvar$.

Finally, let us consider the transformation of transition systems: Given a transition system~$\tsys = (V, E, I, \lambda)$ we construct a transition system~$\extend(\tsys)$ satisfying $\traces(\extend(\tsys)) = \extend(\traces(\tsys))$.
Recall the transition system~$\tsys_\allsets$ depicted in Figure~\ref{fig_tallsets}, which satisfies $\traces(\tsys_\allsets) = T_\allsets$. 
Let $\tsys_\allsets = (V_a, E_a, I_a,\lambda_a)$.
We define $\extend(\tsys) = (V', E', I', \lambda')$ where 
\begin{itemize}
    \item $V' = V \times V_a \times [k] \times\set{0,1} $,
    \item $E' = \set{((v,v_a,i,b),(v',v_a',i,b)) \mid (v,v') \in E, (v_a, v_a') \in E_a, i \in [k] \text{ and } b \in \set{0,1} }$,
    \item $I' = I \times I_a \times [k] \times\set{0,1} $, and
    \item $\lambda'(v,v_a,i,b) = 
    \begin{cases}
        \lambda(v) \cup \lambda_a(v_a) \cup \set{\marker_i} &\text{ if } b=1,\\
        \lambda(v) \cup \lambda_a(v_a)  &\text{ if } b=0.
    \end{cases}$
\end{itemize}
Note that we indeed have $
\traces(\extend(\tsys)) = \extend(\traces(\tsys))$.
Hence, if $\Pi(\unidisvar) = \traces(\tsys)$ and $\Pi'(\unidisvar) = \traces(\extend(\tsys))$, then $\Pi$ and $\Pi'$ satisfy the requirement spelled out in Corollary~\ref{coro_mmcorrectness}.

\begin{lem}
\label{lem_so2sofp}
Let $\phi$ be a $\sohyltl$ sentence and $\phi^{\smalar}$ as defined above (for some $\smalar \in \set{\largest,\smallest}$), and let $\tsys$ be a transition system.
\begin{enumerate}
    \item\label{lem_so2sofp_sat} $\phi$ is satisfiable under closed-world semantics if and only if $\phi^{\smalar} \wedge \phi_\extend^{\smalar}$ is satisfiable under closed-world semantics.
    \item\label{lem_so2sofp_fssat} $\phi$ is finite-state satisfiable under closed-world semantics if and only if $\phi^{\smalar} \wedge \phi_\extend^{\smalar}$ is finite-state satisfiable under closed-world semantics. 
    \item\label{lem_so2sofp_mc} $\tsys \models_\cw \phi$ if and only if $\extend(\tsys) \models_\cw \phi^{\smalar}$.  
\end{enumerate}
\end{lem}

\begin{proof}
\ref{lem_so2sofp_sat}.) 
We have
\begin{align*}
{}&{} \text{$\phi$ is satisfiable under closed-world semantics}\\
\Leftrightarrow {}&{} \text{there exists a $T \subseteq (\pow{\ap})^\omega$ such that $ T \models_\cw \phi$}\\
\xLeftrightarrow{\ast} {}&{} \text{there exists a $T' \subseteq (\pow{\ap'})^\omega$ such that $ T' \models_\cw \phi^{\smalar}\wedge \phi_\extend^{\smalar}$}\\
\Leftrightarrow {}&{} \text{$\phi^{\smalar} \wedge \phi_\extend^{\smalar}$ is satisfiable under closed-world semantics.}
\end{align*}
Here, the equivalence marked with $\ast$ follows from Corollary~\ref{coro_mmcorrectness}:
\begin{itemize}
    \item For the left-to-right direction, we pick $T' = \extend(T)$. 
    \item For the right-to left direction, we pick $T$ to be the unique subset of $(\pow{\ap})^\omega$ such that $T' = \extend(T)$, which is well-defined as $T'$ satisfies $\phi_\extend^{\smalar}$. 
\end{itemize}

\ref{lem_so2sofp_fssat}.)
We have
\begin{align*}
{}&{} \text{$\phi$ is finite-state satisfiable under closed-world semantics}\\
\Leftrightarrow {}&{} \text{there exists a transition system~$\tsys$ over $\pow{\ap}$ such that $\traces(\tsys) \models_\cw \phi$}\\
\xLeftrightarrow{\ast} {}&{} \text{there exists a transition system~$\tsys'$ over $\pow{\ap'}$ such that $\traces(\tsys') \models_\cw \phi^{\smalar}\wedge \phi_\extend^{\smalar}$}\\
\Leftrightarrow {}&{} \text{$\phi^{\smalar} \wedge \phi_\extend^{\smalar}$ is finite-state satisfiable under closed-world semantics.}
\end{align*}
Here, the equivalence marked with $\ast$ follows from Corollary~\ref{coro_mmcorrectness}:
\begin{itemize}
    \item For the left-to-right direction, we pick $\tsys' = \extend(\tsys)$, which satisfies $\traces(\extend(\tsys)) = \extend(\traces(\tsys))$.
    \item For the right-to left direction, we pick $\tsys$ to be the transition system obtained from $\tsys'$ by removing all propositions in $\ap' \setminus \ap$ from the vertex labels.
    This implies that $\traces(\tsys')$ is equal to $\extend(\traces(\tsys))$, as $\traces(\tsys')$ satisfies $\phi_\extend^{\smalar}$.
\end{itemize}

\ref{lem_so2sofp_mc}.) 
Due to Corollary~\ref{coro_mmcorrectness} and $\traces(\extend(\tsys)) = \extend(\traces(\tsys))$, we have $\tsys \models_\cw \phi$ if and only if $\extend(\tsys)\models_\cw \phi^\smalar$.
\end{proof}

Finally, Theorem~\ref{thm_hyltlmm_complexity} is now a direct consequence of Lemma~\ref{lem_so2sofp}, the fact that closed-world semantics can be reduced to standard semantics, and the fact that all three problems for $\sohyltl$ (under closed-world semantics) are equivalent to truth in third-order arithmetic~\cite{sohypercomplexity}.

\section{\texorpdfstring{The Least Fixed Point Fragment of $\sohyltlfp$}{Second-order HyperLTL with Least Fixed Points}}
\label{sec_lfp}

We have seen that even restricting second-order quantification to smallest/largest sets that satisfy a guard formula is essentially as expressive as full $\sohyltl$, and thus as difficult.
However, Beutner et al.~\cite{DBLP:conf/cav/BeutnerFFM23} note that applications like common knowledge and asynchronous hyperproperties do not even require quantification over smallest/largest sets satisfying arbitrary guards, they \myquot{only} require quantification over least fixed points of $\hyltl$ definable functions.
This finally yields a fragment with (considerably) lower complexity: 
we show that under standard semantics, satisfiability is $\Sigma_1^2$-complete while finite-state satisfiability and model-checking are both equivalent to truth in second-order arithmetic.
On the other hand, under closed-world semantics, satisfiability is $\Sigma_1^1$-complete while the complexity of the other two problems is not affected, i.e., finite-state satisfiability and model-checking are still equivalent to truth in second-order arithmetic.
Note that this means that $\lfpsohyltlfp$ satisfiability under closed-world semantics is as hard as  $\hyltl$ satisfiability.

The intuitive reason that satisfiability under closed-world semantics is simpler is that we show that every satisfiable sentence has a countable model. 
Thus, not all subsets of natural numbers can be encoded in such a model, which restricts the complexity of the problem.
On the other hand, in, e.g., model-checking, the model is part of the input, i.e., one has to potentially evaluate a formula over a set of traces that encodes all subsets of natural numbers. 
This makes the problem much harder. 

A $\sohyltlfp$ sentence using only minimality constraints has the form 
\[
\phi = 
\gamma_1.\ \quant_1(Y_1, \smallest, \phi_1^\con).\ 
\gamma_2.\ \quant_2(Y_2, \smallest, \phi_2^\con).\ 
\ldots 
\gamma_k.\ \quant_k(Y_k, \smallest, \phi_k^\con).\
\gamma_{k+1}.\
\psi
\]
satisfying the following properties: 
\begin{itemize}
    \item Each $\gamma_j$ is a block~$
    \gamma_j = \quant_{\ell_{j-1}+1} \pi_{\ell_{j-1}+1} \in X_{\ell_{j-1}+1} \cdots \quant_{\ell_{j}} \pi_{\ell_{j}} \in X_{\ell_{j}} 
    $
    of trace quantifiers (with $\ell_0 = 0$). As $\phi$ is a sentence, we must have $\set{X_{\ell_j+1}, \ldots, X_{\ell_{j}}} \subseteq \set{\univar,\unidisvar, Y_1, \ldots, Y_{j-1}}$.

    \item The free variables of $\phi_j^\con$ are among the trace variables quantified in the $\gamma_{j'}$ with $j' \le j$ and $\univar,\unidisvar, Y_1, \ldots, Y_j$.
    
    \item $\psi$ is a quantifier-free formula. Again, as $\phi$ is a sentence, the free variables of $\psi$ are among the trace variables quantified in the $\gamma_j$.
\end{itemize}

Now, $\phi$ is an $\lfpsohyltlfp$ sentence\footnote{Our definition here differs slightly from the one of  \cite{DBLP:conf/cav/BeutnerFFM23} in that we allow to express the existence of some traces in the fixed point (via the subformulas~$\overdot{\pi}_i\tracein Y_j$). All examples and applications of~\cite{DBLP:conf/cav/BeutnerFFM23} are also of this form.}, if additionally each $\phi_j^\con$ has the form 
    \[ \phi_j^\con = \overdot{\pi}_1\tracein Y_j \wedge \cdots \wedge \overdot{\pi}_n \tracein Y_j \wedge   \forall \overdotdot{\pi}_1 \in Z_1.\ \ldots \forall \overdotdot{\pi}_{n'} \in Z_{n'}.\ \psi^\step_j \rightarrow \overdotdot{\pi}_m \tracein Y_j \]
 for some $n\ge 0$, $n' \ge 1$, where $1 \le m \le {n'}$, and where we have
\begin{itemize}
    \item $\set{\overdot{\pi}_1,\ldots, \overdot{\pi}_n} \subseteq \set{\pi_1, \ldots, \pi_{\ell_j} }$,
    \item $\set{Z_1, \ldots, Z_{n'}} \subseteq \set{\univar, \unidisvar, Y_1, \ldots, Y_j}$, and
    \item $\psi^\step_j$ is quantifier-free with free variables among $\overdotdot{\pi}_1, \ldots, \overdotdot{\pi}_{n'}, \pi_1, \ldots, \pi_{\ell_j}$.
\end{itemize} 
As always, $\phi_j^\con$ can be brought into the required prenex normal form.

Let us give some intuition for the definition. To this end, fix some~$j \in \set{1,2,\ldots, k}$ and a variable assignment~$\Pi$ whose domain contains at least all variables quantified before $Y_j$, i.e., all $Y_{j'}$ and all variables in the $\gamma_{j'}$ for $j' \le j$, as well as $\univar$ and $\unidisvar$. 
Then, 
\[ \phi_j^\con= \overdot{\pi}_1\in Y_j \wedge \cdots \wedge \overdot{\pi}_n \in Y_j \wedge \left(  \forall \overdotdot{\pi}_1 \in Z_1.\ \ldots \forall \overdotdot{\pi}_{n'} \in Z_{n'}.\ \psi^\step_j \rightarrow \overdotdot{\pi}_m \tracein Y_j\right) \]
induces the monotonic function~$f_{\Pi,j} \colon \pow{(\pow{\ap})^\omega} \rightarrow \pow{(\pow{\ap})^\omega}$ defined as
\begin{multline*}
f_{\Pi,j}(S) = S \cup \set{\Pi(\overdot{\pi}_1), \ldots, \Pi(\overdot{\pi}_n)} \cup  \set{\Pi'(\overdotdot{\pi}_m) \mid \Pi' = \Pi[\overdotdot{\pi}_1 \mapsto t_1, \ldots, \overdotdot{\pi}_{n'} \mapsto t_{n'}]\\
 \text{ for } t_i \in \Pi(Z_i) \text{ if } Z_i \neq Y_j \text{ and }t_i \in S \text{ if } Z_i = Y_j \text{ s.t. }\Pi' \models \psi^\step_j}.    
\end{multline*}
We define $S_0 = \emptyset$, $S_{\ell+1} =  f_{\Pi,j}(S_\ell)$, and 
\[\lfp(\Pi,j) = \bigcup\nolimits_{\ell\in\nats} S_\ell,\] 
which is the least fixed point of $f_{\Pi,j}$. Due to the minimality constraint on $Y_j$ in $\phi$, $\lfp(\Pi,j)$ is the unique set in $\solutions(\Pi, (Y_j,\smallest,\phi^\con_j))$. Hence, an induction shows that $\lfp(\Pi,j)$ only depends on the values~$\Pi(\pi)$ for trace variables~$\pi$ quantified before $Y_j$ as well as the values~$\Pi(\unidisvar)$ and $\Pi(\univar)$, but not on the values~$\Pi(Y_{j'})$ for $j' < j$ (as they are unique).

Thus, as $\solutions(\Pi, (Y_j,\smallest,\phi^\con_j))$ is a singleton, it is irrelevant whether $\quant_j$ is an existential or a universal quantifier. Instead of interpreting second-order quantification as existential or universal, here one should understand it as a deterministic least fixed point computation: choices for the trace variables and the two distinguished second-order variables uniquely determine the set of traces that a second-order quantifier assigns to a second-order variable.

\begin{rem}
\label{remark_lfpcw}
Note that the traces that are added to a fixed point assigned to $Y_j$ either come from another $Y_{j'}$ with $j' < j$, from the model (via $\unidisvar$), or from the set of all traces (via $\univar$). 
Thus, for $\univar$-free formulas, all second-order quantifiers range over (unique) subsets of the model, i.e., there is no need for an explicit definition of closed-world semantics. The analogue of closed-world semantics for $\lfpsohyltlfp$ is to restrict oneself to $\univar$-free sentences.
\end{rem}

In the remainder of this section, we study the complexity of $\lfpsohyltlfp$. 

\subsection{\texorpdfstring{Satisfiability for the \boldmath$\univar$-free Fragment}{Satisfiability for the Xa-free Fragment}}
\label{subsec_lfpsatcw}
For satisfiability of $\univar$-free sentences, the key step is again to study the size of models of satisfiable sentences. 
For $\univar$-free $\lfpsohyltlfp$, as for $\hyltl$, we are able to show that each satisfiable sentence has a countable model.
The following result is proven by generalizing the proof for the analogous result for $\hyltl$~\cite{FZ17} showing that every model~$T$ of a $\hyltl$ sentence~$\phi$ contains a countable $R\subseteq T$ that is closed under the application of Skolem functions. This implies that $R$ is also a model of $\phi$.

\begin{lem}
\label{lemma_lfpcwmodelsize}
Every satisfiable $\univar$-free $\lfpsohyltlfp$ sentence has a countable model.
\end{lem}

A crucial step to prove this upper bound is to characterize membership of a trace in some least fixed point by the existence of a finite trace-labeled tree. These trees model the stage-wise definition of the semantics of second-order quantification in $\lfpsohyltlfp$.

To formalize this intuition, we need to introduce some additional notation and concepts. 
Let \[
\phi = 
\gamma_1.\ \quant_1(Y_1, \smallest, \phi_1^\con).\ 
\gamma_2.\ \quant_2(Y_2, \smallest, \phi_2^\con).\ 
\ldots 
\gamma_k.\ \quant_k(Y_k, \smallest, \phi_k^\con).\
\gamma_{k+1}.\
\psi
\]
be an $\univar$-free $\lfpsohyltlfp$ sentence.
We assume without loss of generality that each trace variable is quantified at most once in $\varphi$, which can always be achieved by renaming variables. This implies that for each trace variable~$\pi$ quantified in some $\gamma_j$, there is a unique second-order variable~$X_\pi$ such that $\pi$ ranges over $X_\pi$.
Furthermore, we assume that each $Y_j$ is different from $\unidisvar$ and $\univar$, which can again be achieved by renaming variables, if necessary.

Recall that set quantifiers in $\lfpsohyltlfp$ \myquot{range} over unique least fixed points that do not depend on the choices for the other set variables, but only on the choice for the universe of discourse and the choices for the trace variables quantified before the set variable.

We say that a variable assignment~$\Pi$ is $(\phi,j)$-sufficient, for some $j \in \set{1, \ldots, k}$, if its domain contains $\unidisvar,Y_1, \ldots, Y_{j-1}$, if its domain contains the variables quantified in the $\gamma_{j'}$ with $j' \le j$, and if $\Pi(\pi) \in \Pi(X_\pi)$ for all $\pi$ quantified in some $\gamma_{j'}$ with $j' \le j$.
A $(\phi,j)$-sufficient assignment~$\Pi$ contains all information necessary to compute~$\lfp(\Pi, j)$.
Note that we do not have any requirements on the $\Pi(Y_{j'})$ being the right sets, i.e., being the least fixed points.
This will be taken into account later. 

Our first results states that the computation of least fixed points is monotonic in the model we are starting with (given by the interpretation of $\unidisvar$) and the values for the previously quantified $Y_{j'}$.

\begin{lem}
\label{lemma:lfpmono}
Let $\Pi$ and $\Pi'$ be $(\phi, j)$-sufficient assignments with $\Pi'(\unidisvar) \subseteq \Pi(\unidisvar)$, $\Pi'(Y_{j'}) \subseteq \Pi(Y_{j'})$ for all $1 \le j' < j$, and $\Pi(\pi) = \Pi'(\pi)$ for all trace variables appearing in the $\gamma_{j'}$ with $1 \le j' \le j$.
Then, $\lfp(\Pi',j) \subseteq \lfp(\Pi,j)$.
\end{lem}

\begin{proof}
An induction shows that $S'_\ell \subseteq S_\ell$, where $\lfp(\Pi, j) = \bigcup_{\ell\ge 0} S_\ell$ and $\lfp(\Pi', j) = \bigcup_{\ell\ge 0} S'_\ell$ are the stages of the fixed points.
\end{proof}

Now, let $\Pi$ be a variable assignment whose domain contains at least $\unidisvar$ and all trace variables in the $\gamma_j$ for $1 \le j \le k$.
We define $\Pi_0 = \Pi$ and $\Pi_j = \Pi_{j-1}[Y_j \mapsto \lfp(\Pi_{j-1},j)]$, i.e., we iteratively compute and assign the least fixed points to the set variables in $\phi$.
Note that fixed points are uniquely determined by just the values~$\Pi(\unidisvar)$ and $\Pi(\pi)$ for $\pi$ in the $\gamma_j$ with $1\le j \le k$, but we still need the previous fixed points to determine the $j$-th one. 
Now, we define $\itereval(\Pi) = \Pi_k$, which coincides with $\Pi$ on all variables but the $Y_j$, which are mapped to the least fixed points induced by $\Pi$.

\begin{rem}
\label{remark:lfpmonoiter}
Note that each $\Pi_{j-1}$ with $1 \le j \le k$ is $(\phi,j)$-sufficient.
Hence, if we have variable assignments~$\Pi$ and $\Pi'$ whose domains contain at least $\unidisvar$ and all trace variables in the $\gamma_j$ for $1 \le j \le k$ and such that $\Pi'(\unidisvar) \subseteq \Pi(\unidisvar)$ and $\Pi(\pi) = \Pi'(\pi)$ for all trace variables appearing in the $\gamma_{j}$ with $1 \le j \le k$, then Lemma~\ref{lemma:lfpmono} implies that we have $\itereval(\Pi')(Y_j) \subseteq \itereval(\Pi)(Y_j)$ for all $j$.
\end{rem}

Next, we define a tree structure that witnesses the membership of traces in least fixed points assigned to the variables~$Y_j$, which make the inductive construction of the stages of the least fixed points explicit.
Intuitively, consider the formula
\begin{equation}
\label{eq_phicon}
\phi_j^\con = \overdot{\pi}_1\tracein Y_j \wedge \cdots \wedge \overdot{\pi}_n \tracein Y_j \wedge   \forall \overdotdot{\pi}_1 \in Z_1.\ \ldots \forall \overdotdot{\pi}_{n'} \in Z_{n'}.\ \psi^\step_j \rightarrow \overdotdot{\pi}_m \tracein Y_j.    
\end{equation}
inducing the unique least fixed point that $Y_j$ ranges over. 
It expresses that a trace~$t$ is in the fixed point either because it is of the form~$\Pi(\overdot{\pi}_i)$ for some $i \in \set{1,\ldots, n}$ where $\overdot{\pi}_i$ is a trace variable quantified before the quantification of $Y_j$, or $t$ is in the fixed point because there are traces~$t_1,\ldots, t_{n'}$ such that assigning them to the $\overdotdot{\pi}_i$ satisfies $\psi^\step_j$ and $t = t_m$. 
Thus, the traces~$t_1,\ldots, t_{n'}$ witness that $t$ is in the fixed point. 
However, each $t_i$ must be selected from $\Pi(Z_i)$, which, if $Z_i = Y_{j'}$ for some $j'$, again needs a witness.
Thus, a witness is in general a tree labeled by traces and indices from $\set{1, \ldots, k}$ denoting the fixed point variable the tree proves membership in.
Note that $\psi^\step_j$ may contain free variables that are quantified in the $\gamma_{j'}$ for $1 \le j' \le j$. 
The membership of such variables~$\pi$ in $\Pi(X_\pi)$ will not be witnessed by this witness tree, as their values are already determined by $\Pi$.
Instead, they will be taken care of by their own witness trees.

Formally, let us fix a $j^* \in \set{1,\ldots,k}$, a $(\phi,j^*)$-sufficient assignment~$\Pi$, and a trace~$t^*$. 
A $\Pi$-witness tree for $(t^*,j^*)$ (which witnesses~$t^* \in\lfp(\Pi,j^*)$) is an ordered finite tree whose vertices are labeled by pairs~$(t,j)$ where $t$ is a trace from $\Pi(\unidisvar)$ and where $j $ is in $\set{1,\ldots, k} \cup \set{\bot}$ such that the following are satisfied:
\begin{itemize}
    \item The root is labeled by $(t^*, j^*)$.
        \item If a vertex is labeled with $(t,\bot)$ for some trace~$t$, then it must be a leaf and $t \in \Pi(\unidisvar)$.
        
        \item Let $(t,j)$ be the label of some vertex~$v$ with $j\in\set{1,\ldots,k}$ and let $\phi_j^\con$ as in (\ref{eq_phicon}).
If $v$ is a leaf, then we must have $t = \Pi(\overdot{\pi}_i)$ for some $i \in \set{1, \ldots, n}$.
If $v$ is an internal vertex, then it must have $n'$ successors labeled by $(t_1,j_1),\ldots,(t_{n'},j_{n'})$ (in that order) such that $\Pi[\overdotdot{\pi}_1 \mapsto t_1, \ldots, \overdotdot{\pi}_{n'} \mapsto t_{n'}] \models \psi^\step_j $, $t = t_m$, and
such that the following holds for all $i \in \set{1, \ldots, n'}$:
\begin{itemize}
    \item If $Z_i = \unidisvar$, then the $i$-th successor of $v$ is a leaf, $t_i \in \Pi(\unidisvar)$, and $j_i = \bot$.
    \item If $Z_i = Y_{j'}$ for some $j'$ (which satisfies $1 \le j' \le j$ by definition of $\lfpsohyltlfp$), then we must have that $ j_i = j'$ and that the subtree rooted at the $i$-th successor of $v$ is a $\Pi$-witness tree for $(t_i,j_i)$.
\end{itemize}
\end{itemize}

Given a $\Pi$-witness tree~$\tree$, let $\traces(\tree)$ denote the set of traces labeling the vertices of $\tree$, which is a finite subset of $\Pi(\unidisvar)$.

\begin{lem}
\label{lemma:witnesstreeprops}
Let $\Pi$ be a $(\phi,j)$-sufficient assignment with $\Pi(Y_{j'}) = \itereval(\Pi)(Y_{j'})$ for all $j' < j$.
\begin{enumerate}
    \item \label{lemma:witnesstreeprops_correctness}
    We have $t \in \lfp(\Pi, j)$ if and only if there is a $\Pi$-witness tree for $(t,j)$.

    \item\label{lemma:witnesstreeprops_restr}
    Let $\Pi'$ be another $(\phi, j)$-sufficient assignment with $\Pi'(Y_{j'}) = \itereval(\Pi')(Y_{j'})$ for all $j' < j$ and with $\Pi(\pi) = \Pi'(\pi)$ for all trace variables appearing in the $\gamma_{j'}$ with $j' \le j$.
    If $\Pi'(\unidisvar) \supseteq \traces(\tree)$ for a $\Pi$-witness tree~$\tree$ for $(t,j)$,
    then $\tree$ is also a $\Pi'$-witness tree for $(t,j)$.

\end{enumerate}
\end{lem}

\begin{proof}
\ref{lemma:witnesstreeprops_correctness}. 
Let $\lfp(\Pi, j) = \bigcup_{\ell \ge 0} S_\ell$. The direction from left to right is proven by an induction over $(j,\ell)$ (with lexicographic ordering) constructing for each trace~$t$ in $\lfp(\Pi,j)$ a $\Pi$-witness tree for $(t,j)$.
So, let $t \in \lfp(\Pi, j)$. As we have $S_0 = \emptyset$, there is a minimal $\ell$ such that $t \in S_\ell \setminus S_{\ell-1}$, which satisfies $\ell \ge 1$.

If $\ell = 1$, then there are two cases how $t$ was added to $S_1$:
\begin{itemize}
    \item Either we have $t = \Pi(\overdot{\pi}_i)$ for some $1 \le i \le n$. Then, a tree consisting of a single vertex labeled by $(t,j)$ is a $\Pi$-witness tree for $(t,j)$.
    \item Otherwise, we must have $Z_i \neq Y_j$ for all $1 \le i \le n'$ and there must be traces~$t_1, \ldots, t_{n'}$ such that each $t_i$ is in $\Pi(Z_i)$, $t = t_m$, and 
    \[
    (((\Pi[\overdotdot{\pi}_1 \mapsto t_1])[\overdotdot{\pi}_2 \mapsto t_2]) \cdots )[\overdotdot{\pi}_{n'} \mapsto t_{n'}]
    \models \psi_j^\step,\]
    i.e., we update all variables of the form~$\overdotdot{\pi}_{j'}$ with $t_{j'}$ in $\Pi$ before evaluating $\psi_j^\step$.
    
    If $Z_i = \unidisvar$, then let $\tree_i$ be a tree consisting of a single vertex labeled by $(t_i,\bot)$,
    if $Z_i = Y_{j'}$ for some $j' < j$ (this is the only other possibility), then let $\tree_i$ be a $\Pi$-witness tree for $(t_i,j')$, which exists by induction hypothesis.
    Then,
\begin{center}
\begin{tikzpicture}[thick]
\node (root) at (0,0) {$(t,j)$};

\draw[rounded corners] (-3,-1) -- (-4,-2) -- (-2,-2) -- cycle;
\draw[rounded corners] (-.5,-1) -- (-1.5,-2) -- (.5,-2) -- cycle;
\draw[rounded corners] (3,-1) -- (4,-2) -- (2,-2) -- cycle;
\node at (1.25,-1.6) {$\cdots$};
\node at (-3,-1.6) {$\tree_1$};
\node at (-.5,-1.6) {$\tree_2$};
\node at (3,-1.6) {$\tree_{n'}$};

\path[-stealth]
(root) edge (-3,-1)
(root) edge (-.5,-1)
(root) edge (3,-1);

\end{tikzpicture}
\end{center}
    is a $\Pi$-witness tree for $(t,j)$.
\end{itemize}

Now, consider the case $\ell > 1$. Here, we must have that there is at least one $1 \le i \le n'$ such that $Z_i = Y_j$. Again, there must be traces~$t_1, \ldots, t_{n'}$ such that each $t_i$ is in $\Pi(Z_i)$ if $Z_i \neq Y_j$, $t_i \in S_{\ell-1}$ if $Z_i = Y_j$, $t = t_m$, and 
    \[
        (((\Pi[\overdotdot{\pi}_1 \mapsto t_1])[\overdotdot{\pi}_2 \mapsto t_2]) \cdots )[\overdotdot{\pi}_{n'} \mapsto t_{n'}]
    \models \psi_j^\step
    \]
For $i$ with $Z_i = \unidisvar$ or $Z_i = Y_{j'}$ with $j' < j$, we define $\tree_i$ as in the second case of the induction start. 
For $Z_i = Y_{j}$ we define $\tree_i$ to be a $\Pi$-witness tree for $(t_i,j)$, which exists by induction hypothesis applied to $(j, \ell-1)$.
Then,
\begin{center}
\begin{tikzpicture}[thick]
\node (root) at (0,0) {$(t,j)$};

\draw[rounded corners] (-3,-1) -- (-4,-2) -- (-2,-2) -- cycle;
\draw[rounded corners] (-.5,-1) -- (-1.5,-2) -- (.5,-2) -- cycle;
\draw[rounded corners] (3,-1) -- (4,-2) -- (2,-2) -- cycle;
\node at (1.25,-1.6) {$\cdots$};
\node at (-3,-1.6) {$\tree_1$};
\node at (-.5,-1.6) {$\tree_2$};
\node at (3,-1.6) {$\tree_{n'}$};

\path[-stealth]
(root) edge (-3,-1)
(root) edge (-.5,-1)
(root) edge (3,-1);

\end{tikzpicture}
\end{center}
    is a $\Pi$-witness tree for $(t,j)$.

The direction from right to left is proven by an induction over $(j,h)$ (with lexicographic ordering) showing that for each height~$h$ $\Pi$-witness tree for some $(t,j)$ we have $t \in S_h \subseteq \lfp(\Pi, j)$, where $\lfp(\Pi,j)$ is again $\bigcup_{\ell \ge 0}S_\ell$.
Here, we measure the height of a tree whose root is labeled by $(t,j)$ as the length (in the number of vertices) of a longest path starting at the root that visits only vertices that have a $j$ in the second component of their label.

So, let $\tree$ be a $\Pi$-witness tree for some $(t,j)$.
If $\tree$ has height~$1$, then there are two cases:
\begin{itemize}
    \item Either, $t = \Pi(\overdot{\pi}_i)$ for some $1 \le i \le n$. Then, $t \in S_1 \subseteq \lfp(\pi,j)$.
    \item Otherwise, we must have $Z_i \neq Y_j$ for all $1 \le i \le n'$ and there must be traces~$t_1, \ldots, t_{n'}$ such that each $t_i$ is in $\Pi(Z_i)$, $t = t_m$, and 
    \[    (((\Pi[\overdotdot{\pi}_1 \mapsto t_1])[\overdotdot{\pi}_2 \mapsto t_2]) \cdots )[\overdotdot{\pi}_{n'} \mapsto t_{n'}]
    \models \psi_j^\step.\]
    
    Let $\tree_i$ be the subtree of $\tree$ rooted at the $i$-th successor of the root, which is labeled by $(t_i,j')$ if $Z_i = Y_{j'}$ for some $j' < j$ or by $(t_i, \bot)$, if $Z_i = \unidisvar$.
    In the latter case, we have $t_i \in \Pi(\unidisvar)$ by the requirement that a $\Pi$-witness tree is only labeled by traces from $\Pi(\unidisvar)$.
    In the former case, $\tree_i$ is by definition a $\Pi$-witness tree for $(t_i,j')$.
    Hence, the induction hypothesis yields $t_i \in \lfp(\Pi, j') = \Pi(Y_{j'})$.

    Thus, we can apply the definition of $S_1 = f_{\Pi,j}(\emptyset)$ and obtain $t \in S_1 = S_h$.
\end{itemize}

Now, assume $\tree$ has height~$h > 1$. Then, there must be at least one $1 \le i \le n'$ such that 
$Z_i = Y_j$. Again, there must be traces~$t_1, \ldots, t_{n'}$ such that each $t_i$ is in $\Pi(Z_i)$ if $Z_i \neq Y_j$, $t_i \in S_{\ell-1}$ if $Z_i = Y_j$, $t = t_m$, and 
    \[    (((\Pi[\overdotdot{\pi}_1 \mapsto t_1])[\overdotdot{\pi}_2 \mapsto t_2]) \cdots )[\overdotdot{\pi}_{n'} \mapsto t_{n'}]
    \models \psi_j^\step.\]
For $i$ with $Z_i = Y_{j'}$ for some $j < j'$ and for $i$ with $Z_i = \unidisvar$ we can argue as in the second case of the induction start that we have $t_i \in \Pi(Z_i)$.
So, consider the remaining $i$, which satisfy $Z_i = Y_j$. 
By definition, the subtree~$\tree_i$ rooted at the $i$-th successor of the root is a $\Pi$-witness tree for $(t_i,j)$ of height at most $h-1$. 
Hence, the induction hypothesis applied to $(j,h-1)$ yields $t_i \in S_{h-1}$. 
    Thus, we can apply the definition of $S_h = f_{\Pi,j}(S_{h-1})$ and obtain $t \in S_h$.

\ref{lemma:witnesstreeprops_restr}.
This can be shown by induction over the height of witness trees exploiting the fact that $\tree$ being a $\Pi^*$-witness tree only depends on $\Pi^*(\unidisvar)$ containing $\traces(\tree)$ and the values~$\Pi^*(\pi)$ for trace variables appearing in the $\gamma_{j'}$ with $j' \le j$.
\end{proof}

Now we have collected all the prerequisites to prove Lemma~\ref{lemma_lfpcwmodelsize} which states that every satisfiable $\univar$-free $\lfpsohyltlfp$-sentence has a countable model.

\begin{proof}[Proof of Lemma~\ref{lemma_lfpcwmodelsize}]
Let \[
\phi = 
\gamma_1.\ \quant_1(Y_1, \smallest, \phi_1^\con).\ 
\gamma_2.\ \quant_2(Y_2, \smallest, \phi_2^\con).\ 
\ldots 
\gamma_k.\ \quant_k(Y_k, \smallest, \phi_k^\con).\
\gamma_{k+1}.\
\psi
\]
be a satisfiable $\univar$-free $\lfpsohyltlfp$ sentence.
Again, we assume without loss of generality that each trace variable is quantified at most once in $\varphi$, which implies that for each trace variable~$\pi$ quantified in some $\gamma_j$, there is a unique second-order variable~$X_\pi$ such that $\pi$ ranges over $X_\pi$.
Furthermore, we assume that each $Y_j$ is different from $\unidisvar$ and $\univar$, which can again be achieved by renaming variables, if necessary.

As $\varphi$ is satisfiable, there exists a set~$T$ of traces such that $T \models\phi$. 
In the following, we assume that $T$ is uncountable (and therefore nonempty), as the desired result otherwise holds trivially.
We prove that there is a countable~$R \subseteq T$ with $R\models\varphi$.
Intuitively, we show that the smallest set~$R$ that is closed under the application of the Skolem functions and that contains the traces labeling witness trees (for the fixed points computed w.r.t.\ $T$) for the traces in $R$ has the desired properties.

So, let $E$ be the set of existentially quantified trace variables in the $\gamma_j$ and let, for each $\pi \in E$, $k_\pi$ denote the number of trace variables that are universally quantified before $\pi$ (in any block~$\gamma_j$).
As $T$ satisfies $\phi$, there is a Skolem function~$f_{\pi} \colon T^{k_\pi} \rightarrow T$ for each  $\pi \in E$, i.e., functions witnessing $T\models \varphi$ in the following sense: $\Pi\models\psi$ for every variable assignment~$\Pi$ where 
\begin{itemize}
    \item $\unidisvar$ is mapped to $T$,
    \item each second-order variable~$Y_j$ is mapped to $\itereval(\Pi)(Y_j)$,
    \item for each universally quantified trace variable~$\pi$ that appears in a $\gamma_j$ we have $\Pi(\pi) \in \Pi(X_\pi)$, and
    \item for each existentially quantified trace variable~$\pi$ that appears in a $\gamma_j$, we have $\Pi(\pi) = f_{\pi}(\Pi(\pi_1), \ldots, \Pi(\pi_{k_\pi}))$, which has to be in $\Pi(X_\pi)$. Here, $\pi_1, \ldots, \pi_{k_\pi}$ are the universally quantified trace variables  that appear in some $\gamma_j$ before $\pi$.
\end{itemize}
We fix such Skolem functions for the rest of the proof. 

Given a set~$T' \subseteq T$, we define 
\[f_{\pi}(T') = \set{f_\pi(t_1, \ldots, t_{k_{\pi}}) \mid t_1, \ldots, t_{k_\pi} \in T' }.\]
If $T'$ is finite, then $f_{\pi}(T')$ is also finite. 

For each $(\phi,j)$-sufficient assignment~$\Pi$ with $\Pi(\unidisvar) = T$ and each $t \in \lfp(\Pi,j)$, we fix a $\Pi$-witness tree~$\tree(\Pi, t, j)$ for $(t,j)$. 

Now, we fix some arbitrary trace~$t_0 \in T$ and define $R_0 = \set{t_0}$,
\[
R_{\ell+1} = R_\ell \cup 
\bigcup_{\pi \in E} f_\pi(R_\ell) 
\]
for even $\ell \ge 0$ as well as
\[
R_{\ell+1} = R_\ell \cup \bigcup_{\Pi}\bigcup_{j=1}^k \bigcup_{t \in \Pi(Y_j) \cap R_\ell} \traces(\tree(\Pi, t, j)),
\]
for odd $\ell \ge 0$, where $\Pi$ ranges over all assignments that satisfy the following requirements:
\begin{enumerate}
    \item\label{admiss:domain} The domain of $\Pi$ contains exactly $\unidisvar$, the set variables~$ Y_1, \ldots, Y_k$ appearing in $\varphi$, and the trace variables quantified in the $\gamma_j$.
    \item\label{admiss:disc} $\Pi(\unidisvar) = T$.
    \item\label{admiss:set} $\Pi(Y_j) = \itereval(\Pi)(Y_j)$ for all $Y_j$.
    \item\label{admiss:univ} $\Pi(\pi) \in \Pi(X_\pi) \cap R_\ell$ for all $\pi$ that are universally quantified in some $\gamma_j$. 
    \item\label{admiss:ex} $\Pi(\pi) = f_\pi(\Pi(\pi_1), \ldots, \Pi(\pi_{k_\pi})) \in \Pi(X_\pi)$, where $\pi_1, \ldots, \pi_{k_\pi}$ are again the universally quantified trace variables that appear in some $\gamma_j$ before $\pi$.    
\end{enumerate}
We call a $\Pi$ with this property admissible at $\ell$.

We first show that each $R_\ell$ is a finite subset of $T$. This is obviously true for $\ell = 0$ and if it is true for even $\ell$, then it is also true for $\ell+1$, as applying the finitely many Skolem functions to tuples of traces from $R_\ell$ (of which there are, by induction hypothesis, only finitely many) adds only finitely many traces from $T$. 
So, let us consider an odd $\ell$ and consider the $\Pi$ that are admissible at $\ell$: The values~$\Pi(\unidisvar)$, $\Pi(Y_j)$ for all $j$, and $\Pi(\pi)$ for $\pi$ existentially quantified in some $\gamma_j$ are all determined, so the only degree of freedom are the values~$\Pi(\pi)$ for $\pi$ universally quantified in some $\gamma_j$.
However, as these must come from $R_\ell$, which is finite by induction hypothesis, there are only finitely many $\Pi$ that are admissible at $\ell$. Now, we can conclude that $R_{\ell+1}$ is obtained by adding the traces of finitely many witness trees to $R_\ell$, each of which contains finitely many traces from $T$. Thus, $R_\ell$ is a finite subset of $T$ as well.

Now, define $R = \bigcup_{\ell \ge 0} R_\ell$, which is a countable subset of $T$, as it is a countable union of finite subsets of $T$.
Recall that the Skolem functions~$f_\pi$ above witness $T \models \phi$.
We show that they also witness $R \models \phi$, which completes the proof.

So, let $\Pi'$ be a variable assignment such that
\begin{itemize}
    \item $\unidisvar$ is mapped to $R$,
    \item each second-order variable~$Y_j$ is mapped to $\itereval(\Pi')(Y_j)$,
    \item for each universally quantified trace variable~$\pi$ that appears in a $\gamma_j$ we have $\Pi'(\pi) \in \Pi'(X_\pi)$, and
    \item for each existentially quantified trace variable~$\pi$ that appears in a $\gamma_j$, we have $\Pi'(\pi) = f_{\pi}(\Pi'(\pi_1), \ldots, \Pi'(\pi_{k_\pi}))$, where $\pi_1, \ldots, \pi_{k_\pi}$ are the universally quantified trace variables that appear in some $\gamma_j$ before $\pi$. Note that we do not claim or require (yet) that $\Pi'(\pi)$ is in $\Pi'(X_\pi)$.
\end{itemize}
We show that $\Pi' \models \psi$ and that $\Pi'(\pi)$ is indeed in $\Pi'(X_\pi)$. 
If this is true for every such $\Pi'$, then we can conclude that the Skolem functions witness $R \models \varphi$.

To this end, let the variable assignment~$\Pi$ be obtained by defining
\begin{itemize}
    \item $\Pi(\unidisvar) = T$, 
    \item $\Pi(Y_j) = \itereval(\Pi)(Y_j)$, 
    \item $\Pi(\pi) = \Pi'(\pi)$ for all trace variables~$\pi$ that appear in some $\gamma_j$, and
    \item $\Pi$ is undefined for all other variables.
\end{itemize}

Due to Remark~\ref{remark:lfpmonoiter}, we have $\Pi(Y_j) \supseteq \Pi'(Y_j)$.
Hence, in $\Pi$, each universal variable~$\pi$ is mapped to a trace in $\Pi(X_\pi)$: If $X_\pi = \unidisvar$, then we have \[\Pi(\pi) = \Pi'(\pi) \in \Pi'(X_\pi) = \Pi'(\unidisvar) = R \subseteq T = \Pi(\unidisvar) = \Pi(X_\pi),\]
and if $X_\pi = Y_j$, then we have
\[
\Pi(\pi) = \Pi'(\pi) \in \Pi'(X_\pi) = \Pi'(Y_j) \subseteq \Pi(Y_j) = \Pi(X_\pi).
\]
Thus, as the traces for existentially quantified variables are obtained via the Skolem functions witnessing $T \models \varphi$, we conclude $\Pi \models \psi$.
Now, as $\psi$ is quantifier-free, $\Pi\models\psi$ only depends on the restriction of $\Pi$ to trace variables. 
As $\Pi$ and $\Pi'$ coincide on the these variables, we have $\Pi'\models\psi$.

This almost concludes the proof, we just have to argue that we indeed have that $\Pi'(\pi)= f_{\pi}(\Pi'(\pi_1), \ldots, \Pi'(\pi_{k_\pi}))$ is in $\Pi'(X_\pi) $.
Note that the traces~$\Pi(\pi)$ for universally quantified $\pi$'s are all in $R$ and thus there is an even $\ell^*$ such that all the traces~$\Pi'(\pi)$ for all such $\pi$ are in $R_{\ell^*}$. 
This is in particular true for the traces~$\Pi'(\pi_1), \ldots, \Pi'(\pi_{k_\pi})$. 
We now consider two cases.

If $X_\pi = \unidisvar$, then we have \[\Pi'(\pi) \in R_{\ell^*+1} \subseteq R = \Pi'(\unidisvar) = \Pi'(X_\pi)\] by definition, as $R_{\ell^*+1}$ contains $f_\pi(R_{\ell^*})$, which contains $f_{\pi}(\Pi'(\pi_1), \ldots, \Pi'(\pi_{k_\pi}))$, as all these are universally quantified traces.

For the case $X_\pi \neq \unidisvar$ we need some more preparation.
As just argued, we have $\Pi'(\pi) = \Pi(\pi) \in R_{\ell^*+1}$.
We show that $\Pi$ as defined above is admissible at $\ell^*+1$ by verifying that all five requirements are satisfied:
\begin{enumerate}
    \item The domain requirement is satisfied by definition.
    \item We have $\Pi(\unidisvar) = T$ by definition.
    \item By definition, we have $\Pi(Y_{j'}) = \itereval(\Pi)(Y_{j'})$ for all $Y_{j'}$.
    \item We have $\Pi(\pi) \in \Pi(X_\pi)$ for all universally quantified $\pi$ as argued above and we have $\Pi(\pi) \in R_{\ell^*}$ by the choice of $\ell^*$.
    \item For existentially quantified $\pi$, we have $\Pi(\pi) = f_\pi(\Pi'(\pi_1), \ldots, \Pi'(\pi_{k_\pi}))$ by definition, which is furthermore in $\Pi(X_\pi)$, as $f_\pi$ is a Skolem function.
\end{enumerate}
Thus, all requirements on $\Pi$ for being admissible at $\ell^*+1$ are satisfied.

Hence, we can conclude that the traces of a $\Pi$-witness tree for $(\Pi(\pi),j_\pi)$ are contained in $R_{\ell^*+2} \subseteq R$ for each existentially quantified~$\pi$, where $j_\pi$ is defined such that $X_\pi = Y_{j_\pi}$.
So, there is a $\Pi$-witness tree~$B$ for $(\Pi(\pi),j_\pi)$ for each such $\pi$, as $R$ is a subset of $T = \Pi(\unidisvar)$.
Thus, Lemma~\ref{lemma:witnesstreeprops}.\ref{lemma:witnesstreeprops_restr} yields that $B$ is also a $\Pi'$-witness tree for $(\Pi(\pi),j_\pi) = (\Pi'(\pi),j_\pi)$ for each such $\pi$.
Hence, Lemma~\ref{lemma:witnesstreeprops}.\ref{lemma:witnesstreeprops_correctness} yields that $\Pi'(\pi) \in \lfp(\Pi',j_\pi) = \Pi'(Y_{j_\pi})$ for all such $\pi$. 
\end{proof}

Before we continue with our complexity results, let us briefly mention that the formula from Remark~\ref{remark_unsatisfactory} on Page~\pageref{remark_unsatisfactory} shows that the restriction to $\univar$-free sentences is essential to obtain the upper bound above.

With this upper bound, we can express the existence of (without loss of generality) countable models of a given $\univar$-free $\lfpsohyltlfp$ sentence~$\phi$ via arithmetic formulas that only use existential quantification of type~$1$ objects (sets of natural numbers), which are rich enough to express countable sets~$T$ of traces and objects (e.g., Skolem functions) witnessing that $T$ satisfies $\phi$.
This places satisfiability in $\Sigma_1^1$ while the matching lower bound already holds for $\hyltl$~\cite{hyperltlsatconf}.

\begin{thm}\label{thm_lfpsatcomplexity}
$\lfpsohyltlfp$ satisfiability for $\univar$-free sentences is $\Sigma_1^1$-complete.
\end{thm}

\begin{proof}
The $\Sigma_1^1$ lower bound already holds for $\hyltl$ satisfiability~\cite{hyperltlsatconf}, which is a fragment of $\univar$-free $\lfpsohyltlfp$ (see Remark~\ref{rem_hyltlisfragment}). 
Hence, we focus in the following on the upper bound, which is a generalization of the corresponding upper bound for $\hyltl$~\cite{hyperltlsatconf}.

Throughout this proof, we fix an $\univar$-free $\lfpsohyltlfp$ sentence
\[
\varphi =
\gamma_1.\ \quant_1 (Y_1, \smallest, \phi_1^\con).\ 
\gamma_2.\ \quant_2(Y_2, \smallest, \phi_2^\con).\ 
\ldots 
\gamma_k.\ \quant_k(Y_k, \smallest, \phi_k^\con).\
\gamma_{k+1}.\
\psi
,\] where $\psi$ is quantifier-free, and let $\Phi$ denote the set of quantifier-free subformulas of~$\varphi$ (including those in the guards). As before, we assume without loss of generality that each trace variable is quantified at most once in $\varphi$, i.e., for each trace variable~$\pi$ quantified in some $\gamma_j$ or in some $\phi_j^\con$, there is a unique second-order variable~$X_\pi$ such that $\pi$ ranges over $X_\pi$.

Due to Lemma~\ref{lemma_lfpcwmodelsize}, $\varphi$ is satisfiable if and only if it has a countable model~$T$. Thus, not only is $T$ countable, but the second-order quantifiers range over subsets of $T$, i.e., over countable sets. 
Finally, recall that the least fixed point assigned to $Y_j$ depends on the variable assignment to trace variables in the blocks~$\gamma_1,\ldots, \gamma_j$, but not on the second-order variables~$Y_{j'}$ with $j' < j$, as their values are also uniquely determined by the variables in the blocks~$\gamma_1,\ldots, \gamma_{j'}$.

To prove that the $\lfpsohyltlfp$ satisfiability problem for $\univar$-free sentences is in $\Sigma_1^1$, we express, for a given $\univar$-free $\lfpsohyltlfp$ sentence~$\phi$ (encoded as a natural number), the existence of a countable set~$T$ of traces and a witness that $T$ is indeed a model of $\varphi$.
As we work with second-order arithmetic we can quantify only over natural numbers (type $0$ objects) and sets of natural numbers (type $1$ objects). To simplify our notation, we note that there is a bijection between finite sequences over $\nats$ and $\nats$ itself and that one can encode functions mapping natural numbers to natural numbers as sets of natural numbers via their graphs. As both can be implemented in arithmetic, we will freely use vectors of natural numbers, functions mapping natural numbers to natural numbers, and combinations of both. 

Furthermore, as we only have to work with countable sets of traces, we can use natural numbers to \myquot{name} the traces. Hence, the restriction of a variable assignment to trace variables can be encoded by a vector of trace names, listing (in some fixed order), the names of the traces assigned to the trace variables. 
Finally, the countable sets assigned to the $Y_j$ depend on the variable assignment to the trace variables quantified before $Y_j$, but not on the sets assigned to the $Y_{j'}$ with $j' < j$. Thus we do not have to consider second-order variables in variable assignments: in the following, when we speak of a variable assignment we only consider those that are undefined for all second-order variables. Let $k'$ denote the number of trace variables in $\phi$. Then, variable assignments are encoded by vectors in $\nats^{k'}$.

First, we show how to capture the semantics of quantifier-free formulas in arithmetic.
Recall that $\psi$ is the maximal quantifier-free subformula of $\psi$ and that $\Phi$ denotes the set of quantifier-free subformulas of~$\varphi$ (including those in the guards).
Let $\Pi$ be a variable assignment whose domain contains all the trace variables of $\varphi$, and therefore in particular all free variables of $\psi$.
Then, $\Pi \models \psi$ if and only if there is a function~$\expansion \colon \Phi \times \nats \rightarrow \set{0,1}$ with $\expansion(\psi, 0)=1$ satisfying the following consistency conditions:
\begin{itemize}
	\item $\expansion(a_\pi,j)=1$ if and only if $a \in \Pi(\pi)(j)$.
	\item $\expansion(\neg \psi_1,j)=1$ if and only if  $\expansion(\psi_1,j)=0$.
	\item $\expansion(\psi_1 \vee \psi_2,j)=1$ if and only if $\expansion(\psi_1,j)=1$ or $\expansion(\psi_2,j)=1$.
	\item $\expansion(\X\psi_1,j)=1$ if and only if $\expansion(\psi_1,j+1)=1$.
	\item $\expansion(\psi_1 \U \psi_2,j)=1$ if and only if there is a $j' \ge j$ such that $\expansion(\psi_2,j')=1$ and $\expansion(\psi_1,j'')=1$ for all $j''$ in the range~$j \le j'' < j'$.
\end{itemize}

In fact, there is exactly one function~$\expansion$ satisfying these consistency conditions, namely the function~$\expansion_{\phi,\Pi}$ defined as
\[
\expansion_{\phi,\Pi}(\psi',j) = \begin{cases}
    1 &\text{ if }\suffix{\Pi}{j}\models \psi',\\
    0 &\text{ otherwise.}\\
\end{cases}
\]

Now, given $\varphi$ we express the existence of the following type $1$ objects:
\begin{itemize}

    \item A countable set of traces over the propositions of $\varphi$ encoded as a function~$T$ from~$\nats \times \nats$ to $\nats$, mapping trace names (i.e., natural numbers) and positions to (encodings of) subsets of the set of propositions appearing in $\varphi$.

    \item For each $j \in \set{1,2,\ldots, k}$ a function~$T_j$ from~$\nats^{k'} \times \nats$ to $\set{0,1} \times \nats$ mapping a variable assignment~$\overline{a}$ and a trace name~$n$ to a pair~$(b,\ell)$ where the bit~$b$ encodes whether the trace named~$n$ is in the set assigned to $Y_j$ (w.r.t.\ the variable assignment encoded by $\overline{a}$) and where the natural number~$\ell$ is intended to encode in which stage of the fixed point computation the trace named $n$ was added to the fixed point (computed with respect to $\overline{a}$). That this is correct will be captured by the formula we are constructing.
 
    \item A function~$\skolem$ from $\nats \times \nats^*$
      to $\nats$ to be interpreted as Skolem functions for the existentially quantified trace variables of $\varphi$, i.e., we map a variable name and a variable assignment of the variables preceding it to a trace name. 

	\item A function~$E$ from $\nats^{k'} \times \nats \times \nats$ to $\nats$, where, for a fixed~$\overline{a} \in\nats^{k'}$ encoding a variable assignment~$\Pi$, the function~$x,j\mapsto E(\overline{a}, x, j)$ is intended to encode the function~$\expansion_{\phi,\Pi}$, i.e., $x$ encodes a subformula in $\Phi$ and $j$ is a position. That this is correct will again be captured by the formula we are constructing.
\end{itemize}

Then, we express the following properties which characterize $T$ being a model of $\varphi$:
\begin{itemize}

    \item The function~$\skolem$ is indeed a Skolem function for $\varphi$, i.e., for all variable assignments~$\overline{a} \in \nats^\ell$, if the traces assigned to universally quantified variables~$\pi$ are in the set~$X_\pi$ that $\pi$ ranges over, then the traces assigned to the existentially quantified variables~$\pi'$ are in the set~$X_{\pi'}$ that $\pi'$ ranges over. Note that $X_{\pi}$ and $X_{\pi'}$ can either be $T$ or one of the $Y_j$, for which we can check membership via the functions~$T_j$. Also note that this formula refers to traces by their names (which are natural numbers) and quantifies over variable assignments (which are encoded by natural numbers), i.e., it is a first-order formula.
    
    \item For every variable assignment~$\overline{a}$ such that the traces assigned to universally quantified variables~$\pi$ are in the set~$X_\pi$ that $\pi$ ranges over and that is consistent with the Skolem functions encoded by $\skolem$: the function~$x,j\mapsto E(\overline{a}, x, j)$ satisfies the consistency conditions characterising the expansion, and we have $E(\overline{a},x_0, 0) = 1$, where $x_0$ is the encoding of $\psi$. 
    Again, this formula is first-order.

    \item For every variable assignment~$\overline{a}$ as above, the set assigned to $Y_j$ by $T_j$ w.r.t.\ $\overline{a}$ is indeed the least fixed point of $f_{\overline{a},j}$. Here, we use the information about the stages encoded by $T_j$ and again have to use the expansion~$E$ to check that the step formula is satisfied by the selected traces. As before, this can be done in first-order arithmetic.
\end{itemize}
We leave the tedious, but straightforward, details to the reader.
\end{proof}

\subsection{\texorpdfstring{Satisfiability for full \boldmath$\lfpsohyltlfp$}{Satisfiability for full Second-order HyperLTL with Least Fixed Points}}

We just proved that $\lfpsohyltlfp$ satisfiability for $\univar$-free sentences is in $\Sigma_1^1$, relying on an upper bound on the size of models: 
This bound is obtained by showing that every model of a satisfiable sentence contains a countable subset that is also a model of the sentence.
This is correct, as in $\univar$-free sentences, only traces from the model are quantified over.
On the other hand, the $\lfpsohyltlfp$ sentence~$\forall \pi \in \univar.\ \exists \pi'\in\unidisvar.\ \equals{\pi}{\pi'}{\ap}$ has only uncountable models, if $\size{\ap}>1$.

Thus, in $\lfpsohyltlfp$ formulas with $\univar$, one can refer to all traces and thus mimic quantification over sets of natural numbers. 
Furthermore, the satisfiability problem asks for the existence of a model. This \emph{implicit} existential quantifier can be used to mimic existential quantification over sets of sets of natural numbers.
Together with the fact that one can implement addition and multiplication in $\hyltl$, we show that $\lfpsohyltlfp$ satisfiability for sentences with $\univar$ is $\Sigma_1^2$-hard.

To prove a matching upper bound, we capture the existence of a model and Skolem functions witnessing that it is indeed a model in $\Sigma_1^2$.
Here, the challenge is to capture the least fixed points when mimicking the second-order quantification of $\lfpsohyltlfp$. Naively, this requires an existential quantifier (\myquot{there exists a set that satisfies the guard}) followed by a universal one (\myquot{all strict subsets do not satisfy the guard}).
However, as traces are encoded as sets, this would require universal quantification of type~$3$ objects.
Thus, this approach is not sufficient to prove a $\Sigma_1^2$ upper bound.
Instead, we do not explicitly quantify the fixed points, but instead use witness trees for the membership of traces in the fixed points, which are type~$1$ objects.
This is sufficient, as the sets of traces quantified in $\lfpsohyltlfp$ are only used as ranges for trace quantification.

\begin{thm}
\label{thm_satcomplexity_lfp_ss}
$\lfpsohyltlfp$ satisfiability is $\Sigma^2_1$-complete.
\end{thm}

\begin{proof}
We begin with the lower bound. Let $S \in \Sigma_1^2$, i.e., there exists a formula of arithmetic of the form
\[
\phi(x) = \exists \mathcal{Y}_1\subseteq \pow{\nats}.\ \cdots \exists \mathcal{Y}_k\subseteq \pow{\nats}. \ \psi(x, \mathcal{Y}_1, \ldots,\mathcal{Y}_k)
\]
with a single free (first-order) variable~$x$ such that $S = \set{n \in\nats \mid \natsstruct\models\phi(n) }$, where $\psi$ is a formula with arbitrary quantification over type~$0$ and type~$1$ objects (but no third-order quantifiers) and free third-order variables~$\mathcal{Y}_i$, in addition to the free first-order variable~$x$.
We present a polynomial-time translation from natural numbers~$n$ to $\lfpsohyltlfp$ sentences~$\phi_n$ such that $n \in S$ (i.e., $\natsstruct \models \phi(n)$) if and only if $\phi_n$ is satisfiable.
This implies that $\lfpsohyltlfp$ satisfiability is $\Sigma_1^2$-hard.

Intuitively, we ensure that each model of $\phi_n$ contains enough traces to encode each subset of $\nats$ by a trace (this requires the use of $\univar$). 
Furthermore, we have additional propositions~$\marker_i$, one for each third-order variable~$\mathcal{Y}_i$ existentially quantified in $\phi$, to label the traces encoding sets.
Thus, the set of traces marked by $\marker_i$ encodes a set of sets of natural numbers, i.e., we have mimicked existential third-order quantification by quantifying over potential models. 
Hence, it remains to mimic quantification over natural numbers and sets of natural numbers as well as addition and multiplication, which can all be done in $\hyltl$: quantification over traces mimics quantification over sets and singleton sets (i.e., numbers) and addition and multiplication can be implemented in $\hyltl$ (see Proposition~\ref{prop_plustimesinhyperltl}). 

To this end, let $\ap = \set{\inprop} \cup \ap_\marker \cup \ap_\arith$ with $\ap_\marker = \set{\marker_1,\ldots, \marker_k}$ and $\ap_\arith = \set{ \argone, \argtwo, \res, \add, \mult}$ and consider the following two formulas (both with a free variable~$\pi'$):
\begin{itemize}
    \item $\psi_0 = \left(\bigwedge_{\proposition \in \ap_\arith} \G\neg\proposition_{\pi'}\right) \wedge \left(\bigwedge_{i=1}^k (\G (\marker_i)_{\pi'}) \vee (\G\neg(\marker_i)_{\pi'})\right)$ expressing that the interpretation of $\pi'$ may not contain any propositions from $\ap_\arith$ and, for each $i \in [k]$, is either marked by $\marker_i$ (if $\marker_i$ holds at every position) or is not marked by $\marker_i$ (if $\marker_i$ holds at no position). Note that there is no restriction on the proposition~$\inprop$, which therefore encodes a set of natural numbers on each trace. 
    Thus, we can use trace quantification to mimic quantification over sets of natural numbers and quantification of natural numbers (via singleton sets).
    In our encoding, a trace bound to some variable~$\pi$ encodes a singleton set if and only if the formula~$(\neg \inprop_{\pi})\U(\inprop_{\pi} \wedge \X\G\neg \inprop_{\pi}) $ is satisfied.
    
    As explained above, we use the markings to encode the membership of such sets in the $\mathcal{Y}_i$, thereby mimicking the existential quantification of the $\mathcal{Y}_i$.

    However, we need to ensure that this marking is done consistently, i.e., there is no trace~$t$ over $\set{\inprop}$ in the model that is both marked by $m_i$ and not marked by $m_i$. 
    This could happen, as these are just two different traces over $\ap$. 
    However, the formula
    \[
    \psi_\consistent = \forall \pi \in \unidisvar.\ \forall \pi' \in \unidisvar.\ \equals{\pi}{\pi'}{\set{x}} \rightarrow \equals{\pi}{\pi'}{\set{\inprop} \cup \ap_\marker}
    \]
    disallows this.
    
    \item $\psi_1 = \bigwedge_{\proposition \in \set{\inprop} \cup \ap_\marker}\G\neg\proposition_{\pi'}$ expresses that the interpretation of $\pi'$ may only contain propositions in $\ap_\arith$. We use such traces to implement addition and multiplication in $\hyltl$.
\end{itemize}
Let $\phi_\plustimes$ be the sentence from Proposition~\ref{prop_plustimesinhyperltl} implementing addition and multiplication: 
The $\set{\argone, \argtwo, \res, \add, \mult}$-projection of every model of $\phi_\plustimes$ is $T_\plustimes$ as defined for Proposition~\ref{prop_plustimesinhyperltl}. 
\label{pageref_plustimesaslfp} By closely inspecting the formula~$\phi_\plustimes$, we can see that it can be brought into the form required for guards in $\lfpsohyltlfp$. Call the resulting formula~$\phi_\plustimes'$. It uses two free variables~$\pi_\add$ and $\pi_\mult$ as \myquot{seeds} for the fixed point computation and comes with another $\ltl$ formula~$\psi_s$ with free variables~$\pi_\add$ and $\pi_\mult$ that ensures that the seeds have the right format.

Now, given $n\in\nats$ we define the $\lfpsohyltlfp$ sentence
\begin{align*}
\phi_n = {}&{} \psi_\consistent \wedge \forall\pi \in\univar.\ \left(
(\exists \pi' \in \unidisvar.\ \equals{\pi}{\pi'}{\set{\inprop}} \wedge \psi_0) \wedge
(\exists \pi' \in \unidisvar.\ \equals{\pi}{\pi'}{\ap_\arith} \wedge \psi_1)\right)
\wedge \\
{}&{}\quad\exists \pi_\add \in \unidisvar.\ \exists \pi_\mult \in \unidisvar.\ \psi_s \wedge \exists(X_\arith, \smallest, \phi_\plustimes').\ \hyperize(\exists x.\ ( x = n \wedge \psi)).    
\end{align*}
Intuitively, $\phi_n$ requires that the model contains, for each subset~$S$ of $\nats$, a unique trace encoding $S$ (additionally marked with the $\marker_i$ to encode membership of $S$ in the set bound to $\mathcal{Y}_i$), contains each trace over $\ap_\arith$, the set~$X_\arith$ is interpreted with $T_\plustimes$, and the formula~$\hyperize(\exists x.\ ( x = n \wedge \psi))$ is satisfied, where the translation~$\hyperize$ is defined below.
Note that we use the constant~$n$, which is definable in first-order arithmetic (with a formula that is polynomial in $\log(n)$, using the fact that the constant~$2$ is definable in first-order arithmetic and then using powers of $2$ to define~$n$).

Now, $\hyperize$ is defined inductively as follows: 
\begin{itemize}
    \item For second-order variables~$Y$, $\hyperize(\exists Y.\ \psi) = \exists \pi_Y \in \unidisvar.\ \neg\add_{\pi_Y} \wedge \neg\mult_{\pi_Y} \wedge \hyperize(\psi)$, as only traces not being labeled by $\add$ or $\mult$ encode sets.

    \item For second-order variables~$Y$, $\hyperize(\forall Y.\ \psi) = \forall \pi_Y \in \unidisvar.\ (\neg\add_{\pi_Y} \wedge \neg\mult_{\pi_Y}) \rightarrow \hyperize(\psi)$.

    \item For first-order variables~$y$, $\hyperize(\exists y.\ \psi) = \exists \pi_y \in \unidisvar.\ \neg\add_{\pi_y} \wedge \neg\mult_{\pi_y} \wedge ((\neg \inprop_{\pi_y})\U(\inprop_{\pi_y} \wedge \X\G\neg \inprop_{\pi_y})) \wedge \hyperize(\psi)$.

    \item For first-order variables~$y$, $\hyperize(\forall y.\ \psi) = \forall \pi_y \in \unidisvar.\ (\neg\add_{\pi_y} \wedge \neg\mult_{\pi_y} \wedge (\neg \inprop_{\pi_y})\U(\inprop_{\pi_y} \wedge \X\G\neg \inprop_{\pi_y})) \rightarrow \hyperize(\psi)$.

    \item $\hyperize(\psi_1 \lor \phi_2) = \hyperize(\psi_1) \lor \hyperize(\psi_2)$.
    
    \item $\hyperize(\lnot \psi) = \lnot \hyperize(\psi) $.

    \item For third-order variables~$\mathcal{Y}_i$ and second-order variables~$Y$, $\hyperize(Y\in \mathcal{Y}_i) = (\marker_i)_{\pi_Y}$.
    
    \item For second-order variables~$Y$ and first-order variables~$y$, $\hyperize(y\in Y) = \F(\inprop_{\pi_y} \land \inprop_{\pi_Y})$.
    
    \item For first-order variables~$y,y'$, $\hyperize(y<y') = \F(\inprop_{\pi_y} \land \X\F\inprop_{\pi_{y'}})$.
    
    \item For first-order variables~$y_1,y_2,y$, $\hyperize(y_1+y_2=y) = \exists \pi \in X_\arith.\ \add_\pi \land \F(\inprop_{\pi_{y_1}}\land\argone_\pi) \land \F(\inprop_{\pi_{y_2}}\land\argtwo_\pi) \land \F(\inprop_{\pi_y}\land\res_\pi)$.
    
    \item For first-order variables~$y_1,y_2,y$, $\hyperize(y_1 \cdot y_2=y) = \exists \pi \in X_\arith.\ \mult_\pi \land \F(\inprop_{\pi_{y_1}}\land\argone_\pi) \land \F(\inprop_{\pi_{y_2}}\land\argtwo_\pi) \land \F(\inprop_{\pi_y}\land\res_\pi)$.
    
\end{itemize}
While $\phi_n$ is not in prenex normal form, it can easily be brought into prenex normal form, as there are no quantifiers under the scope of a temporal operator.
An induction shows that we indeed have that $\natsstruct\models\phi(n)$ if and only if $\phi_n$ is satisfiable.
    
For the upper bound, we show that $\lfpsohyltlfp$ satisfiability is in $\Sigma_1^2$. More formally, we show how to construct a formula of the form
\[
\theta(x) = \exists \mathcal{Y}_1\subseteq \pow{\nats}.\ \cdots \exists \mathcal{Y}_k\subseteq \pow{\nats}. \ \psi(x, \mathcal{Y}_1, \ldots,\mathcal{Y}_k)
\]
with a single free (first-order) variable~$x$ such that an $\lfpsohyltlfp$ sentence~$\phi$ is satisfiable if and only if $\natsstruct \models \theta(\encode{\phi})$.
Here, $\psi$ is a formula of arithmetic with arbitrary quantification over type~$0$ and type~$1$ objects (but no third-order quantifiers) and free third-order variables~$\mathcal{Y}_i$, in addition to the free first-order variable~$x$, and $\encode{\cdot}$ is a polynomial-time computable injective function mapping $\lfpsohyltlfp$ sentences to natural numbers.

In the following, we assume, without loss of generality, that $\ap$ is fixed, so that we can use $\size{\ap}$ as a constant in our formulas (which is definable in arithmetic). 
Here, we reuse the encoding of traces as sets of natural numbers as introduced in the proof of Theorem~\ref{thm_satcomplexity} using the pairing function~$\pair$ and a bijection~$e \colon \ap \rightarrow\set{0,1,\ldots,\size{\ap}-1}$.
Recall that we encode a trace~$t \in (\pow{\ap})^\omega$ by the set~$S_t =\set{\pair(j,e(\proposition)) \mid j \in \nats \text{ and } \proposition \in t(j)} \subseteq \nats$.
A finite collection of sets~$S_1, \ldots, S_k$ is uniquely encoded by the set~$\set{\pair(n,j) \mid n \in S_j}$, i.e., we can encode finite sets of sets by type~$1$ objects.
In particular, we can encode a variable assignment whose domain is finite and contains only trace variables by a set of natural numbers and we can write a  formula that checks whether a trace (encoded by a set) is assigned to a certain variable.

Now, the overall proof idea is to let the formula~$\theta$ of arithmetic express the existence of Skolem functions for the existentially quantified variables in the $\lfpsohyltlfp$ sentence~$\phi$ such that each variable assignment that is consistent with the Skolem functions satisfies the maximal quantifier-free subformula of $\varphi$.
For trace variables~$\pi$, a Skolem function is a type~$2$ object, i.e., a function mapping a tuple of sets of natural numbers (encoding a tuple of traces, one for each variable quantified universally before $\pi$) to a set of natural numbers (encoding a trace for $\pi$).
However, to express that the interpretation of a second-order variable~$X$ is indeed a least fixed point we need both existential quantification (\myquot{there exists a set} that satisfies the guard) and universal quantification (\myquot{every other set that satisfies the guard is larger}).
Thus, handling second-order quantification this way does not yield the desired~$\Sigma_1^2$ upper bound. 

Instead we use that fact that membership of a trace in an $\hyltl$-definable least fixed point can be witnessed by a finite tree labeled by traces (Lemma~\ref{lemma:witnesstreeprops}), i.e., by a type~$1$ object .
Thus, instead of capturing the full least fixed point in arithmetic, we verify on-the-fly for each trace quantification of the form~$\exists \pi \in Y_j$ or $\forall \pi \in Y_j$ whether the interpretation of $\pi$ is in the interpretation of $Y_j$, which only requires the quantification of a witness tree.

To this end, fix some \[
\phi = 
\gamma_1.\ \quant_1(Y_1, \smallest, \phi_1^\con).\ 
\gamma_2.\ \quant_2(Y_2, \smallest, \phi_2^\con).\ 
\ldots 
\gamma_k.\ \quant_k(Y_k, \smallest, \phi_k^\con).\
\gamma_{k+1}.\
\psi,
\]
where $\psi$ is quantifier-free.
As before, we assume without loss of generality that each trace variable is quantified at most once in $\varphi$, which can always be achieved by renaming variables. This implies that for each trace variable~$\pi$ quantified in some $\gamma_j$, there exists a unique second-order variable~$X_\pi$ such that $\pi$ ranges over $X_\pi$. 
Furthermore, we assume that each $Y_j$ is different from $\unidisvar$ and $\univar$, which can again be achieved by renaming variables, if necessary.
Recall that the values of the least fixed points are uniquely determined by the interpretations of $\unidisvar$, $\univar$, and the trace variables in the $\gamma_j$. 

We need to adapt the definition of witness trees (and related concepts) to account for the fact that we here allow the usage of the variable~$\univar$.
Hence, we say that a variable assignment $\Pi$ is $\phi$-sufficient if $\Pi$'s domain contains exactly the variables~$\unidisvar$, $\univar$, and the trace variables in the $\gamma_j$. 

The expansion of a formula~$\xi$ with respect to $\Pi$ is  still denoted by $e_{\Pi,\xi}\colon \Xi\times\nats \rightarrow \set{0,1}$ and defined as in the proof of Theorem~\ref{thm_lfpsatcomplexity}.
Here, $\Xi$ denotes the set of subformulas of $\xi$.
Let us also remark for further use that a function from $\Xi\times\nats $ to $\set{0,1}$ is a type~$1$ object (as functions from $\nats$ to $\nats$ can be encoded as subsets of $\nats$ via their graph and the pairing function introduced below) and that all five requirements on the expansion can then be expressed in first-order logic.

Now, we adapt the definition of witness trees as follows.
Let
\begin{equation}
\label{eq_phiconscnd}
\phi_j^\con = \overdot{\pi}_1\tracein Y_j \wedge \cdots \wedge \overdot{\pi}_n \tracein Y_j \wedge   \forall \overdotdot{\pi}_1 \in Z_1.\ \ldots \forall \overdotdot{\pi}_{n'} \in Z_{n'}.\ \psi^\step_j \rightarrow \overdotdot{\pi}_m \tracein Y_j.    
\end{equation}
inducing the unique least fixed point that $Y_j$ ranges over. 
In the original definition in Section~\ref{subsec_lfpsatcw}, each $Z_i$ is either $\unidisvar$ or one of the $Y_{j'}$ with $j' \le j$.
Here, it may also be $\univar$.
However, selecting a trace from the interpretation of $\univar$ does not introduce any constraints (as it is interpreted by the set of all traces, as we do not allow to re-quantify $\univar$).
Hence, such a node will be just a leaf without any constraint. 

Formally, let us fix a $j^* \in \set{1,\ldots,k}$, a $\phi$-sufficient assignment~$\Pi$, and a trace~$t^*$. 
A $\Pi$-witness tree for $(t^*,j^*)$ (which witnesses~$t^* \in\lfp(\Pi,j^*)$) is an ordered finite tree whose vertices are labeled by pairs~$(t,j)$ where $t$ is a trace and where $j $ is in $\set{1,\ldots, k} \cup \set{a,d}$ (where $a$ and $d$ are fresh symbols) such that the following are satisfied:
\begin{itemize}
    \item The root is labeled by $(t^*, j^*)$.
    
    \item If a vertex is labeled with $(t,a)$ for some trace~$t$, then it must be a leaf. Note that $t$ is in $\Pi(\univar)$, as that set contains all traces.
    
    \item If a vertex is labeled with $(t,d)$ for some trace~$t$, then it must be a leaf and $t \in \Pi(\unidisvar)$.
    
    \item Let $(t,j)$ be the label of some vertex~$v$ with $j \in \set{1,\ldots, k}$ and let $\phi_j^\con$ as in (\ref{eq_phiconscnd}).
If $v$ is a leaf, then we must have $t = \Pi(\overdot{\pi}_i)$ for some $i \in \set{1, \ldots, n}$.
If $v$ is an internal vertex, then it must have $n'$ successors labeled by $(t_1,j_1),\ldots,(t_{n'},j_{n'})$ (in that order) such that $\Pi[\overdotdot{\pi}_1 \mapsto t_1, \ldots, \overdotdot{\pi}_{n'} \mapsto t_{n'}] \models \psi^\step_j $, $t = t_m$, and
such that the following holds for all $i \in \set{1, \ldots, n'}$:
\begin{itemize}
    \item If $Z_i = \univar$, then the $i$-th successor of $v$ is a leaf and $j_i = a$.    
    
    \item If $Z_i = \unidisvar$, then the $i$-th successor of $v$ is a leaf, $t_i \in \Pi(\unidisvar)$, and $j_i = d$.
    
    \item If $Z_i = Y_{j'}$ for some $j'$ (which satisfies $1 \le j' \le j$ by definition of $\lfpsohyltlfp$), then we must have that $ j_i = j'$ and that the subtree rooted at the $i$-th successor of $v$ is a $\Pi$-witness tree for $(t_i,j_i)$.
\end{itemize}
\end{itemize}

The following proposition states that membership in the fixed points is witnessed by witness trees. It is obtained by generalizing the similar argument for $\univar$-free sentences presented in Lemma~\ref{lemma:witnesstreeprops}.

\begin{lem}
Let $\Pi$ be a $(\phi, j)$-sufficient assignment with $\Pi(Y_{j'}) = \itereval(\Pi)(Y_{j'})$ for all $j' < j$.
Then, we have $t \in \lfp(\Pi, j)$ if and only if there is a $\Pi$-witness tree for $(t,j)$.
\end{lem}

Note that a witness tree is a type~$1$ object: The (finite) tree structure can be encoded by 
\begin{itemize}
    \item a natural number~$s>0$ (encoding the number of vertices), 
    \item a function from $\set{1,\ldots, s} \times \set{1,\ldots, n''} \rightarrow \set{0,1,\ldots, s}$ encoding the child relation, i.e., $(v,j) \mapsto v'$ if and only if the $j$-th child of $v$ is $v'$ (where we use $0$ for undefined children and $n''$ is the maximum over all $n'$ in the $\phi_j^\con$), 
    \item $s$ traces over $\ap$ and $s$ values in $\set{1,\ldots, k, k+1, k+2}$ to encode the labeling (where we use $k+1$ for $a$ and $k+2$ for $d$).
\end{itemize}
Note that the function encoding the child relation can be encoded by a finite set by encoding its graph using the pairing function while all other objects can directly be encoded by sets of natural numbers, and thus can be encoded by a single set as explained above.

Furthermore, one can write a second-order formula~$\psi_{\hastree}(X_D,A, X_{t^*}, j^*)$ with free third-order variable~$X_D$ (encoding a set of traces~$T$), free second-order variables~$A$ (encoding a variable assignment~$\Pi$ whose domain contains exactly the trace variables in the $\gamma_j$) and $X_{t^*}$ (encoding a trace~$t^*$), and free first-order variable~$j^*$ that holds in $\natsstruct$ if and only if there exists a $\Pi[\unidisvar\mapsto T, \univar\mapsto (\pow{\ap})^\omega]$-witness tree for $(t^*,j^*)$.
To evaluate the formulas~$\psi_{j'}^\step$ as required by the definition of witness trees, we rely on the expansion as introduced above, which here is a mapping from vertices in the tree and subformulas of the $\psi_{j'}^\step$ to $\set{0,1}$, and depends on the set of traces encoded by $X_D$ and the variable assignment encoded by $A$.
Such a function is a type~$1$ object and can therefore be quantified in $\psi_{\hastree}(X_D,A, X_{t^*}, j^*)$.

Recall that we construct a formula~$\theta(x)$ with a free first-order variable~$x$ such that an $\lfpsohyltlfp$ sentence~$\phi$ is satisfiable if and only if $\natsstruct \models \theta(\encode{\phi})$, where $\encode{\cdot}$ is a polynomial-time computable injective function mapping $\lfpsohyltlfp$ sentences to natural numbers.

With this preparation, we can define $\theta(x)$ such that it is not satisfied in $\natsstruct$ if the interpretation of $x$ does not encode an $\lfpsohyltlfp$ sentence.
If it does encode such a sentence~$\phi$, let $\psi$ be its maximal quantifier-free subformula (which is \myquot{computable} in first-order arithmetic using a suitable encoding~$\encode{\cdot}$).
Then, $\theta$ expresses the existence of
\begin{itemize}
    \item a set~$T$ of traces (bound to the third-order variable~$X_D$ and encoded as a set of sets of natural numbers, i.e., a type~$2$ object),
    \item Skolem functions for the existentially quantified trace variables in the $\gamma_j$ (which can be encoded by functions from $(\pow{\nats})^\ell$ to $\pow{\nats}$ for some $\ell$, i.e., by a type~$2$ object), and
    \item a function~$e \colon \pow{\nats} \times \nats \times \nats \rightarrow \set{0,1}$
\end{itemize}
such that the following is true for all variable assignments~$\Pi$ (restricted to the trace variables in the $\gamma_j$ and bound to the second-order variable~$A$):
If 
\begin{itemize}
    \item $\Pi$ is consistent with the Skolem functions for all existentially quantified variables,
    \item for all universally quantified $\pi$ in some $\gamma_j$ with $X_\pi = \unidisvar$, we have $\Pi(\pi) \in T$, and
    \item for all universally quantified $\pi$ in some $\gamma_j$ with $X_\pi = Y_{j^*}$, $\psi_{\hastree}(X_D,A, \Pi(\pi), j^*)$ holds, 
\end{itemize}
then
\begin{itemize}
    \item for all existentially quantified $\pi$ in some $\gamma_j$ with $X_\pi = \unidisvar$, we have $\Pi(\pi) \in T$, and
    \item for all existentially quantified $\pi$ in some $\gamma_j$ with $X_\pi = Y_{j^*}$, $\psi_{\hastree}(X_D,A, \Pi(\pi), j^*)$ holds, and
    \item the function~$(x,y)\mapsto e(A,x,y)$ is the expansion of $\psi$ with respect to $\Pi$ and we have $e(A,\psi,0) =1$ (here we identify subformulas of $\psi$ by natural numbers).
\end{itemize}
We leave the tedious, but routine, details to the reader. 
\end{proof}

Note that in the lower bound proof, a single second-order quantifier (for the set of traces implementing addition and multiplication) suffices. 

\subsection{\texorpdfstring{Finite-State Satisfiability and Model-Checking for \boldmath$\lfpsohyltlfp$}{Finite-State Satisfiability and Model-Checking for Second-order HyperLTL with Least Fixed Points}}

In this subsection, we settle the complexity of finite-state satisfiability and model-checking for $\lfpsohyltlfp$, both for general sentences and $\univar$-free sentences.

\begin{thm}
\label{thm_fssatmccomplexity_lfp}
$\lfpsohyltlfp$ finite-state satisfiability and model-checking are poly\-nomial-time equivalent to truth in second-order arithmetic. The lower bounds already hold for $\univar$-free formulas.
\end{thm}

The result follows from the following building blocks, which are visualized in Figure~\ref{fig_reductions}:

\begin{itemize}
    \item All lower bounds proven for $\univar$-free $\lfpsohyltlfp$ sentences trivially also hold for $\lfpsohyltlfp$ while upper bounds for $\lfpsohyltlfp$ also hold for the fragment of $\univar$-free $\lfpsohyltlfp$ sentences.
    \item We show that $\lfpsohyltlfp$ finite-state satisfiability can in polynomial time be reduced to $\lfpsohyltlfp$ model-checking for $\univar$-free sentences (see Lemma~\ref{lemma_fssat2mc} below).
    \item $\lfpsohyltlfp$ finite-state satisfiability for $\univar$-free sentences is at least as hard as truth in second-order arithmetic (see Lemma~\ref{lemma_fssat_soahard} below).
    \item $\lfpsohyltlfp$ model-checking can in polynomial time be reduced to truth in second-order arithmetic (see Lemma~\ref{lemma_mc_soaeasy} below).
\end{itemize}

\begin{figure}[h]
    \centering
    \begin{tikzpicture}

\fill[lightgray!40,rounded corners] (-5.25,-.4) rectangle (5.25,.4);    
\fill[lightgray!40,rounded corners] (-5.25,2.6) rectangle (5.25,3.4);    

\fill[lightgray!40,rounded corners] (-.75,-1) rectangle (.75,5.75);    
\fill[lightgray!40,rounded corners] (3.25,-1) rectangle (4.75,5.75);    

\fill[lightgray] (-.75,-.4) rectangle (.75,.4);    
\fill[lightgray] (-.75,2.6) rectangle (.75,3.4);    

\fill[lightgray] (3.25,-.4) rectangle (4.75,.4);    
\fill[lightgray] (3.25,2.6) rectangle (4.75,3.4);

\node[anchor = west] at (-5.1,3) {$\lfpsohyltlfp$};
\node[anchor = west] at (-5.1,0) {$\univar$-free $\lfpsohyltlfp$};

\node[align=center,anchor = north] at (0,5.5) {finite-\\state\\satisfi-\\ability};
\node[align=center,anchor = north] at (4,5.5) {model-\\checking};

\fill (0,0) circle (.1cm);
\fill (0,3) circle (.1cm);
\fill (4,0) circle (.1cm);
\fill (4,3) circle (.1cm);

\path[->,>=stealth,ultra thick, shorten > = .2cm, shorten < = .2cm]
(0,0) edge[bend left=15] node[left] {trivial} (0,3)
(4,0) edge[bend right=15] node[right] {trivial} (4,3)
(0,0) edge[bend right=15] node[below] {Lemma~\ref{lemma_fssat2mc}} (4,0)
(0,3) edge[bend left=15] node[above,rotate=-37] {Lemma~\ref{lemma_fssat2mc}} (4,0)
;

\node[align = center, anchor = west] (ub) at (5.75,3) {upper bound:\\Lemma~\ref{lemma_mc_soaeasy}};
\node[align = center, anchor = north] (lb) at (0,-1.5) {lower bound:\\Lemma~\ref{lemma_fssat_soahard}};

\path[->,>=stealth,ultra thick, shorten > = .2cm, shorten < = .2cm, dotted]
(ub) edge (4,3)
(lb) edge (0,0)
;

    \end{tikzpicture}    
    \caption{The reductions (drawn as solid arrows) and lemmata proving Theorem~\ref{thm_fssatmccomplexity_lfp}.}
    \label{fig_reductions}
\end{figure}

We begin with reducing finite-state satisfiability to model-checking.

\begin{lem}
\label{lemma_fssat2mc}
$\lfpsohyltlfp$ finite-state satisfiability is polynomial-time reducible to $\lfpsohyltlfp$ model-checking for $\univar$-free sentences.
\end{lem}

\begin{proof}
Intuitively, we reduce finite-state satisfiability of $\varphi$ to model-checking by writing a sentence~$\varphi'$ existentially quantifying a finite transition system~$\tsys$ (encoded by two traces) and expressing that $\tsys$ satisfies $\phi$ and then model-checking $\varphi'$ in a fixed transition system~$\tsys'$. 

To simplify our construction, we begin by showing that we can restrict ourselves without loss of generality to transition systems with a unique initial vertex, i.e., we show that a $\sohyltl$ sentence is finite-state satisfiable if and only if it is satisfied by a transition system with a single initial vertex.
Then, we show how to encode finite transition systems by pairs of traces,  one trace encoding the labels of the vertices and the other one the adjacency matrix of transition relation. 
Using this encoding, we then show how to capture the set of (encodings of) prefixes of traces using a least fixed point.
Then, we can translate finite-state satisfiability of a sentence~$\varphi$ to model-checking by existentially quantifying traces encoding a finite transition system~$\tsys$ and then relativizing the quantifiers of $\varphi$ to traces whose prefixes are all in the least fixed point, which implies that they are traces of $\tsys$.
To this end, it suffices to model-check the resulting formula in a fixed transition system~$\tsys'$ that contains all traces over the atomic propositions occurring in $\varphi$ as well as some additional propositions used in the construction.

We begin by showing that we can restrict ourselves, without loss of generality, to finite-state satisfiability by transition systems with a single initial vertex, which simplifies our construction.
Let $\tsys = (V, E, I, \lambda)$ be a transition system and $\phi$ a $\sohyltl$ sentence.
Consider the transition system~$\tsys_{\X} = (V\cup\set{v_\initmark}, E', \set{v_\initmark}, \lambda')$ with a fresh initial vertex~$v_\initmark$,
\[
E' = E \cup \set{(v_\initmark, v_\initmark)} \cup \set{(v_\initmark, v) \mid v \in I},
\]
and $\lambda'(v_\initmark) = \set{\$}$ (for a fresh proposition~$\$$) and $\lambda'(v) = \lambda(v)$ for all $v \in V$.
Here, we add the self-loop on the fresh initial vertex~$v_\initmark$ to deal with the special case of $I$ being empty, which would make $v_\initmark$ terminal without the self-loop.

Now, we have
\[
\traces(\tsys_{\X}) = \set{\$}^\omega \cup \set{\set{\$}^+\cdot t \mid t \in \traces(\tsys) }.
\]
Furthermore, let $\phi_{\X}$ be the formula obtained from $\phi$ by adding an $\X$ to the maximal quantifier-free subformula of $\phi$ and by inductively replacing 
\begin{itemize}
\item each $\exists \pi \in X.\ \psi$ by $\exists \pi \in X.\ \X(\neg\$_\pi \wedge \psi)$ and
\item each $\forall \pi \in X.\ \psi$ by $\forall \pi \in X.\ \X(\neg\$_\pi \rightarrow \psi)$.
\end{itemize}
Then, $\tsys \models \phi$ if and only if $\tsys_{\X} \models \phi_{\X}$.
Thus, $\phi$ is finite-state satisfiable if and only if there exists a finite transition system with a single initial vertex that satisfies $\phi_{\X}$.

Now, we show how to encode a finite transition system using two traces. 
To this end, let $\phi$ be a $\lfpsohyltlfp$ sentence (over $\ap$, which we assume to be fixed).
We define $\ap' =  \ap \cup \set{\inprop,\#} \cup \set{\argone, \argtwo, \res, \add, \mult}$. 
Throughout the construction, we use a second-order variable~$Y_\arith$ which will be interpreted by $T_\plustimes$ (see Proposition~\ref{prop_plustimesinhyperltl}).
We now explain how we encode a finite transition system by two traces over $\ap'$:
\begin{itemize}
	\item Consider the formula~$\phi_\vertices = \neg \#_{\pi_\vertices} \wedge (\neg \#_{\pi_\vertices}) \U (\#_{\pi_\vertices} \wedge \X\G \neg \#_{\pi_\vertices})$ expressing that the interpretation of ${\pi_\vertices}$ contains a unique position where $\#$ holds. This position must not be the first one. Hence, the $\ap$-projection of such a trace up to the $\#$ is a nonempty word~$w(0)w(1)\cdots w(n-1)$ for some $n > 0$. 
	It induces $V = \set{0,1,\ldots, n-1}$ and $\lambda \colon V \rightarrow \pow{\ap}$ with $\lambda(v) = w(v) \cap \ap$. Furthermore, we fix $I = \set{0}$.
	
	\item 
Let $t$ be a trace over $\ap'$. It induces the edge relation~$E = 
\set{(i,j) \in V\times V \mid \inprop \in t(i\cdot n + j)}$, i.e., we interpret 
the first $n^2$ truth values of $\inprop$ in $t$ as an adjacency matrix (encoded as a sequence of rows of the matrix). Furthermore, if $t$ is the interpretation of $\pi_\edges$ satisfying $\phi_\edges$ defined below, then every vertex in $\set{0,1,\ldots, n-1}$ has a successor in $E$.
	\begin{align*}
	\phi_\edges ={}&{} \forall \pi \in Y_\arith.\ [\mult_\pi \wedge (\F (\argone_{\pi} \wedge \X\F \#_{\pi_\vertices}))\wedge \F(\argtwo_\pi \wedge \#_{\pi_\vertices})]\rightarrow \\
	{}&{}\quad\exists \pi' \in Y_\arith.\ 
	\add_{\pi'} \wedge 
	\F(\argone_{\pi'} \wedge \res_{\pi}) \wedge 
	(\F (\argtwo_{\pi'} \wedge \X\F \#_{\pi_\vertices}) \wedge 
	\F(\res_{\pi'} \wedge 
	\inprop_{\pi_\edges}).
	\end{align*}
Intuitively, the formula~$\phi_\edges$ expresses that for every $i \in V = \set{0,1,\ldots, n-1}$ there is a $j \in V$ such that the edge~$(i,j)$ is in the edge relation encoded by the trace assigned to $\pi_E$, which is the case if and only if the proposition~$\inprop$ holds at position~$i \cdot n +j$.
To this end, we universally quantify a trace (assigned to $\pi$) encoding a multiplication~$n_1 \cdot n_2= n_3$.
Now, if 
\begin{itemize}
    \item $n_1$ is strictly smaller than $n$, expressed using the subformula~$\F (\argone_{\pi} \wedge \X\F \#_{\pi_\vertices})$ (recall that $\#$ holds at position~$n$ (and nowhere else) of the trace assigned to $\pi_V$), and if
    \item $n_2$ is equal to $n$, expressed using the subformula~$\F(\argtwo_\pi \wedge \#_{\pi_\vertices})$,
\end{itemize}
i.e., $n_1$ corresponds to $i$ and $n_2$ corresponds to $n$ in $i \cdot n +j$,
then we existentially quantify a trace (assigned to $\pi'$) encoding an addition~$n_4 + n_5 = n_6$.
Here, we require that
\begin{itemize}
    \item $n_4 $ is equal to $ n_3$ (and thus equal to $i \cdot n$), using the subformula~$\F(\argone_{\pi'} \wedge \res_{\pi})$,
    \item $n_5$ is strictly smaller than $n$, using the subformula~$\F (\argtwo_{\pi'} \wedge \X\F \#_{\pi_\vertices})$, and
    \item the proposition~$\inprop$ holds at position~$n_6$ of the trace~$t$ assigned to $\pi_\edges$, using the subformula~$\F(\res_{\pi'} \wedge 
	\inprop_{\pi_\edges})$,
\end{itemize}
i.e., $n_5$ corresponds to $j$ in $i \cdot n +j$ and $n_6$ is therefore equal to $i \cdot n +j$.

\end{itemize}
Thus, $\pi_\vertices$ and $\pi_\edges$ satisfying the above formulas indeed encode a
finite transition system~$\tsys$ with a single initial vertex as described above.

Our next step is to encode trace prefixes of the encoded transition system and then inductively construct the set of all such prefixes, by describing it as the minimal fixed point that contains the trace prefix of the unique path of length one starting at the initial vertex and that is closed under extending any trace prefix by one transition. Note that this only requires us to keep track of the last vertex of the path inducing the trace. We do so by marking a position of the trace that encodes the vertex as done before.

 Consider the formula
	\[\phi_\prefix = 
    [(\neg\#_\pi) \U (\#_\pi \wedge \X\G\neg \#_\pi)] \wedge 
    [(\neg\inprop_\pi) \U (\inprop_\pi \wedge \X\G\neg \inprop_\pi)] \wedge 
    \F (\inprop_{\pi} \wedge \X\F \#_{\pi_\vertices})
	,\]
	 which is satisfied if the interpretation~$t$ of $\pi$ has a unique position~$\ell$ at which $\#$ holds and if it has a unique position~$v$ where $\inprop$ holds. Furthermore, $v$ must be strictly smaller than $n$, i.e., it is in $V$.
	In this situation, we interpret the $\ap$-projection of $t(0)\cdots t(\ell-1)$ as a trace prefix over $\ap$ and $v$ as a vertex (intuitively, this is interpreted as the vertex where the path prefix inducing the trace prefix ends).

	Using this encoding, the formula
	 \[
\phi_\init = \phi_\prefix[\pi/\pi_\init] \wedge \inprop_{\pi_\init} \wedge \X\#_{\pi_\init} \wedge \bigwedge_{\proposition \in \ap} \proposition_{\pi_\vertices} \leftrightarrow \proposition_{\pi_\init}
\]
ensures that $\pi_\init$ is interpreted with a trace that encodes the trace prefix of the unique path of length one starting at the initial vertex. Here, $\phi_\prefix[\pi/\pi_\init]$ denotes the formula obtained from $\phi_\prefix$ by replacing each occurrence of $\pi$ by $\pi_\init$.

Recall that we have fixed a $\lfpsohyltlfp$ sentence~$\phi$ over $\ap$ and that we defined $\ap' =  \ap \cup \set{\inprop,\#} \cup \set{\argone, \argtwo, \res, \add, \mult}$.
Let $\tsys'$ be a fixed transition system such that $\traces(\tsys') = (\pow{\ap'})^\omega$.
We can construct $\tsys'$ such that it has $2^{\size{\ap'}}$ many vertices (which is constant as $\ap$ is fixed!)
Given $\phi$, we construct an $\univar$-free $\lfpsohyltlfp$ sentence~$\phi'$ (over $\ap'$) such that $\phi$ is satisfied by a finite transition system with a single initial vertex if and only if $\tsys' \models\phi'$.

To this end, first consider the formula
\begin{align*}
\phi'' = \exists \pi_\vertices \in \unidisvar.\ \exists \pi_\edges \in \unidisvar.\ \exists \pi_\init \in \unidisvar.\ \phi_\vertices \wedge \phi_\edges \wedge \phi_\init \wedge \exists (Y_\prefixes,\smallest, \phi_\con).\ \rel(\phi),
\end{align*}
where $\phi_\con$ and $\rel(\phi)$ are introduced below.
Here, we quantify $\pi_\vertices$ and $\pi_\edges$ and ensure that they encode a finite transition system~$\tsys$ as described above, and want to express that the interpretation of $Y_\prefixes$  must contain exactly the encodings of prefixes of traces of $\tsys$. 

To this end, $\phi_\con$ is defined as 
\begin{align*}
	\phi_\con = \pi_\init \tracein Y_\prefixes \wedge \forall \pi_0 \in Y_\prefixes.\ \forall \pi_1 \in \unidisvar.\ 
	\psi_\step \rightarrow \pi_1 \tracein Y_\prefixes
\end{align*}
with the formula~$\psi_\step $
\begin{align*}
\phi_\prefix[\pi/\pi_1] \wedge &
	(\bigwedge_{\proposition \in \ap} \proposition_{\pi_0} \leftrightarrow \proposition_{\pi_1}) \U(\#_{\pi_0} \wedge \X \#_{\pi_1}) \wedge\\
&	\Big(\bigwedge_{\proposition \in \ap} [\F(\proposition_{\pi_\vertices} \wedge \inprop_{\pi_1}) ]\leftrightarrow [\F(\proposition_{\pi_1}  \wedge \X\#_{\pi_1})] \Big) \wedge 
\phi_{\edge},    
\end{align*}
where $\phi_\edge$ checks whether there is an edge in the encoded transition system~$\tsys$ between the vertex induced by $\pi_0$ and the vertex induced by $\pi_1$ (analogously to $\phi_\edges$):
\begin{align*}
\phi_\edge = \exists \pi_m \in Y_\arith.\ \exists \pi_a\in Y_\arith.\
{}&{}
\mult_{\pi_m} \wedge
\F(\argone_{\pi_m} \wedge \inprop_{\pi_0}) \wedge
\F(\argtwo_{\pi_m} \wedge \#_{\pi_\vertices}) \wedge\\
{}&{}\hspace{-2cm}
\add_{\pi_a} \wedge
\F(\argone_{\pi_a} \wedge \res_{\pi_m} ) \wedge
\F(\argtwo_{\pi_a} \wedge \inprop_{\pi_1}) \wedge
\F(\res_{\pi_a} \wedge \inprop_{\pi_\edges})
\end{align*}
Intuitively, the until subformula of $\psi_\step$ ensures that the finite trace encoded by the interpretation of $\pi_0$ is a prefix of the finite trace encoded by the interpretation of $\pi_1$, which has one additional letter.
The equivalence~$\F(\proposition_{\pi_\vertices} \wedge \inprop_{\pi_1}) \leftrightarrow \F(\proposition_{\pi_1}  \wedge \X\#_{\pi_1})$ then ensures that the additional letter is the label of the vertex~$v$ induced by the interpretation of $\pi_1$: the left-hand side of the equivalence \myquot{looks up} the truth values of the label of $v$ in $\pi_\vertices$ and the right-hand side the truth values at the additional letter in $\pi_1$, which comes right before the $\#$.

Thus, the least fixed point induced by $\psi_\con$ contains exactly the encodings of the prefixes of traces of $\tsys$ as the prefix of length one is included and $\psi_\step$ extends the prefixes by one more letter.

Thus, it that remains is to relativize the quantifiers of $\phi$ (the sentence we started with) so that they range over the traces of the transition system~$\tsys$ we have quantified using $\pi_\vertices$ and $\pi_\edges$.
Recall that the formula~$\varphi'$ we are constructed is model-checked over a transition system~$\tsys'$ that contains all traces over $\ap'$.
Thus, quantification over $\univar$ in $\phi$ (which ranges over arbitrary traces over $\ap$) is in $\varphi'$ replaced by quantification over $\unidisvar$ (which ranges over arbitrary traces in $\ap' \supseteq \ap$ and the additional propositions are irrelevant here).
Further, quantification over $\unidisvar$ in $\phi$ (which ranges over traces of a finite transition system) is in $\varphi'$ replaced by quantification over $\unidisvar$, but restricted to traces whose prefixes are all in $Y_\prefixes$. 
Such a trace must be a trace of the transition system, as it is finite.
Note that trace quantification over second-order variables~$Y \notin \set{\unidisvar,\univar}$ does not have to not relativized, as the traces added to the fixed points interpreting $Y$ come from $\unidisvar$, $\univar$, or previously quantified second-order variables. 

So, the formula~$\rel(\phi)$ is defined by inductively replacing
\begin{itemize}
    \item each $\exists \pi \in \univar.\ \psi$ by $\exists \pi \in \unidisvar.\ \psi$, 
    
    \item each $\forall \pi \in \univar.\ \psi$ by $\forall \pi \in \unidisvar.\ \psi$, 
    
    \item each $\exists \pi \in \unidisvar.\ \psi$ by $\exists \pi \in \unidisvar.\ \psi_r \wedge \psi$, and 
    
    \item each $\forall \pi \in \unidisvar.\ \psi$ by $\forall \pi \in \unidisvar.\ \psi_r \rightarrow \psi$,
\end{itemize}
where $\psi_r$ expresses that the trace assigned to $\pi$ must be one of the transition system~$\tsys$:
\begin{align*}
\forall \pi' \in \unidisvar.\ {}&{} \neg\#_{\pi'} \wedge ((\neg\#_{\pi'})\U(\#_{\pi'}\wedge \X\G\neg \#_{\pi'})) \rightarrow \\
{}&{}\exists \pi_p\in Y_\prefixes.\ \F(\#_{\pi'} \wedge \#_{\pi_p}) \wedge (\bigwedge_{\proposition\in\ap} \proposition_{\pi} \leftrightarrow \proposition_{\pi_p})\U\#_{\pi_p}
\end{align*}
Here, we use the fact that if all prefixes of a trace are in the prefixes of the traces of $\tsys$, then the trace itself is also one of $\tsys$. Thus, we can require for every $n > 0$, 
there is a trace in $Y_\prefixes$ encoding a prefix of length~$n$ of a trace of $\tsys$ such that $\pi$ has that prefix. 

To finish the proof, we need to ensure that $Y_\arith$ does indeed contain the traces implementing addition and multiplication (see Proposition~\ref{prop_plustimesinhyperltl} and the proof of Theorem~\ref{thm_satcomplexity_lfp_ss} on Page~\pageref{pageref_plustimesaslfp} for details): $\tsys'$ is a model of 
\[
\phi' = \exists \pi_\add \in \unidisvar.\ \exists \pi_\mult \in \unidisvar.\ \psi_s \wedge \exists(X_\arith, \smallest, \phi_\plustimes').\ \phi''
\]
if and only if $\phi$ is satisfied by a finite transition system with a single initial vertex.
\end{proof}

Now, we prove the lower bound for $\lfpsohyltlfp$ finite-state satisfiability.

\begin{lem}
\label{lemma_fssat_soahard}
Truth in second-order arithmetic is polynomial-time reducible to finite-state satisfiability for $\univar$-free $\lfpsohyltlfp$ sentences.
\end{lem}

\begin{proof}
Given a sentence~$\phi$ of second-order arithmetic, we show how to construct an $\univar$-free $\lfpsohyltlfp$ sentence~$\phi'$ such that  $\natsstruct\models\phi$ if and only if $\phi'$ is finite-state satisfiable.
Intuitively, we write a formula that is finite-state satisfiable, but only by finite transition systems that contain all traces over $\set{\inprop}$.
As before, this allows us to mimic first- and second-order quantification by quantification of traces.
Together with the fact that addition and multiplication can be implemented in $\hyltl$, this suffices to obtain the desired construction.
Note that this only works because we are asking for a finite transition system to satisfy $\phi'$: general satisfiability for $\univar$-free $\lfpsohyltlfp$ sentences is simpler (see Theorem~\ref{thm_lfpsatcomplexity}).

Let $\ap = \set{\inprop,\#}$.
We begin by presenting a (satisfiable) $\lfpsohyltlfp$ sentence~$\phi_\prefs$ that ensures that the set of prefixes of the $\set{\inprop}$-projection of each of its models is equal to $(\pow{\set{\inprop}})^*$.
To this end, consider the conjunction~$\phi_\prefs$ of the following formulas:
\begin{itemize}
    \item $\forall \pi \in \unidisvar.\ \G(\#_\pi \rightarrow \X\G\neg \#_\pi))$, which expresses that each trace in the model has at most one position where the proposition~$\#$ holds.
    \item $\exists \pi \in \unidisvar.\  \#_{\pi}$, which expresses that each model contains a trace where $\#$ holds at the first position.
    \item $\forall \pi \in \unidisvar.\ \exists \pi' \in \unidisvar.\ (\F\#_\pi) \rightarrow (\inprop_\pi \leftrightarrow \inprop_{\pi'}) \U (\#_{\pi} \wedge \neg\inprop_{\pi'} \wedge \X\#_{\pi'} )$ expressing that for every trace~$t$ in the model with a $\#$, there is another trace~$t'$ in the model such that the $\set{\inprop}$-projection~$w$ of $t$ up to the $\#$ and the $\set{\inprop}$-projection~$w'$ of $t'$ up to the $\#$ satisfy $w' = w\emptyset$.
    \item $\forall \pi \in \unidisvar.\ \exists \pi' \in \unidisvar.\ (\F\#_\pi)\rightarrow (\inprop_\pi \leftrightarrow \inprop_{\pi'}) \U (\#_{\pi} \wedge \inprop_{\pi'} \wedge \X\#_{\pi'} )$ expressing that for every trace~$t$ in the model with a $\#$ there is another trace~$t'$ in the model such that the $\set{\inprop}$-projection~$w$ of $t$ up to the $\#$ and the $\set{\inprop}$-projection~$w'$ of $t'$ up to the $\#$ satisfy $w' = w\set{\inprop}$.   
\end{itemize}

A straightforward induction shows that the set of prefixes of the $\set{\inprop}$-projection of each model of $\phi_\prefs$ is equal to $(\pow{\set{\inprop}})^*$.
Furthermore, let $\tsys$ be a finite transition system that is a model of $\phi_\prefs$, i.e., with $\traces(\tsys) \models \phi_\prefs$. 
Then, the $\set{\inprop}$-projection of $\traces(\tsys)$ must be equal to $(\pow{\set{\inprop}})^\omega$, which follows from the fact that the set of traces of a finite transition system is closed (see, e.g.,~\cite{BK08} for the necessary definitions).
Thus, as usual, we can mimic set quantification over $\nats$ by trace quantification. 

Further, recall that we can implement addition and multiplication in $\lfpsohyltlfp$ (see Proposition~\ref{prop_plustimesinhyperltl} and the proof of Theorem~\ref{thm_satcomplexity_lfp_ss} on Page~\pageref{pageref_plustimesaslfp} for how to turn the formula~$\phi_\plustimes$ into a least fixed point guard): there is an $\lfpsohyltlfp$ guard for a second-order quantifier such that the $\set{\argone, \argtwo, \res, \add, \mult}$-projection of the unique least fixed point that satisfies the guard is equal to $T_\plustimes$.

Thus, we can mimic set quantification over $\nats$ and implement addition and multiplication, which allow us to reduce truth in second-order arithmetic to finite-state satisfiability for $\univar$-free sentences using the function~$\hyperize$ presented in the proof of Theorem~\ref{thm_satcomplexity}:
\[
\phi' = \phi_\prefs \wedge \exists \pi_\add \in \unidisvar.\ \exists \pi_\mult \in \unidisvar.\ \psi_s \wedge \exists(X_\arith, \smallest, \phi_\plustimes').\ \hyperize(\phi)
\]
is finite-state satisfiable if and only if $\natsstruct\models\phi$.
\end{proof}
Note that the formula $\phi'$ we have constructed in the proof just uses one second-order quantifier.

This result, i.e., that finite-state satisfiability for $\univar$-free $\lfpsohyltlfp$  sentences is at least as hard as truth in second-order arithmetic, should be contrasted with the general satisfiability problem for $\univar$-free $\lfpsohyltlfp$ sentences, which is \myquot{only} $\Sigma_1^1$-complete~\cite{sohypercomplexity}, i.e., much simpler. 
The reason is that every satisfiable $\lfpsohyltlfp$ sentence has a countable model (i.e., a countable set of traces). 
This is even true for the formula~$\phi_\prefs$ we have constructed. However, every \emph{finite-state} transition system that satisfies the formula must have uncountably many traces. 
This fact allows us to mimic quantification over arbitrary subsets of $\nats$, which is not possible in a countable model. 
Thus, the general satisfiability problem is simpler than the finite-state satisfiability problem.

Finally, we prove the upper bound for $\lfpsohyltlfp$ model-checking.

\begin{lem}
\label{lemma_mc_soaeasy}
$\lfpsohyltlfp$ model-checking is polynomial-time reducible to truth in second-order arithmetic. 
\end{lem}

\begin{proof}
As in the (upper bound) proof of Theorem~\ref{thm_satcomplexity} (and several times thereafter), we mimic trace quantification via quantification of sets of natural numbers and capture the temporal operators using addition. 
To capture the \myquot{quantification} of fixed points, we use again witness trees (see Subsection~\ref{subsec_lfpsatcw} for details). 
Recall that these depend on a variable assignment of the trace variables. 
Thus, in our construction, we need to explicitly handle such an assignment as well (encoded by a set of natural numbers) in order to be able to correctly apply witness trees.

Let $\tsys = (V,E,I,\lambda)$ be a finite transition system. 
We assume without loss of generality that $V = \set{0,1,\ldots, n}$ for some $n \ge 0$.
Recall that $\pair \colon \nats\times\nats \rightarrow\nats$ denotes Cantor's pairing function defined as $\pair(i,j) = \frac{1}{2}(i+j)(i+j+1) +j$, which is a bijection and implementable in arithmetic.
We encode a path~$\rho(0)\rho(1)\rho(2)\cdots$ through $\tsys$ by the set~$\set{\pair(j,\rho(j)) \mid j \in\nats} \subseteq \nats$.
Not every set encodes a path, but the first-order formula
\begin{align*}
\phi_\ispath(Y) = {}&{} \forall x.\ \forall y.\ y > n \rightarrow \pair(x,y) \notin Y \wedge \\
{}&{} \forall x.\ \forall y_0.\ \forall y_1.\ (\pair(x,y_0) \in Y \wedge
 \pair(x,y_1) \in Y) \rightarrow y_0=y_1  \wedge  \\
{}&{} \bigvee_{v \in I} \pair(0,v) \in Y \wedge \\
{}&{} \forall j.\ \bigvee_{(v,v') \in E} \pair(j,v) \in Y \wedge \pair(j+1,v') \in Y
\end{align*}
checks if a set does encode a path of $\tsys$.

Furthermore, fix some bijection~$e \colon \ap \rightarrow\set{0,1,\ldots,\size{\ap}-1}$.
As before, we encode a trace~$t \in (\pow{\ap})^\omega$ by the set~$S_t =\set{\pair(j,e(\proposition)) \mid j \in \nats \text{ and } \proposition \in t(j)} \subseteq \nats$.
While not every subset of $\nats$ encodes some trace~$t$, the first-order formula 
\[\phi_\istrace(Y) = \forall x.\ \forall y.\ y \ge \size{\ap} \rightarrow \pair(x,y) \notin Y \] checks if a set does encode a trace.

Finally, the first-order formula
\[
\phi_{\tsys}(Y) = \exists Y_p.\ \ispath(Y_p) \wedge \forall j.\ \bigwedge_{\proposition\in\ap} \left( \pair(j,e(\proposition))\in Y \leftrightarrow \bigvee_{v\colon \proposition \in \lambda(v)} \pair(j,v) \in Y_p\right)
\]
checks whether the set~$Y$ encodes the trace of some path through $\tsys$.

As in the proof of Theorem~\ref{thm_satcomplexity_lfp_ss}, we need to encode variable assignments (whose domain is restricted to trace variables) via sets of natural numbers. 
Using this encoding, one can \myquot{update} encoded assignments, i.e., there exists a formula~$\phi_{\update}^\pi(A,A',Z)$ that is satisfied if and only if
\begin{itemize}
    \item the set~$A$ encodes a variable assignment~$\Pi$,
    \item the set~$A'$ encodes a variable assignment~$\Pi'$,
    \item the set~$Z$ encodes a trace~$t$, and
    \item $\Pi[\pi\mapsto t] = \Pi'$.
\end{itemize}

Now, we inductively translate an $\lfpsohyltlfp$ sentence~$\phi$ into a formula~$\arithmetize_\tsys(\phi)$ of second-order arithmetic.
This formula has two free variables, one first-order one and one second-order one.
The former encodes the position at which the formula is evaluated while the latter one encodes a variable assignment (which, as explained above, is necessary to give context for the witness trees).  
We construct $\arithmetize_\tsys(\phi)$ such that 
$\tsys \models \phi$ if and only if $\natsstruct\models (\arithmetize_\tsys(\phi))(0,\emptyset)$, where the empty set encodes the empty variable assignment.

\begin{itemize}
    
    \item $\arithmetize_\tsys(\exists (Y,\smallest,\psi_j^\con).\ \psi) = \arithmetize_\tsys(\psi)$. Here, the free variables of $\arithmetize_\tsys(\exists (Y,\smallest,\psi_j^\con).\ \psi)$ are the free variables of $\arithmetize_\tsys(\psi)$. Thus, we ignore quantification over least fixed points, as we instead use witness trees to check membership in these fixed points.
    
    \item $\arithmetize_\tsys(\forall (Y,\smallest,\psi_j^\con).\ \psi) = \arithmetize_\tsys(\psi)$. Here, the free variables of $\arithmetize_\tsys(\forall (Y,\smallest,\psi_j^\con).\ \psi)$ are the free variables of $\arithmetize_\tsys(\psi)$. 
    
    \item $\arithmetize_\tsys(\exists\pi\in \univar.\ \psi) = \exists Z_\pi.\  \exists A'.\ \phi_{\istrace}(Z_\pi) \wedge \phi_\update^\pi(A,A',Z_\pi) \wedge \arithmetize_\tsys(\psi)$. 
    Here, the free second-order variable of $\arithmetize_\tsys(\exists\pi\in \univar.\ \psi)$ is $A$ while $A'$ is the free second-order variable of $\arithmetize_\tsys(\psi)$. The free first-order variable of $\arithmetize_\tsys(\exists\pi\in \univar.\ \psi)$ is the free first-order variable of $\arithmetize_\tsys(\psi)$.
    
    \item $\arithmetize_\tsys(\forall\pi\in \univar.\ \psi) = \forall Z_\pi.\ \forall A'.\ (\phi_{\istrace}(Z_\pi) \wedge \phi_\update^\pi(A,A',Z_\pi) )\rightarrow \arithmetize_\tsys(\psi)$. 
    Here, the free second-order variable of $\arithmetize_\tsys(\forall\pi\in \univar.\ \psi)$ is $A$ while $A'$ is the free second-order variable of $\arithmetize_\tsys(\psi)$. The free first-order variable of $\arithmetize_\tsys(\forall\pi\in \univar.\ \psi)$ is the free first-order variable of $\arithmetize_\tsys(\psi)$.

    \item $\arithmetize_\tsys(\exists\pi\in \unidisvar.\ \psi) = \exists Z_\pi.\  \exists A'.\ \phi_\tsys(Z_\pi) \wedge \phi_\update^\pi(A,A',Z_\pi) \wedge \arithmetize_\tsys(\psi)$. 
    Here, the free second-order variable of $\arithmetize_\tsys(\exists\pi\in \unidisvar.\ \psi)$ is $A$ while $A'$ is the free second-order variable of $\arithmetize_\tsys(\psi)$. The free first-order variable of $\arithmetize_\tsys(\exists\pi\in \unidisvar.\ \psi)$ is the free first-order variable of $\arithmetize_\tsys(\psi)$.
    
    \item $\arithmetize_\tsys(\forall\pi\in \unidisvar.\ \psi) = \forall Z_\pi.\ \forall A'.\ (\phi_\tsys(Z_\pi) \wedge \phi_\update^\pi(A,A',Z_\pi) )\rightarrow \arithmetize_\tsys(\psi)$. 
    Here, the free second-order variable of $\arithmetize_\tsys(\forall\pi\in \unidisvar.\ \psi)$ is $A$ while $A'$ is the free second-order variable of $\arithmetize_\tsys(\psi)$. The free first-order variable of $\arithmetize_\tsys(\forall\pi\in \unidisvar.\ \psi)$ is the free first-order variable of $\arithmetize_\tsys(\psi)$.

    \item $\arithmetize_\tsys(\exists\pi\in Y_j.\ \psi) = \exists Z_\pi.\ \exists A'.\ \phi_\hastree(A,Z_\pi,j) \wedge \phi_\update^\pi(A,A',Z_\pi) \wedge \arithmetize_\tsys(\psi)$, where $\phi_\hastree(A,Z_\pi,j)$ is a formula of second-order arithmetic that captures the existence of a witness tree for the trace being encoded by $Z_\pi$ being in the fixed point assigned to $Y_j$ w.r.t.\ the variable assignment encoded by $A$. Its construction is similar to the corresponding formula in the proof of Theorem~\ref{thm_satcomplexity_lfp_ss}, but we replace the free third-order variable~$X_D$ encoding the model there by hardcoding the set of traces of $\tsys$ using the formula~$\phi_\tsys$ from above.
   
    Here, the free second-order variable of $\arithmetize_\tsys(\exists\pi\in Y_j.\ \psi)$ is $A$ while $A'$ is the free second-order variable of $\arithmetize_\tsys(\psi)$. The free first-order variable of $\arithmetize_\tsys(\exists\pi\in Y_j.\ \psi)$ is the free first-order variable of $\arithmetize_\tsys(\psi)$.

    \item $\arithmetize_\tsys(\forall\pi\in Y_j.\ \psi) = \forall Z_\pi.\ \forall A'.\ (\phi_\hastree(A,Z_\pi,j) \wedge \phi_\update^\pi(A,A',Z_\pi) )\rightarrow \arithmetize_\tsys(\psi)$.
    Here, the free second-order variable of $\arithmetize_\tsys(\forall\pi\in Y_j.\ \psi)$ is $A$ while $A'$ is the free second-order variable of $\arithmetize_\tsys(\psi)$. The free first-order variable of $\arithmetize_\tsys(\forall\pi\in Y_j.\ \psi)$ is the free first-order variable of $\arithmetize_\tsys(\psi)$.

    \item $\arithmetize_\tsys(\psi_1 \vee \psi_2) = \arithmetize_\tsys(\psi_1) \vee \arithmetize_\tsys(\psi_2)$. Here, we require that the free variables of $\arithmetize_\tsys(\psi_1)$ and $\arithmetize_\tsys(\psi_2)$ are the same (which can always be achieved by variable renaming), which are then also the free variables of $\arithmetize_\tsys(\psi_1 \vee \psi_2)$. 
    
    \item $\arithmetize_\tsys(\neg\psi) = \neg\arithmetize_\tsys(\psi)$. Here, the free variables of $\arithmetize_\tsys(\neg\psi) $ are the free variables of $ \arithmetize_\tsys(\psi)$.
    
     \item $\arithmetize_\tsys(\X\psi) =  \exists i' (i' = i+1) \wedge \arithmetize_\tsys(\psi)$, where $i'$ is the free first-order variable of $\arithmetize_\tsys(\psi)$ and $i$ is the free first-order variable of $\arithmetize_\tsys(\X\psi)$.
     The free second-order variable of $\arithmetize_\tsys(\X\psi)$ is equal to the free second-order variable of $\arithmetize_\tsys(\psi)$.
    
    \item $\arithmetize_\tsys(\psi_1\U\psi_2) =  \exists i_2.\ i_2 \ge i \wedge \arithmetize_\tsys(\psi_2) \wedge \forall i_1.\ (i \le i_1 \wedge i_1 < i_2) \rightarrow \arithmetize_\tsys(\psi_1)$, where $i_j$ is the free first-order variable of $\arithmetize_\tsys(\psi_j)$, 
    and $i$ is the free first-order variable of $\arithmetize_\tsys(\psi_1\U\psi_2)$.
    Furthermore, we require that the free second-order variables of the $\arithmetize_\tsys(\psi_j)$ are the same, which is then also the free second-order variable of $\arithmetize_\tsys(\psi_1\U\psi_2)$. Again, this can be achieved by renaming variables.

    \item $\arithmetize_\tsys(\proposition_\pi) = \pair(i,e(\proposition)) \in Z_\pi$, i.e., $i$ is the free first-order variable of $\arithmetize_\tsys(\proposition_\pi)$. Note that this formula does not have a free second-order variable. For completeness, we can select an arbitrary one to serve that purpose.
    
\end{itemize}
Now, an induction shows that $\tsys\models \phi$ if and only if $\natsstruct$ satisfies $(\arithmetize(\phi))(0,\emptyset)$, where $\emptyset$ again encodes the empty variable assignment.
\end{proof}


\section{Related Work}
\label{sec_relatedwork}
As mentioned in Section~\ref{sec:intro}, the complexity problems for HyperLTL were thoroughly studied~\cite{FinkbeinerH16,hyperltlsatconf,hyperltlsat}. For $\sohyltl$, Beutner et al.\ mainly focused on the algorithmic aspects by providing model checking~\cite{DBLP:conf/cav/BeutnerFFM23} and monitoring~\cite{BeutnerFFM24} algorithms.
Finkbeiner et al. studied the complexity of $\sohyltl$ and $\lfpsohyltlfp$ over tree-shaped and acyclic structures, which contain finitely many traces~\cite{FoSSaCS26}, thus all problems become decidable.
None of the above studied the respective complexity problems for general structures in depth. 

Logics related to $\sohyltl$ are asynchronous and epistemic logics. Much research has been done regarding epistemic properties~\cite{DBLP:conf/clima/Dima08,DBLP:books/mit/FHMV1995,DBLP:conf/atal/LomuscioR06a,DBLP:conf/fsttcs/MeydenS99} and their relations to hyperproperties~\cite{DBLP:conf/fossacs/BozzelliMP15}. However, most of this work concerns expressiveness and decidability results (e.g.,~\cite{DBLP:journals/tocl/BozzelliMM24}), and not complexity analysis for the undecidable fragments. This is similar for asynchronous hyperlogics~\cite{BartocciHNC23,DBLP:conf/cav/BaumeisterCBFS21,BeutnerF23,DBLP:journals/corr/abs-2404-16778,DBLP:conf/lics/BozzelliPS21,DBLP:conf/concur/BozzelliPS22,DBLP:journals/pacmpl/GutsfeldMO21,DBLP:conf/mfcs/KontinenSV23,KontinenSV24,DBLP:conf/mfcs/KrebsMV018}, where most work concerns decidability results and expressive power, but not complexity analysis. 

Another related logic is $\teamltl$~\cite{DBLP:conf/mfcs/KrebsMV018}, a hyperlogic for the specification of dependence and independence. Lück~\cite{Luck20} studied similar problems to those we study in this paper and showed that, in general, satisfiability and model checking of $\teamltl$ with Boolean negation is equivalent to truth in third-order arithmetic. Kontinen and Sandström~\cite{DBLP:conf/wollic/KontinenS21} generalize this result and show that any logic between $\teamltl$ with Boolean negation and second-order logic inherits the same complexity results. 
Kontinen et al.~\cite{DBLP:conf/mfcs/KontinenSV23} study set semantics for asynchronous $\teamltl$, and provide positive complexity and decidability results. Gutsfeld et al.~\cite{GutsfeldMOV22} study an extension of $\teamltl$ to express refined notions of asynchronicity and analyze the expressiveness and complexity of their logic, proving it also highly undecidable. 
While $\teamltl$ is closely related to $\sohyltl$, the exact relation between them is still unknown.

Finally, the logic~$\hyqptl$, which extends $\hyltl$ by quantification over propositions~\cite{Rabe16diss}, is also related to $\sohyltl$. 
With uniform quantification, $\hyqptl$ satisfiability is equivalent to truth in second-order arithmetic~\cite{regaud2024complexityhyperqptl} while finite-state satisfiability and model-checking have the same complexity as for $\hyltl$~\cite{FinkbeinerH16,Rabe16diss}. Non-uniform quantification makes $\hyqptl$ as expressive as second-order $\hyltl$~\cite{regaud2024complexityhyperqptl}, which implies that all three problems are equivalent to truth in third-order arithmetic.

\section{Conclusion}
\label{sec_conc}
We have investigated and settled the complexity of satisfiability, finite-state satisfiability, and model-checking for $\sohyltl$ and its fragments~$\sohyltlfp$ and $\lfpsohyltlfp$.
For the former two, all three problems are equivalent to truth in third-order arithmetic, and therefore (not surprisingly) much harder than the corresponding problems for $\hyltl$, which are \myquot{only} $\Sigma_1^1$-complete, $\Sigma_1^0$-complete, and \tower-complete, respectively.
This shows that the addition of second-order quantification increases the already high complexity of $\hyltl$ significantly. 
However, for the fragment $\lfpsohyltlfp$, in which second-order quantification degenerates to least fixed point computations, the complexity is lower (albeit still high): satisfiability is $\Sigma_1^2$-complete while finite-state satisfiability and model-checking are equivalent to truth in second-order arithmetic.
All our results, but the one for $\lfpsohyltlfp$ satisfiability, also hold for closed-world semantics. 
On the other hand, $\lfpsohyltlfp$ satisfiability under closed-world semantics is \myquot{only} $\Sigma_1^1$-complete, i.e., as hard as $\hyltl$ satisfiability. 

In future work, we aim to study the rich landscape of temporal logics for asynchronous hyperproperties. 
Here, the first results~\cite{rz25} for generalized $\hyltl$ with stuttering and contexts~\cite{DBLP:journals/corr/abs-2404-16778} have been obtained: 
Model-checking is equivalent to truth in second-order arithmetic while satisfiability is $\Sigma_1^1$-complete.

\paragraph*{Acknowledgments}

\noindent 
\begin{wrapfigure}[2]{l}{3.2cm}
\vspace{-.4cm}
\includegraphics[scale=.1]{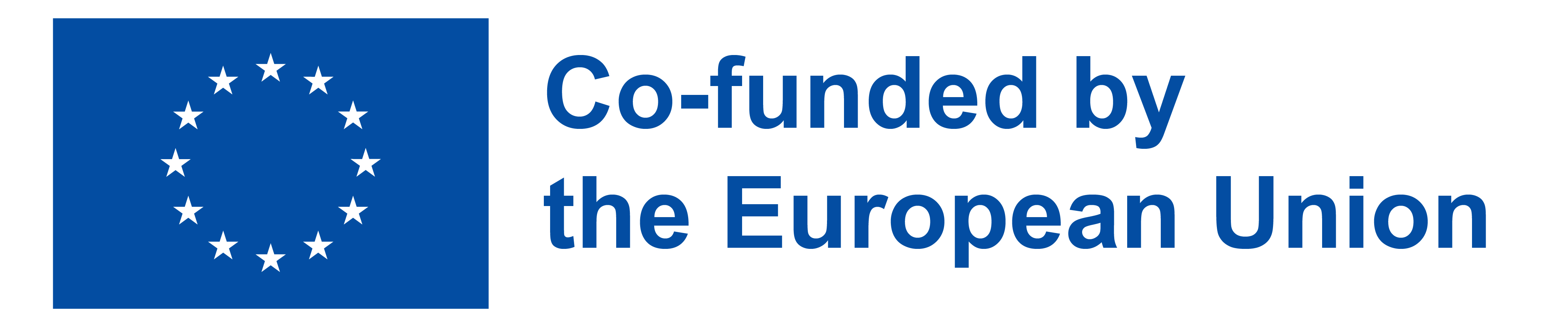}  
\end{wrapfigure}
Gaëtan Regaud has been supported by the European Union. Martin Zimmermann has been supported by DIREC – Digital Research Centre Denmark. 
Hadar Frenkel has been supported in part by the Israel Science Foundation (grant No.\ 655/25).

This work was initiated by a discussion at Dagstuhl~Seminar 23391 \myquot{The Futures of Reactive Synthesis} and some results were obtained at Dagstuhl Seminar~24111 \myquot{Logics for Dependence and Independence: Expressivity and Complexity}.

\bibliographystyle{alphaurl}
\bibliography{bib}

\end{document}